\documentclass[a4paper,twoside,envcountsect,envcountsame
]{llncs}
\usepackage[english]{babel}
\usepackage[utf8]{inputenc}
\usepackage[hang,centerlast]{subfigure}
\usepackage{amsmath,amssymb,multirow,ifthen,enumitem,xspace,mathpartir,datetime,manfnt,stmaryrd}
\usepackage[clock,geometry]{ifsym}

\usepackage[algo2e,ruled,lined]{algorithm2e}

\usepackage{tikz}
\usetikzlibrary{automata,patterns,topaths,shapes,calc,arrows,decorations.pathreplacing,positioning,intersections}

\tikzstyle{every picture}+=[>=stealth,initial text=]
\tikzstyle{accepting}=[accepting by arrow]
\colorlet{irrelevant}{black!60}
\tikzstyle{interval duration}=[draw,dotted,thick,<->]
\tikzstyle{interval length}=[draw,>=angle 90,>-<]

\newif\ifDraft
\Draftfalse
\usepackage[layout={margin,index}]{fixme}
\makeatletter
\newcounter{fixcount}
\newcommand{\mini@dbend}{\raise -1.5pt \hbox{\Large$\triangle$}\hspace{-2.05ex}\raise -.5pt \hbox{!}\hspace{1.25ex}}
\newcommand{\defineNote}[3][black!65!green]{%
  \expandafter\def\csname @#2\endcsname ##1{\stepcounter{fixcount}\fxwarning{\textcolor{#1}{\textbf{#3}: ##1}}}%
  \expandafter\def\csname @@#2\endcsname ##1{\stepcounter{fixcount}\small\fxwarning[author={\textcolor{red}{\mini@dbend #3,}},margin=false,inline=true]{\textcolor{#1}{##1}}}%
  \expandafter\def\csname #2\endcsname{\@ifstar{\csname @@#2\endcsname}{\csname @#2\endcsname}}
}

\pgf@draftmodefalse
  \makeatother

\ifDraft
\fi

\title{Polynomial Interrupt Timed Automata}
\ifDraft\date\today\fi
\author{B\'eatrice B\'erard$^1$ \and Serge Haddad$^2$ \and Claudine Picaronny$^2$ \and Mohab Safey El Din$^1$ \and Mathieu Sassolas$^3$}
\institute{Sorbonne Université, Université P. \& M. Curie, LIP6, CNRS UMR 7606, Paris, France \and
\'Ecole Normale Sup\'erieure de Cachan, LSV, CNRS UMR 8643, INRIA, Cachan, France \and
Université Paris-Est, LACL, Cr\'eteil, France}

\input{macros}

\defineNote{BB}{B\'eatrice}
\defineNote{MaS}{Mathieu}
\defineNote{SH}{Serge}
\defineNote{MoS}{Mohab}
\defineNote{CP}{Claudine}

\begin{document}
\maketitle
\ifDraft\centerline{\today}\fi
\begin{abstract}
  Interrupt Timed Automata (\ita) form a subclass of stopwatch
  automata where reachability and some variants of timed model
  checking are decidable even in presence of parameters. They are well
  suited to model and analyze real-time operating systems.  Here we
  extend \ita with polynomial guards and updates, leading to the class
  of polynomial ITA (\polita).  We prove the decidability of the
  reachability and model checking of a timed version of \CTL by an
  adaptation of the \emph{cylindrical decomposition} method for the
  first-order theory of reals. Compared to previous approaches, our
  procedure handles parameters and clocks in a unified way. Moreover,
  we show that \polita are incomparable with stopwatch automata.
  Finally additional features are introduced while preserving
  decidability.
\end{abstract}
\section{Introduction}

\paragraph{Hybrid Automata.} Hybrid systems~\cite{hybrid93}
combine continuous evolution of variables according to flow functions
(described by differential inclusions) in control nodes, and discrete
jumps between these nodes, where the variables can be tested by guards
and updated.  This class of models is very expressive and all relevant
verification questions (\emph{e.g.} reachability) are undecidable.
For the last twenty years, a large amount of research was devoted to
identifying subclasses with decidable properties, by restricting the
continuous dynamics and/or the discrete behavior of the systems. Among
these classes lie the well known Timed Automata (TA)~\cite{alur94a},
where all variables are \emph{clocks} (with derivative $\dot{x}=1$),
guards are comparisons of clocks with rational constants, and updates
are resets. It is proved in~\cite{HenzingerKPV98} that reachability
becomes undecidable when adding one stopwatch (with $\dot{x}=1$ or
$\dot{x}=0$) to timed automata.  Decidability results were also
obtained for larger classes (see
\cite{asarin95,alur95,HenzingerKPV98,LafferrierePS00,AlurHLP00}),
usually by building from the associated transition system (with
uncountable state space) a finite abstraction preserving a specific
class of properties, like reachability or those expressed by temporal
logic formulas. In all these abstractions, a state is a pair composed
of a control node and a polyhedron of variable values.
Examples of such classes include initialized rectangular
automata~\cite{HenzingerKPV98} where $\dot{x} \in [a,b]$ or o-minimal
hybrid systems~\cite{LafferrierePS00} where the flow is more general,
for instance of the form $\dot{x}= Ax$ over $\R^n$ for some matrix
$A$. In both cases, the variables must be (possibly non
deterministically) reinitialized at discrete jumps.

\paragraph{Interrupt Timed  Automata.} The class of Interrupt
Timed Automata (ITA), incomparable with TA, was introduced
in~\cite{BH-Fossacs09,BHS-fmsd2012} as another subclass of hybrid
automata with a (time-abstract) bisimulation providing a finite
quotient, thus leading to decidability of reachability and some
variants of timed model checking.  In a basic $n$-dimensional ITA,
control nodes are organized along $n$ levels, with $n$ stopwatches
(also called clocks hereafter), one per level. At a given level, the
associated clock is active, while clocks from lower levels are frozen
and clocks from higher levels are irrelevant. Guards are linear
constraints and the clocks can be updated by linear expressions (using
only clocks from lower levels).  The particular hierarchical structure
of ITA makes them particularly well suited for modeling systems with
interruptions, like real-time operating systems.  \ita were extended
with parameters in~\cite{BHJL-rp13} while preserving decidability by
combining the finite abstraction of original \ita with a finite
partition of parameter values.

\paragraph{Contribution.} We define the class \polita, of polynomial
ITA, where linear expressions on clocks are replaced by polynomials
with rational coefficients both for guards and updates. For instance, a
guard at level $2$ with clock $x_2$ can be of the form $P_1(x_1)x_2^2
+ P_2(x_1) \geq 0$, where $P_1$ and $P_2$ are polynomials with single
variable $x_1$, the clock of level $1$. Thus, guards are more
expressive than in the whole class of linear hybrid automata. 
Such guards can be useful for instance if some objects are
produced at given levels, and operations on higher levels on these
objects require polynomial-time computations w.r.t. 
the size of these objects.  In addition, such
guards can simulate irrational (algebraic) constraints, a case that becomes
undecidable in the setting of timed automata~\cite{miller00}.

We establish that model checking of a timed extension of \CTL (which
contains reachability) is decidable in 2EXPTIME for \polita by adapting the
cylindrical decomposition~\cite{Collins75,BPR} related to the first
order theory of reals. This decomposition produces a finite partition
of the state space, which is the basis for the construction of a
finite bisimulation quotient.  The first order theory of reals has
already been used in several works on hybrid
automata~\cite{LafferrierePS00,AlurHLP00} but it was restricted to the
dynamical part, with discrete jumps that must reinitialize the
variables (like in o-minimal hybrid systems).  Our adaptation consists
in an on-the-fly construction avoiding in the favorable cases to build
the whole decomposition.

From an expressiveness point of view, we show that (contrary to \ita)
\polita are incomparable with stopwatch automata (\swa).  Finally, we
prove that the decidability result still holds with several
extensions: adding auxiliary clocks and parameters, and enriching the
possible updates. In particular, parametric \ita~\cite{BHJL-rp13} can
be seen as a subclass of \polita, and the complexity of our
reachability algorithm is better than~\cite{BHJL-rp13} (2EXPSPACE).

\paragraph{Outline.}
We describe the model of polynomial ITA in
Section~\ref{sec:definition}, with an example and the presentation of
the model checking problem. In Section~\ref{sec:cyldec}, we revisit and
adapt in this context the cylindrical decomposition for the first
theory of reals, with a special focus on the related algorithmic
questions. The decision procedure for the model checking problem in
\polita is then presented in Section~\ref{sec:reach}, with an example
of the construction. Finally, we describe several extensions in
Section~\ref{sec:extensions} and conclude in Section~\ref{sec:conc}.

\section{Polynomial ITA}
\label{sec:definition}

\subsection{Definition}

Let $\N$ denote the set of natural numbers, $\Z$ the set of integers,
$\Q$ the set of rationals, and $\R$ the set of real numbers, with
$\R_{\geq 0}$ the set of non negative real numbers.

Let $X=\{x_1,\dots,x_n\}$ be a finite set of $n$ variables called
clocks. We write $\Q[x_1,\dots,x_n]$ for the set of polynomials with
$n$ variables and rational coefficients. A \emph{polynomial
  constraint} is a conjunction of constraints of the form $P \rel 0$
where $P \in \Q[x_1, \ldots, x_n]$ and $\rel \in \{<,\leq,=,\geq,>\}$,
and we denote by $\Cs(X)$ the set of polynomial constraints. We also
define $\U(X)$, the set of \emph{polynomial
updates} over $X$ as:
\[\U(X) = \left\{\bigwedge_{x\in X} x:= P_x \mmid \forall x,\, P_x \in
  \Q[x_1, \ldots,x_n]\right\}.\]\MoS{D\'esol\'e si la question est
  naive : je ne comprends pas pourquoi x est en indice}
  \MaS{Il manquait un quantificateur; est-ce mieux?}\SH{c'est pour indiquer
  que la m.a.j. dépend de $x$}


A valuation for $X$ is a mapping $v \in \R^X$, sometimes also
identified to the vector $(v(x_1), \ldots, v(x_n)) \in \R^n$. The
valuation where $v(x)=0$ for all $x \in X$ is denoted by
$\vect{0}$. For $P \in \Q[x_1, \ldots, x_n]$ and $v$ a valuation, the
value of $P$ at $v$ is $P(v)=P(v(x_1),\dots,v(x_n))$. A valuation $v$
satisfies the constraint $P \rel 0$, written $v \models P \rel 0$, if
$P(v) \rel 0$. The notation is naturally extended to a polynomial
constraint: $v \models \fee$ with $\fee = \bigwedge_i P_i \rel_i 0$ if
$v \models P_i \rel_i 0$ for every $i$.

An update of valuation $v$ by $u=\wedge_{x\in X} x:= P_x \in\U(X)$ is
the valuation $v[u]$ defined by $v[u](x) = P_x(v)$ for every $x \in
X$.  Hence an update is atomic in the sense that all variables are set
at the same time: the new value of variables depend on the old values
of $v$.

For a valuation $v$ and a delay $d \in \R_{\geq0}$, the valuation
$v'=v +_k d$, corresponding to \emph{time elapsing for clock $x_k$},
is defined by $v'(x_k)=v(x_k)+ d$ and $v'(x)=v(x)$ for any other clock
$x$.

\begin{definition}[PolITA]\label{def:polita}
  A \emph{polynomial interrupt timed automaton (\polita)} is a tuple
  $\A=\langle\Sigma,Q, q_0, F, X, \lambda, \Delta\rangle$,
  where:
\begin{itemize}
\item $\Sigma$ is a finite alphabet,
	\item $Q$ is a finite set of states, $q_0$ is the initial
        state, $F \subseteq Q$ is the set of final states,
	\item $X=\{x_1, \ldots, x_n\}$ consists of $n$ interrupt clocks,
	\item the mapping $\lambda : Q \rightarrow \{1, \ldots, n\}$
          associates with each state its level and $x_{\lambda(q)}$ is
          called the \emph{active clock} in state $q$.
	\item $\Delta \subseteq Q \times \Cs(X) \times (\Sigma \cup
          \{\eps\}) \times \U(X) \times Q$ is the set of transitions.
          Let $q \tr{\fee, a, u} q'$ in $\Delta$ be a transition with
          $k=\lambda(q)$ and $k'=\lambda(q')$. The guard $\fee$ is a
          conjunction of constraints $P \rel 0$ with $P \in
          \Q[x_1,\dots,x_{k}]$ ($P$ is a polynomial over clocks from
          levels less than or equal to $k$). The update $u$ is of the
          form $\wedge_{i=1}^{n} x_i := C_i$ with:
	\begin{itemize}
        \item if $k > k'$, \textit{i.e.} the transition decreases the
          level, then for $1 \leq i \leq k'$, $C_i=x_i$ and for $i >
          k'$, $C_i=0$;
        \item if $k \leq k'$ then for $1 \leq i < k$, $C_i=x_i$,
          $C_k=P$ for $P \in \Q[x_1,\dots,x_{k-1}]$ or $C_k=x_k$, and
          for $i > k$, $C_i=0$.
	\end{itemize}
\end{itemize}
\end{definition}

Remark that although it is possible to compare an active clock in a
non-polynomial way, \textit{e.g.} $x_2 > \sqrt{x_1}$ (which can be
translated as $x_2^2 > x_1 \wedge x_1 \geq 0$), it cannot be updated
in such a fashion.

\begin{example}
  \polita $\A_0$ of \figurename~\ref{fig:exPolIta} has two levels,
  with $q_0$ at level $1$ and $q_1$ and $q_2$ at level $2$, with $q_2$
  the single final state.  At level $1$, only $x_1$ appears in guards
  and updates (here the only update is the resetting of $x_1$ by
  action $a'$), while at level $2$ guards use polynomials in both
  $x_1$ and $x_2$.
\end{example}

\begin{figure}
\centering
\begin{tikzpicture}[auto]
\node[state,initial] (q1) at (0,-1.5) {$q_0,1$};
\node[state] (q2) at (3,0) {$q_1,2$};
\node[state] (q3) at (7,0) {$q_2,2$};

\path[->] (q1) edge node {$x_1^2 \leq x_1 + 1$, $a$} (q2);
\path[->] (q1) edge[loop right] node {$x_1^2 > x_1 + 1$, $a'$, $x_1:=0$} (q1);
\path[->] (q2) edge[bend left] node {$(2x_1-1)x_2^2 > 1$, $b$} (q3);
\path[->] (q3) edge[bend left] node {$x_2 \leq 5 - x_1^2$, $c$} (q2);
\end{tikzpicture}
\caption{A sample \polita $\A_0$.}
\label{fig:exPolIta}
\end{figure}

A \emph{configuration} $(q,v)$ 
consists of a state $q$ of $\A$ and a clock valuation $v$.
\begin{definition}\label{def:semantics}
  The semantics of a \polita $\A$ is defined by the (timed) transition
  system $\T_{\A}= (S, s_0, \rightarrow)$, where $S=\left\{(q,v) \mid
    q \in Q, \ v \in \R^X\right\}$ is the set of configurations, with
  initial configuration $s_0=(q_0, \vect{0})$. The relation
  $\rightarrow$ on $S$ consists of two types of steps:
\begin{description}[font=\em]
\item[Time steps:] Only the active clock in a state can evolve, all
  other clocks are frozen.  For a state $q$ with active clock
  $x_{\lambda(q)}$, a time step of duration $d \in \R_{\geq0}$ is
  defined by $(q,v) \tr{d} (q, v')$ with $v' = v +_{\lambda(q)} d$.  A
  time step of duration $0$ leaves the system $\T_{\A}$ in the same
  configuration.
\item[Discrete steps:] There is a discrete step $(q, v) \tr{a} (q',
  v')$ whenever there exists a transition $q \tr {\fee, a, u} q'$ in
  $\Delta$ such that $v \models \fee$ and $v' = v[u]$.
\end{description}
\end{definition}

An run of a \polita $\A$ is a path in $\T_\A$.  The \emph{trace} of a
run is the sequence of letters (or word) appearing in the path.  The
\emph{timed word} is the sequence of letters along with the absolute
time of the occurrence, \emph{i.e.} the sum of all delays appearing
before the letter.  Given a subset $F\subseteq Q$ of final states, a
run is accepting if it ends in a state of $F$.  This defines the
\emph{language} (resp. \emph{timed language}) as the set of traces
(resp. timed words) of accepting runs.

\begin{example}
  The \polita $\A_0$ can only take the transition from $q_0$ to $q_1$
  before $x_1$ reaches $\frac{1+\sqrt5}2$, \emph{i.e.} at the point
  where the red curve crosses the $x_1$ axis on
  \figurename~\ref{fig:exPolItaTrajectory}.  Then, transition $b$ from
  $q_1$ to $q_2$ can only be taken once $x_2$ reaches the grey areas.
  Transition $c$ cannot however be taken once the green curve has been
  crossed.  Hence the loop $bc$ can be taken as long as the clocks
  remain in the dark gray zone.  In the sequel, we show how to
  symbolically compute these zones.  Since $q_2$ is a final state, the
  run depicted in \figurename~\ref{fig:exPolItaTrajectory} is accepted
  by $\A$.  The associated timed word is
  $(a,1.2)(b,2.3)(c,2.6)(b,3.3)(c,3.9)(b,5.1)$, and the trace is the
  word $abcbcb$.
\end{example}

\begin{figure}
\centering
\begin{tikzpicture}
\newlength\epslength
\setlength\epslength{0.1cm}
\draw[->] (-\epslength,0) -- (5.25,0) node[anchor=west] {$x_1$};
\draw[->] (0,-1.75) -- (0,5.25) node [anchor=south] {$x_2$};

\path[fill=black!7] (0.522376,4.727123) plot[domain=0.52:5,smooth,tension=0.2] (\x,{1/sqrt(2*\x-1)}) -- (5,5) -- cycle;
\path[fill=black!20] (0.522376,4.727123) -- (0.7,1.581138) --plot[domain=0.7:1.61803398875,smooth,tension=0.2] (\x,{1/sqrt(2*\x-1)}) -- (1.61803398875,2.38196601125)
                                       -- plot[domain=1.61803398875:0.522376] (\x,5-\x*\x);

\draw[draw=blue,thick,smooth,tension=0.2,domain=0.52:5] plot (\x,{1/sqrt(2*\x-1)}) node[draw=blue,outer xsep=2pt,outer ysep=7pt,anchor=south] {$(2x_1-1)x_2^2-1=0$};
\draw[draw=blue,thick,smooth,domain=0.65:5] plot (\x,{-1/sqrt(2*\x-1)});
\draw[draw=green!80!black,thick,domain=-0.1:2.6] plot (\x,5-\x*\x) node[draw=green!80!black,outer sep=2pt,anchor=south west] {$x_2+x_1^2-5=0$};
\draw[draw=red!80!black,thick] (1.61803398875,-1.75) -- ++(0,6.75) node[draw=red!80!black,outer sep=2pt,anchor=north west] {$x_1^2-x_1-1=0$};

\tikzstyle{action}=[circle,fill=black,minimum width=4pt,inner sep=0pt]
\draw[very thick] (0,0) -- ++(1.2,0) node[action](up){}
	                      -- ++(0,1.1) node[action](a1){}
                        -- ++(0,0.3) node[action](a2){}
                        -- ++(0,0.7) node[action](a3){}
                        -- ++(0,0.6) node[action](a4){}
                        -- ++(0,1.2) node[action](a5){}
                        -- ++(0,1);
\node[anchor=north] at (up.south) {$a$};
\tikzstyle{actlabel}=[anchor=west,inner sep=0pt]
\node[actlabel] at (a1.east) {$b$};
\node[actlabel] at (a2.east) {$c$};
\node[actlabel] at (a3.east) {$b$};
\node[actlabel] at (a4.east) {$c$};
\node[actlabel] at (a5.east) {$b$};


\end{tikzpicture}
\caption{A trajectory of clocks of $\A_0$ in the 2-dimensional plane.}
\label{fig:exPolItaTrajectory}
\end{figure}

\subsection{Verification problems for \polita}

Given a \polita $\A$, natural questions arise regarding its behavior.
The most standard one is the \emph{reachability problem} which is the decision problem asking whether
a given state can be reached from the initial configuration.
This allows in particular to decide whether the timed language is nonempty, which is equivalent to testing the reachability of a final state.

More elaborate queries regarding the behavior of a \polita can be expressed through temporal logics like \CTL~\cite{emerson82,queille82} or timed extensions of such logics like \TCTL~\cite{alur93,HNSY94}.
Here we use a timed extension of \CTL which allows to reason over the values of clocks of the \polita.

Let $AP$ be a set of atomic 
propositions,
we equip the states of $\A$ with a labeling 
$lab: Q \rightarrow 2^{AP}$ of propositions 
that hold in the given state. For convenience, we assume
that $Q \subseteq AP$ with for all $q,q'\in Q$,
$q' \in lab(q)$ iff $q=q'$. 

\begin{definition} Formulas of the timed logic \TCTLint are defined by the
  following grammar:
\[
\psi ::= p \mid \psi \wedge \psi \mid \neg \psi \mid P \rel 0 \mid \always \psi \until \psi \mid \expath \psi \until \psi
\]
where $p \in AP$, $P$ is a polynomial of $\Q[x_1,\dots,x_n]$, and $\rel \,\in \{>,\geq,=,\leq,<\}$.
\end{definition}
We use the classical shorthands $\eventually p= \textit{true}\until p$, $\globally p=\neg\eventually\neg p$, and boolean operators.
The reachability problem  of a state $q$ is simply the satisfaction of $\expath\eventually q$.

The formulas of \TCTLint are interpreted over configurations of $\A$, hence the semantics of \TCTLint is defined as follows on the transition
system $\T_{\A}$ associated with $\A$.
Let  $\Run(s)$ denote all runs starting from configuration $s=(q,v)$.
For $\rho =(q,v) \xrightarrow{d_1} (q,v+_{\lambda(q)}d_1) \xrightarrow{a_1} (q_2,v_2) \cdots \in Run(s)$, a position in $\rho$ is a pair $\pi=(i,\delta)$ where $1 \leq i$ and $0 \leq \delta \leq d_i$.
The configuration corresponding to $\pi$ is $s_\pi=(q_i,v_i+_{\lambda(q_i)}\delta)$ (with $q_1=q$ and $v_1=v$).
We denote by $<_\rho$ the strict lexicographical order over positions of $\rho$.
\\
For basic formulas:
\[
\begin{array}{lcl}
  s \models p &\quad\textrm{iff}\ & p \in lab(s) \\
  s \models P \rel 0 &\quad\textrm{iff}\ &
  v \models P \rel 0
\end{array}
\]
and inductively:
\[
\begin{array}{lcl}
  s \models \varphi \wedge \psi &\quad\textrm{iff}\ & s \models \varphi\ \textrm{and}
  \ s \models \psi \\
  s \models \neg\varphi &\quad\textrm{iff}\ & s \not\models \varphi \\
  s \models \always \fee \until \psi &
\quad\textrm{iff}\ & \textrm{for all } \rho \in \Run(s), \ \rho \models \fee \until \psi\\
  s \models \expath \fee \until \psi &
\quad\textrm{iff}\ & \textrm{there exists } \rho \in \Run(s) 
\textrm{ s. t. } \rho \models \fee \until \psi\\
  \textrm{with } 
  \rho \models \fee \until \psi  &
\quad\textrm{iff}\ &\textrm{there is a position } \pi \in \rho 
\textrm{ s. t. } s_\pi \models \psi \\
& &\textrm{and } \forall \pi' <_{\rho} \pi, \ s_{\pi'} \models \fee \vee \psi.
\end{array}
\]

The automaton $\A$ satisfies $\psi$ (written $\A \models \psi$) if the initial configuration $s_0$
of $\T_{\A}$ satisfies $\psi$.
The \emph{model checking} problem asks, given $\A$ and $\psi$, whether $\A \models \psi$.


As mentioned in the introduction, an exhaustive traversal of the
(uncountable) transition system $\T_{\A}$ is not possible, and the model checking
algorithm relies on an abstraction of said transition system.
This abstraction needs to be refined enough to capture both time
elapsing and discrete jumps through the crossing of a transition.
Namely, two configurations in the same abstraction class should reach
the same successor classes when time elapses or when an update is
applied.
Moreover, the truth value of subformulas  $P \rel 0$ should be invariant in each abstraction class.

The previous works of~\cite{BH-Fossacs09,BHS-fmsd2012,BHJL-rp13} on
\ita built such an abstraction by relying on a set of
\emph{expressions} with rational coefficients.  These expressions
contained linear forms involved in guards and updates, along with the
active clock of the level.  Moreover, since the ordering of two
expressions at a given level could rely on the value of lower-level
clocks, some expressions were required at inferior levels.  The
classes were then defined as subsets of $\R^n$ where the ordering of
expressions was constant.

In the sequel, we adapt the above process in the context of \polita, where the
constraints are polynomial rather than linear, and hence yield regions
that are not polyhedra, but cells defined by a so called
\emph{cylindrical decomposition}.

\section{Cylindrical algebraic decomposition for first-order theory of
  reals}
\label{sec:cyldec}
\emph{Cylindrical algebraic decomposition} is introduced by
Collins in \cite{Collins75} for solving quantifier elimination
problems of first-order formulas over the reals. 
The first algorithm for solving this problem
was given by Tarski in \cite{Tarski48} but its complexity was non
elementary recursive. Cylindrical algebraic decomposition is doubly
exponential in the number of variables and is now a popular technique
for solving polynomial systems over the reals. Given a polynomial
family, it essentially partionates the ambient space into {\em cells}
which are homeomorphic to $]0,1[^i$ over which the input is
sign-invariant. These cells are also intrinsically arranged together
with a nice cylindrical structure which we explain further.
Later on, a procedure in EXPSPACE was
established~\cite{Ben-Or1984}.
The best lower bound currently known for this problem is
$STA(*,{2^n}^{O(1)},n)$ (a complexity class defined by machines with
limited alternations and located between EXPTIME and EXPSPACE) and it
already holds without the multiplication~\cite{Berman1980}.

We consider formulas that express properties of reals.
There are inductively defined as follows. 
An arithmetic expression is: 
\begin{itemize}
  \item either an integer constant, a variable;
  \item or $e_1+e_2$,  $e_1*e_2$
  where $e_1$ and $e_2$ are arithmetic expressions.
\end{itemize}
A formula is: 
\begin{itemize}
  \item a basic formula: $e \sim 0$
  where $\sim \in \{<,=\}$ and $e$ is an arithmetic expression;
  \item or $\varphi_1\wedge \varphi_2$,  $\varphi_1\vee \varphi_2$,
  $\neg \varphi_1$, $\forall x \varphi_1$, $\exists x \varphi_1$ 
  where $\varphi_1$ and $\varphi_2$ are formulas
  and $x$ is a variable.
\end{itemize}
A sentence is a formula without free variables. A sentence has a truth
value when interpreted over $\R$ and we are looking for deciding the
truth of a formula. 

For our purposes, we will adapt the cylindrical algebraic decomposition.
So we develop in the section all the required machinery.
Here we only describe the general principles and we 
explain how it can be used for deciding the truth of a formula.
The first concept that we introduce is the one of \emph{cell}.

\begin{definition}
A cell of level $n$ is a subset of $\R^n$ inductively defined as follows.
\begin{itemize}
 \item When $n=1$, it is either a point or an open interval.
 \item A cell $C$ of level $n+1$ is based on a cell $C'$
 of level $n$. It has one of the following shapes.
 \begin{enumerate}
 %
 %
  \item $C=\{(x,f(x))\mid x \in C'\}$ with $f$ a continuous
  function from $C'$ to $\R$;
 %
 %
  \item $C=\{(x,y)\mid x \in C'\wedge l(x)<y< u(x)\}$ with $l<u$ continuous
  functions from $C'$ to $\R$, possibly with $l=-\infty$ and/or $u=+\infty$.
 \end{enumerate}
\end{itemize}
\end{definition}
By convention the single cell of level 0 is $\R^0$.

Let $\mathcal P=\{\mathcal P_i\}_{1\leq i\leq n}$ be a family of
subsets of polynomials such that for all $P \in \mathcal P_i$, $P\in
\R[X_1,\ldots,X_i]$.  By convention, we extend $\mathcal P$ with
$\mathcal P_0=\emptyset$.  The second concept that we introduce is the
\emph{sign invariance} of a cell w.r.t. $\mathcal P$.

\begin{definition}
\label{def:cad}
Let $\mathcal P=\{\mathcal P_i\}_{i\leq n}$.
A cell $C$ of level $i$ is $\mathcal P$-invariant if:
\begin{itemize}
 \item For all $j\leq i$, for all $P \in \mathcal P_j$, for all $x,y \in C$
 $sign(P(x))=sign(P(y))$.
 \item When $i<n$,
 \begin{enumerate}
  \item either $C\times \R$ is $\mathcal P$-invariant;
  
  \item or there exists $f_1<\ldots<f_r$ continuous functions
  from $C$ to $\R$ such that all the following cells are $\mathcal P$-invariant:\\
\begin{itemize}
\item for all $1\leq i\leq r$,
  $\{(x,f_i(x))\mid x \in C\}$;
 %
 %
\item for all $0\leq i\leq r$, $\{(x,y)\mid x \in C\wedge f_{i}(x)<y<
  f_{i+1}(x)\}$ with the convention that $f_0=-\infty$ and
  $f_{r+1}=+\infty$.
\end{itemize}
 \end{enumerate}
\end{itemize}
\end{definition}

Observe that $\R^0$ is $\mathcal P$-invariant, and that one can
inductively define a tree of $\mathcal P$-invariant cells as follows.
\begin{itemize}
 \item The root of the tree is $\R^0$;
 \item Let $C$ be a $\mathcal P$-invariant cell of level $i<n$
 belonging to the tree. Then depending on the kind of invariance,
 \begin{enumerate}
  \item either $C$ has a single child $C\times \R$;
  \item or for some $r \in \N\setminus\{0\}$, $C$ has $2r+1$ ordered children
  $\{(x,y)\mid x \in C\wedge y< f_1(x)\}$, $\{(x,f_1(x))\mid x \in C\}$,
  $\{(x,y)\mid x \in C\wedge f_{1}(x)<y< f_{2}(x)\}$,
  $\ldots$ , $\{(x,y)\mid x \in C\wedge y> f_r(x)\}$.   
 \end{enumerate}
\end{itemize}
This tree is also called a \emph{cylindrical decomposition}.

\begin{figure}
\centering
\includegraphics{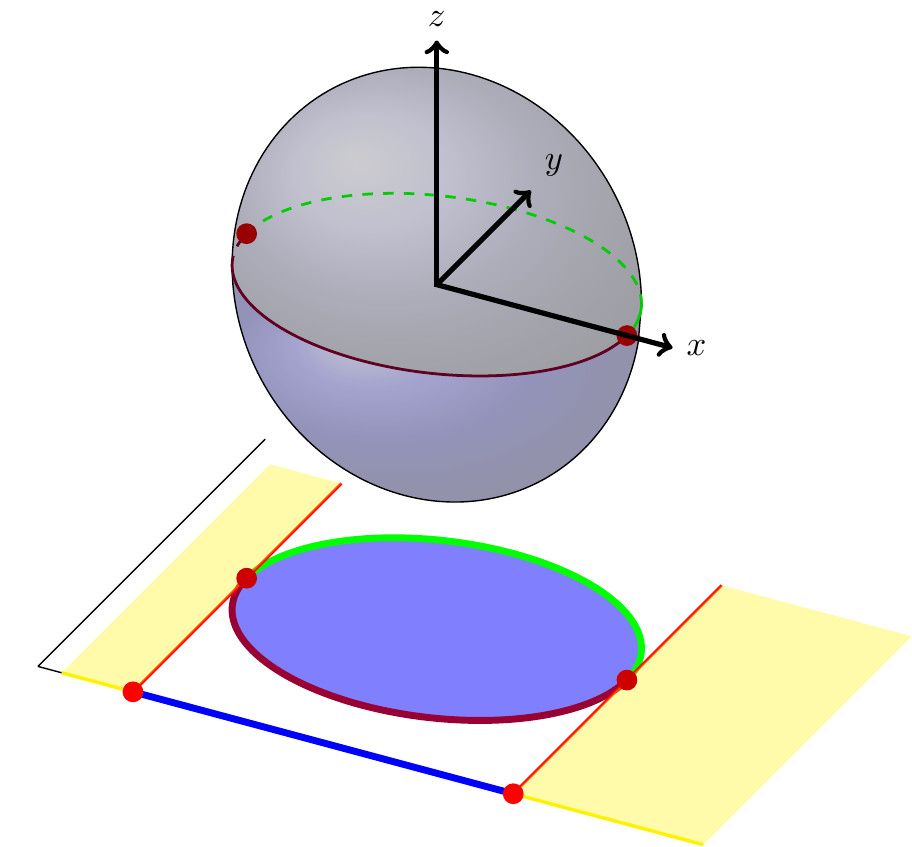}
\caption{Cylindrical decomposition of a sphere.}
\label{fig:shpere}
\end{figure}
\begin{example}[\cite{BPR}]
  Consider the single polynomial $P_s = X_1^2+X_2^2+X_3^2-1$, with
  $P_s =0$ representing a sphere of radius $1$ in $\R^3$, as shown in
  \figurename~\ref{fig:shpere}.  At level 1, $\R$ is partitioned into
  $5$ cells:
\begin{mathpar}
C_{-\infty}=]-\infty,-1[ \and C_{-1}=\{-1\} \and C_0= ]-1,1[ \\ 
C_1= \{1\} \and C_{+\infty}=]1, +\infty[
\end{mathpar}
At level 2, $\R^2$ is partitioned above the previous cells.  There is
a single cell $C_{-\infty}\times \R$ above $C_{-\infty}$ (and
similarly $C_{+\infty}\times \R$ above $C_{+\infty}$).  Above $C_{-1}$
are three cells, its children in the tree:
\begin{mathpar}
\{-1\} \times ]-\infty,0[ \and \{(-1,0)\} \and \{-1\} \times ]0,+\infty[
\end{mathpar}
The cells above $C_1$ are similar.

And above $C_0$ are 5 cells: the interior of the disc $C_{0,0}$, its
lower and upper edges $C_{0,-1}$ and $C_{0,1}$ and the exterior of the
circle (the lower and upper parts) $C_{0,-\infty}$ and
$C_{0,+\infty}$:
\begin{mathpar}
C_{0,1}: \left \{ 
\begin{array}{l} -1<x_1<1 \\ x_2= \sqrt{1 - x_1^2}
\end{array} \right.

\and
C_{0,-1}: \left \{ 
\begin{array}{l} -1<x_1<1 \\ x_2= -\sqrt{1 - x_1^2}
\end{array} \right.

\and
C_{0,+\infty}:\left \{ 
\begin{array}{l} -1<x_1<1 \\ x_2> \sqrt{1 - x_1^2}
\end{array} \right.

\and
C_{0,-\infty}: \left \{ 
\begin{array}{l} -1<x_1<1 \\ x_2< -\sqrt{1 - x_1^2}
\end{array} \right.

\and
C_{0,0}: \left \{ 
\begin{array}{l} -1<x_1<1 \\- \sqrt{1 - x_1^2}< x_2 < \sqrt{1 - x_1^2}
\end{array} \right.
\end{mathpar}
At level 3, cell $C_{0,-1}$ is further lifted in three cells where
$C_{0,-1,0}$ is half the equator circle of the sphere:
\begin{mathpar}
C_{0,-1,-\infty}: \left \{ 
\begin{array}{l} -1<x_1<1 \\ x_2= -\sqrt{1 - x_1^2} \\ x_3<0
\end{array} \right.\and
C_{0,-1,0}: \left \{ 
\begin{array}{l} -1<x_1<1 \\ x_2= -\sqrt{1 - x_1^2} \\ x_3=0
\end{array} \right.\and
C_{0,-1,+\infty}: \left \{ 
\begin{array}{l} -1<x_1<1 \\ x_2= -\sqrt{1 - x_1^2} \\ x_3>0
\end{array} \right.
\end{mathpar}
And $C_{0,0}$ is lifted into 5 cells: below (and above) the inferior
(resp. superior) half of the sphere, said inferior (resp. superior)
half, and the interior of the sphere.  These cells are determined by
two functions $f_1(x_1,x_2)=-\sqrt{1-x_1^2-x_2^2}$ and
$f_2(x_1,x_2)=\sqrt{1-x_1^2-x_2^2}$.
\end{example}

\begin{algorithm2e}
\DontPrintSemicolon

\SetKwFunction{Check}{Check}

\KwData{A cylindrical decomposition having parameter $C$
as an element.}
{\Check}$(\varphi,i,C,\mathcal S)$: a boolean\;
\KwIn{$\varphi$, a prenex  sentence with $n$ variables,
$C$, a $\mathcal P$-invariant cell of level $i$}
\KwIn{$\mathcal S$, a set of pairs of polynomials and signs}
\KwOut{the truth value of $\varphi$}
\KwData{$j,k$, some indices}
 \tcp{The expression $sign(P(C))$ uses the sign invariance of $C$}
 $\mathcal S \leftarrow \mathcal S \cup \{(P,sign(P(C))\mid P \in \mathcal P_i \}$\;
 \tcp{When $i=n$, all atomic formulas of $\psi$ are determined by $\mathcal S$}
 \lIf {$i=n$}
 {
  \Return $\psi(\mathcal S)$
 }
 \tcp{Let $C_1,\ldots,C_k$ be the children of $C$}
 \uIf{$Q_{i+1}=\exists$}
 {
  \For {$j{\bf ~from~}1{\bf ~to~}k$}
  {\lIf{$\Check(\varphi,i+1,C_j,\mathcal S)$}{\Return {\bf true}}}
  \Return {\bf false}\;
 }
 \Else
 {
  \For {$j{\bf ~from~}1{\bf ~to~}k$}
  {\lIf{$\neg \Check(\varphi,i+1,C_j,\mathcal S)$}{\Return {\bf false}}}
  \Return {\bf true}\;
 }

\caption{Checking the truth of a formula}
\label{algo-fosolve}
\end{algorithm2e}
Let us explain how a cylindrical decomposition is useful for
first-order theory of reals. Any sentence can be transformed into an
equivalent prenex formula $\varphi=Q_1 x_1 \ldots Q_n x_n \psi$ such
that $Q_i \in \{\forall,\exists\}$ and $\psi$ is a quantifier free
formula that checks signs of polynomials evaluated on some of the
$x_i$'s.  Thus by syntactical examination, we first build the family
$\mathcal P$ from the polynomials occurring in $\psi$. Assume that we
produce a cylindrical decomposition for $\mathcal P$. Then
Algorithm~\ref{algo-fosolve} solves the decision problem with the call
${\tt Check}(\varphi,0,\R^0,\emptyset)$.  The correctness of the
algorithm is proved by (1) the sign invariance of the cells, (2) the
partition of $C \times \R$ between the children of a cell $C$ and (3)
a backward inductive property: given a cell C of level $i$, the truth
of $Q_{i+1} x_{i+1} \ldots Q_n x_n \psi$ does not depend on the point
$(x_1,\ldots,x_i) \in C$.

The section is organized as follows. In
subsection~\ref{subsec:algo-se} we develop algorithms for rings with
some additional assumptions that depend on the algorithms (also
presented in~\cite{BPR}). The main hypothesis is that we consider
subrings of $\R$ for which there is a decision procedure for
evaluation of the sign of an item. In
Subsection~\ref{subsec:triangular}, we introduce triangular systems
which are representations of algebraic reals and domains of $\R$ and
we establish that they are
sign-effective. Subsection~\ref{subsec:decomposition} is devoted to
the building of a cylindrical decomposition.  It consists in two
stages: the elimination stage that enlarges $\mathcal P$ and the
lifting stage that builds the cylindrical decomposition. In this
decomposition a cell is represented by an algebraic real (i.e. a
triangular system) belonging to it.

\subsection{Algorithms in sign-effective subrings of reals}
\label{subsec:algo-se}

{\bf Preliminary remarks.}  Let us denote by $\An$ a domain
\textit{i.e.}, a ring with no divisors of zero\MoS{Pour l'instant je
  ne change pas mais il faudrait que $\An$ soit un domaine r\'eel ou
  \`a tout le moins ordonn\'e car \`a la fin du paragraphe on regarde
  des conditions de signe sur les \'el\'ements de $\An$ (quand on
  passe aux indices de Cauchy}. \SH{C'est pr\'ecis\'e au paragraphe suivant
  et pour certains algo on ne travaille pas sur un sous-anneau de $\R$}$
  \F_{\An}$ denotes the field of
fractions of $\An$.  Whenever we will describe algorithms involving a
domain $\An$, we assume a representation of an item of $\An$. For
instance, the representation of $\frac{p}{q}\in \Q$ could be the pair
of integers $(p,q)$. We do not require that the representation is
unique but that the following operations are effective: addition,
multiplication and zero-test. We denote multiplication and addition as
usual. The function that performs the zero-test is denoted ${\tt
  Null}(\An,d)$ with $d$, a representation of some item of $\An$.

\smallskip
The goal of this section is to exhibit some problems that can be
solved in $\An[X]$ (for $\An \subseteq \R$) 
when, in addition to the previous operations, 
the sign of an element of $\An$ can be determined. The sign is
defined by $\sign(0)=0$ and for $x \in \An \setminus\{0\}$, $\sign(x) =
1$ if $x>0$, $\sign(x)=-1$ if $x<0$. 
The function that computes the sign
is denoted ${\tt Sign}(\An,d)$ with $d$, a representation of some
item of $\An$.
Since the procedures
we describe may depend on additional properties like this one,
we will indicate which properties are assumed for the algorithms.

\noindent {\bf Notations.}
The sign of a permutation that reverts the order of $i$ items is
denoted by $\varepsilon_{i}=(-1)^{\frac{i(i-1)}{2}}$.  We denote by
$Rem$ the remainder of the Euclidean division in $\An[X]$: for
polynomials $P,Q \in \An[X]$ with respective degrees $p,q$,
$Rem(P,Q)\in \F_{\An}[X]$ is the unique polynomial of degree less than
$q$ such that there exists $C \in \F_{\An}[X]$ with $P = QC +
Rem(P,Q)$.

\paragraph*{Computing the degree of a gcd.}

We start with a characterization of the degree of the gcd of
two polynomials that holds in any domain. The interest
of this characterization is  that it only involves whether some
determinants in $\An$ are null
and thus can be computed by additions, multiplications and zero-tests.
Furthermore, subresultants will also be useful later on.

\begin{definition}[Sylvester-Habicht matrices and subresultants]
  Let $\An$ be a domain. Let $P,Q \in \An[X]$ with $P=\sum_{i\leq p}
  a_iX^i$ and $Q=\sum_{i\leq q} b_iX^i$ such that $a_p\neq 0$,
  $b_q\neq 0$ and $q\leq p$. Then the Sylvester-Habicht matrix of
  order $j$ for $0\leq j \leq \min(p-1,q)$ is the $(p+q-2j) \times
  (p+q-j)$ matrix $SyHa_j(P,Q)$ whose rows are
  $X^{q-j-1}P,\ldots,P,Q,\ldots,X^{p-j-1}Q$ considered as vectors with
  respect to the basis $X^{p+q-j-1},\ldots,X,1$.

\noindent
The $j$-th subresultant denoted $sRes_j(P,Q)$ is the determinant
of the square matrix $SyHa_{j,j}(P,Q)$ obtained by taking the first
$p+q-2j$ columns of $SyHa_j(P,Q)$. 
When $q<p$, this definition
is extended for $q<j\leq p$ by: $sRes_p(P,Q)=a_p$,
and $sRes_j(P,Q)=0$ for $q<j \leq p-1$.
\end{definition}

\begin{remark}\label{rq:sresq} 
Observe that when $q<p$, $SyHa_q(P,Q)$ consists of
  $Q,\ldots,X^{p-q-1}Q$ (without any occurrence of $P$). Hence
  $sRes_q(P,Q)$ is the determinant of a matrix obtained by reverting
  the rows of $b_qId_{p-q}$, which yields
  $sRes_q(P,Q)=\varepsilon_{p-q}b_q^{p-q}$.
\end{remark}

\begin{example}\label{ex:sRes}
  Consider polynomials $P=\alpha X^2-1$ and $Q=X+\beta$, obtained from
  the \polita of \figurename~\ref{fig:exPolIta} when the value of
  $X_1$ has been fixed.  By definition, we have $sRes_2(P,Q)=\alpha$,
  and by the above remark, $sRes_1(P,Q)=1$.  Precisely $SyHa_1(P,Q)$
  is the one row matrix $(1,\beta)$ and $SyHa_{1,1}(P,Q)=(1)$.  For
  $j=0$, one must compute the determinant of the matrix whose rows are
  $P,Q,XQ$, namely
\[
SyHa_{0}(P,Q)=
\begin{pmatrix}
\alpha & 0 & -1 \\
0 & 1 & \beta \\
1 & \beta & 0
\end{pmatrix}\]
whose determinant is $1-\alpha\beta^2$. 
\end{example}

\begin{proposition}
\label{proposition:syha-gcd}
Let $\An$ be a domain and
$P,Q \in \An[X]$ with $P=\sum_{i\leq p} a_iX^i$
and $Q=\sum_{i\leq q} b_iX^i$ such that $a_p\neq 0$,
$b_q\neq 0$ and $q\leq p$. Let $0\leq j \leq \min(p-1,q)$. Then
$deg(gcd(P,Q))=j$ if and only if 
$sRes_0(P,Q)=\cdots=sRes_{j-1}(P,Q)=0$ and $sRes_j(P,Q)\neq 0$. 
Consequently when $p=q$, $deg(gcd(P,Q))=p$
if and only if $sRes_0(P,Q)=\cdots=sRes_{p-1}(P,Q)=0$.
\end{proposition}
\begin{proof}
  Observe that $sRes_j(P,Q)=0$ if and only if there is a non trivially null
  linear combination of polynomials
  $\alpha_{q-j-1}X^{q-j-1}P+\cdots+\alpha_0P
  +\beta_0Q+\cdots+\beta_{p-j-1}X^{p-j-1}Q$ of degree strictly less
  than $j$.  This is equivalent to the existence of two non null polynomials
  $U=\sum_{i\leq q-j-1}\alpha_iX^i$ and $V=\sum_{i\leq
    p-j-1}\beta_iX^i$ such that $deg(UP+VQ)<j$.

\noindent
We claim that $sRes_0(P,Q)=\cdots=sRes_{j-1}(P,Q)=0$ if, and
only if, it is the case that $deg(gcd(P,Q))\geq j$, which will yield the desired
conclusion. Assume that $deg(gcd(P,Q))\geq j$ which is equivalent
to $deg(lcm(P,Q))\leq p+q-j$ which is equivalent to the
existence of polynomials $U,V$ with $deg(U)\leq q-j$,
$deg(V)\leq p-j$ and $UP=-VQ$. Our previous observation implies that
$sRes_0(P,Q)=\cdots=sRes_{j-1}(P,Q)=0$.

\noindent
The reverse implication is established by induction on $j$.  When
$sRes_0(P,Q)=0$, the existence of $U$ and $V$ such that $UP+VQ=0$ with
$deg(U)<q$ and $deg(V)<p$
implies $deg(gcd(P,Q))\geq 1$.  When
$sRes_0(P,Q)=\cdots=sRes_{j}(P,Q)=0$, the inductive hypothesis applied
to $j-1$ implies $deg(gcd(P,Q))\geq j$. From $sRes_{j}(P,Q)=0$, we
again obtain $U,V$ such that with $deg(U)< q-j$, $deg(V)< p-j$ and
$deg(UP+VQ)<j$. Since $gcd(P,Q)$ divides $UP+VQ$ this implies that
$UP+VQ=0$ and so $deg(lcm(P,Q))< p+q-j$ and finally $deg(gcd(P,Q))\geq
j+1$.

\qed 
\end{proof}

Due to the importance of the subresultant notion, we want a way to
compute them efficiently. To this aim, we introduce the
``polynomial'' matrices  and determinants. 
Let us introduce additional notations.

\begin{definition}
Let $P_1,\ldots, P_m$ be polynomials in $\An[X]$of degrees less than $n$
with $m\leq n$ and $P_i=\sum_{j<n} p_{i,j}X^j$. 
Then $pmat_n(P_1,\ldots, P_m)$ is the $m\times m$ matrix whose
items are defined by:
\begin{itemize}
	\item For all $i\leq m$, $j<m$, $pmat_n(P_1,\ldots, P_m)[i,j]=p_{i,n-j}$.
	\item For all $i\leq m$, $pmat_n(P_1,\ldots, P_m)[i,m]=P_i$.
\end{itemize}
Additionally, let $pdet_n(P_1,\ldots, P_m)=det(pmat_n(P_1,\ldots, P_m))$.
\end{definition}
Otherwise stated, the $i$th row of matrix $pmat_{n}(P_1,\ldots, P_m)$
consists of coefficients of $P_i$ in descending order down to $n-m+1$
ended by polynomial $P_i$ itself.

\smallskip 
\begin{definition}
Consider $P,Q$ polynomials with respective degrees $p>q$.
We define, for $0\leq j\leq p$, 
\begin{itemize}
 \item for $0\leq j\leq q$,
 $sResP_j(P,Q)pdet_{p+q-j}(X^{q-j-1}P,\ldots, P,
Q,\ldots,X^{p-j-1}Q)$,\\
that is  $\det(SyHaP_j(P,Q))$, where \\
 $SyHaP_j(P,Q) = pmat_{p+q-j}(X^{q-j-1}P,\ldots, P,Q,\ldots,X^{p-j-1}Q)$.
 
\item for $q<j<p-1$, $sResP_j(P,Q)=0$;
\item for $j=p-1$, $sResP_j(P,Q)=Q$ (which is consistent with the
  original definition in case $q=p-1$);
	\item for $j=p$, $sResP_j(P,Q)=P$.
\end{itemize}
\end{definition}

From the above definition, one can straightforwardly see:
\begin{proposition}
$sResP_j(P,Q)$ is a polynomial of degree at most $j$
and the coefficient of degree $j$ of this polynomial is  $sRes_j(P,Q)$.
\end{proposition}

\noindent\textbf{Additional assumption.}
We assume here that the integral division is effective in $\An$:
given $a,b \in \An$, there is an algorithm that answers whether
there exists $c \in \An$ with $a=bc$ and returns $c$ in the positive case.
This is the case in particular in any ring
over $\Z[X_1,\ldots,X_k]=\Z[X_1]\cdots[X_{k-1}][X_k]$ 
or $\Q[X_1,\ldots,X_k]$ where the algorithm consists in trying to perform 
a (recursive) Euclidean division, stopping and answering negatively
when a coefficient of the quotient is not in the corresponding ring 
or there is a non null remainder. We denote the integral division
by the usual fraction symbol since we will only use it when the result
is defined.

\begin{algorithm2e}
\DontPrintSemicolon

\SetKwFunction{Subresultants}{Subresultants}
\SetKwFunction{Degree}{Degree}
\SetKwFunction{NullPol}{NullPol}

{\Subresultants}$(\An,P,p,Q,q)$: a vector\;
\KwIn{$P,Q$, non null polynomials in $\An[X]$ with respective degrees $p>q$}
\KwOut{the vector of subresultants $(sRes_i(P,Q))_{0\leq i\leq p}$}
\KwData{$SresP$ a vector over $\An[X]$ indexed by $[0,p]$} 
\KwData{$s,t$ vectors over $\An$ indexed by $[0,p]$; 
$i,j,k,\ell$, some indices}
$SresP[p] \leftarrow P$;
$s[p]\leftarrow 1$; $t[p]\leftarrow 1$;
$SresP[p-1] \leftarrow Q$; $t[p-1]\leftarrow Q[q]$\;
\leIf {$q=p-1$}{$s[p-1]\leftarrow t[p-1]$}{$s[p-1]\leftarrow 0$}
\lFor{$\ell {\bf ~from~} q+1 {\bf ~to~} p-2$}
{$s[\ell]\leftarrow 0$}
$SresP[q] \leftarrow \eps_{p-q}t[p-1]^{p-q-1}Q$;
$t[q]\leftarrow SresP[q][q]$; $s[q]\leftarrow t[q]$\;
$i\leftarrow p+1$; $j\leftarrow p$\;
\While{$\Degree(\An,SresP[j-1])\neq -\infty$}
{
$k\leftarrow \Degree(\An,SresP[j-1])$;
$t[j-1] \leftarrow SresP[j-1][k]$\;
\lIf{$k=j-1$}
{ 
$s[j-1]\leftarrow t[j-1]$
}
\Else
{
\lFor{$\ell {\bf ~from~} k+1 {\bf ~to~} j-1$} {$s[\ell]\leftarrow 0$}
$t[k]\leftarrow 
\varepsilon_{j-k}(\frac{t[j-1]}{s[j]})^{j-k-1}t[j-1]$; 
$s[k]\leftarrow t[k]$\;
$SresP[k] \leftarrow \frac{s[k]sResP[j-1])}{t[j-1]}$\;
}
$SresP[k-1] \leftarrow \frac{-Rem(t[j-1]s[k] sResP[i-1],sResP[j-1])}{s[j]t[i-1]}$;
$i\leftarrow j$; $j\leftarrow k$;
}
\lFor{$\ell {\bf ~from~} 0 {\bf ~to~} j-2$}{$s[\ell]\leftarrow 0$}
\Return $s$
\caption{Computing the subresultants for $P,Q$.}
\label{algo-subresultant}
\end{algorithm2e}

Our goal is to compute $sRes_j(P,Q)$ by decreasing
values of $j$ and only relying on Euclidean divisions that remain in $\An[X]$.
For sake of clarity, we denote $s_j=sRes_j(P,Q)$ and $t_j$ the leading
coefficient of $sResP_j(P,Q)$ \emph{except for $s_p=t_p=1$}. 
When $sResP_j(P,Q)$ has degree $j$, we have $s_j=t_j$.
Developing the last column  w.r.t. the degrees of $X$
and observing that for degrees $>j$ the corresponding vector of 
reals already occurs in a former column, we can safely substitute to the
polynomials their truncation up to degree $j$. 
Then it is immediate that $sResP_j(P,Q)=0$ iff
there exist polynomials $U,V$ with $deg(U)<q-j$,
$deg(V)< p-j$ and $UP+VQ=0$. As a consequence,
for all $j'\leq j$, $sResP_{j'}(P,Q)=0$.

The next proposition is the basis of Algorithm~\ref{algo-subresultant}
for the efficient computation of subresultants. As can be deduced from
this proposition, the computation consists in taking successive
remainders of Euclidean divisions (up to some constant) in order to
get $sResP_{i_j-1}(P,Q)$ and then some scalar multiplications and
divisions to get $sResP_{i_{j+1}}(P,Q)$. Function {\tt Degree} returns
the degree of a polynomial in $\N \cup \{-\infty\}$ by looking at the
first non null coefficient (using {\tt Null} function).

\begin{proposition}
Let $P,Q$ be non null polynomials of $\An[X]$ with $p=deg(P)>deg(Q)=q$.
There exists a sequence of strictly decreasing indices $i_1,i_2,\ldots,i_J$
with $i_1=p+1$, $i_2=p$, $i_3=q$ that fulfills the following properties:
\begin{itemize}
	\item for all $1<j\leq J$, $sResP_{i_j}(P,Q)$ has degree $i_j$ 
	(and so $s_{i_j}=t_{i_j}$),
	for all $j< J$, $sResP_{i_j-1}(P,Q)$ has degree $i_{j+1}$ and if $i_J>0$
	then for all $k<i_J$, \\
	$sResP_{k}(P,Q)=0$ and $sResP_{i_{J-1}-1}(P,Q)=gcd(P,Q)$;
	\item for all $j<J$, when $i_j-1>i_{j+1}$, for all $i_{j+1}<k<i_j-1$,
	$sResP_{k}(P,Q)=0$ and $t_{i_j-1}sResP_{i_{j+1}}(P,Q)=s_{i_{j+1}}sResP_{i_j-1}(P,Q)$
	with\\ $s_{i_{j+1}}=\varepsilon_{i_j-i_{j+1}}\frac{(t_{i_{j}-1})^{i_j-i_{j+1}}}{(s_{i_j})^{i_j-i_{j+1}-1}}$;
	\item for all $1<j<J$,
	$s_{i_j}t_{i_{j-1}-1}sResP_{i_{j+1}-1}(P,Q)=$\\
	\hspace{4cm} $ -Rem(s_{i_{j+1}}t_{i_{j}-1}sResP_{i_{j-1}-1}(P,Q),sResP_{i_{j}-1}(P,Q))$.
\end{itemize}
\end{proposition}
Substituting in the equation of the third item
$sResP_{i_{j-1}-1}(P,Q)$ by \\
$\frac{t_{i_{j-1}-1}}{s_{i_{j}}}sResP_{i_{j}}(P,Q)$
(justified by the equation of the second item)
and then multiplying by $\frac{t_{i_{j-1}-1}}{s_{i_{j}}}$
one also obtains:\\
$s_{i_j}^2sResP_{i_{j+1}-1}(P,Q)=
	-Rem(s_{i_{j+1}}t_{i_{j}-1}sResP_{i_{j}}(P,Q),sResP_{i_{j}-1}(P,Q))$. 

\begin{proof}
  Let $R=Rem(P,Q)$.  Let us look at $SyHaP_j(P,Q)$ for $j\leq q-1$.
  Write $C=\sum_{i\leq p-q}c_iX^i$ (the quotient of Euclidean division
  of $P$ by $Q$). We have $R=P-\sum_{i\leq p-q}c_iX^iQ$. Due to this
  equality, changing the rows $X^{q-j-1}P,\ldots,P$ by
  $X^{q-j-1}R,\ldots,R$ does not modify the determinant
  $sResP_j(P,Q)$. We define the determinant $D_j$ of the matrix
  obtained by reverting the order of the rows and replacing $R$ by
  $-R$.  The first operation amounts to multiplying by $\eps_{p+q-2j}$
  and the second one by $(-1)^{q-j}$.  Since
  $\eps_{p+q-2j}(-1)^{q-j}=\eps_{p-q}$, we have:\\
  $D_j=\eps_{p-q}sResP_j(P,Q)$.

  \smallskip We first prove the properties related to indexes between
  $p$ and $q-1$. Let us look at the second item.  For the first part
  by convention for all $q<j<p-1$, $sResP_j(P,Q)=0$. The second part
  of the second item corresponds to the case $j=2$ with $s_{p}=a_p$,
  $t_{p-1}=s_q=b_q$.  So the equation can be written as:
$$b_qsResP_q(P,Q)=s_qQ \mbox{ with } 
s_q=\varepsilon_{p-q}\frac{b_q^{p-q}}{1^{p-q-1}}$$
which is equivalent to: 
$$sResP_q(P,Q)=\varepsilon_{p-q}b_q^{p-q-1}Q.$$
Since $sResP_q(P,Q)=pdet_{p}(Q,\ldots,X^{p-q-1}Q)$, the result is immediate.
Let us look at the third item: $D_{q-1}=-b^{p-q+1}R$.
So \[sResP_{q-1}(P,Q)=-Rem(\eps_{p-q}b^{p-q+1}P,Q).\]
By convention, $s_p=t_p=1$,  $sResP_{p}(P,Q)=P$ and  
$sResP_{p-1}(P,Q)=Q$ implying $t_{p-1}=b_q$. Furthermore we have shown that
$s_q=\eps_{p-q}b^{p-q}$. Substituting in the previous equation establishes
the third item.

\smallskip\noindent
We prove the remaining properties by induction on $J$. 
Let $R=Rem(P,Q)=0$ which implies that $Q=gcd(P,Q)$
and $sResP_{q-1}(P,Q)=0$. 
So the base case ($J=3$) is established.

\smallskip
Let $R=Rem(P,Q)\neq 0$.
Let $r$ be the degree of $R$,
we claim that: 
\begin{equation}
\label{eq:sResPj}
  \forall j<q-1\ sResP_j(P,Q)=\varepsilon_{p-q}b_q^{p-r}sResP_j(Q,-R)
\end{equation}
When $j\leq r=deg(R)$, $D_j$ can be obtained starting from
$SyHaP_j(Q,-R)$ by adding the rows $X^{p-j-1}Q,\ldots,X^{r-j}Q$ and
taking the determinant.  Thus $D_j=$ \\
$b_q^{p-r}sResP_j(Q,-R)$ and so
$sResP_j(P,Q)=b_q^{p-r}\eps_{p-q}sResP_j(Q,-R)$.  When $r<j<q-1$ by
definition $sResP_j(Q,-R)=0$ but $sResP_j(P,Q)=D_j=0$ since the
polynomial matrix
$pmat_{p+q-j}(X^{p-j-1}Q,\ldots,Q,X^{q-j-1}R,\ldots,R)$ is upper
triangular up to its $p-j+1^{th}$ column and since the degree
$X^{q-j-1}R$ is less than $q-1$, the diagonal term of this column is
null.

\smallskip
Due to this proportionality between
$sResP_j(P,Q)$ and $sResP_j(Q,-R)$ with factor $\eps_{p-q}b_q^{p-r}$
and the inductive hypothesis,
it only remains to prove that the two following equalities hold:
\begin{equation}
\label{eq:sres1}
s_{q}t_{p-1}sResP_{r-1}(P,Q)=
	-Rem(s_{r}t_{q-1}sResP_{p-1}(P,Q),sResP_{q-1}(P,Q))
\end{equation}
and
\begin{equation}
\label{eq:sres2}
s_{r}=\varepsilon_{q-r}\frac{(t_{q-1})^{q-r}}{(s_{q})^{q-r-1}}
\end{equation}

For Equation~(\ref{eq:sres1}), using the inductive hypothesis for the
pair $(Q,-R)$, the following equation holds:
$$s'_{q}{}^2 sResP_{r-1}(Q,-R)=
	-Rem(s'_{r}t'_{q-1}sResP_{q}(Q,-R),sResP_{q-1}(Q,-R)) $$
where the primed version of $s_i$ and $t_i$ are related to the pair
$(Q,-R)$.
By convention, $s'_{q}=1$. So:
\[s_{q}t_{p-1}sResP_{r-1}(P,Q)=s_{q}t_{p-1}\eps_{p-q}b_q^{p-r}sResP_{r-1}(Q,-R)\]
\[=(\eps_{p-q}b_q^{p-q})(b_q)\eps_{p-q}b_q^{p-r}sResP_{r-1}(Q,-R)\]
\[=-Rem((\eps_{p-q}b_q^{p-r}s'_{r})(\eps_{p-q}b_q^{p-q+1}t'_{q-1})
sResP_{q}(Q,-R),sResP_{q-1}(Q,-R)).\]
Observe that the factor of proportionality established above
implies that\\
$s_r=\eps_{p-q}b_q^{p-q+1}s'_{r}$. 
\\Since $sResP_{q-1}(P,Q)=-\eps_{p-q}b_q^{p-q+1}R$ 
and $sResP_{q-1}(Q,-R)=-R$, one obtains 
$t_{q-1}=\eps_{p-q}b_q^{p-q+1}$.
So:
$$s_{q}t_{p-1}sResP_{r-1}(P,Q)=-Rem(s_rt_{q-1}sResP_{q}(Q,-R),sResP_{q-1}(Q,-R))$$
$$=-Rem(s_rt_{q-1}Q,-R)=-Rem(s_rt_{q-1}Q,-\eps_{p-q}b_q^{p-q+1}R)$$
$$=-Rem(s_{r}t_{q-1}sResP_{p-1}(P,Q),sResP_{q-1}(P,Q))$$

For Equation~(\ref{eq:sres2}), let us look at the following matrices. 
$$
\begin{pmatrix}
b_q& b_{q-1}& \ldots& \ldots &X^{p-q-1}Q\\
0    &  b_q    & \ldots & \ldots &X^{p-q-2}Q\\
\ldots &\ldots& \ldots & \ldots & \ldots\\
0& 0& \ldots& b_q& XQ\\
0& 0& \ldots& 0& Q
\end{pmatrix}
\quad
\begin{pmatrix}
b_q& b_{q-1}& \ldots& \ldots &X^{p-q}Q\\
0    &  b_q    & \ldots & \ldots &X^{p-q-1}Q\\
\ldots &\ldots& \ldots & \ldots & \ldots\\
0& 0& \ldots& b_q& Q\\
0& 0& \ldots& 0& -R
\end{pmatrix}
$$
The left
matrix that we define $D_q$ has been obtained by reverting the $p-q$ rows
of $SyHaP_q(P,Q)$. So its determinant is equal to $\eps_{p-q}sResP_q(P,Q)$.
The right matrix is  $D_{q-1}$. As we have already seen, its determinant is
equal to $\eps_{p-q}sResP_{q-1}(P,Q)$. Denoting $-R=\sum_{i\leq r} \alpha_i X^i$,
it is now obvious that $b_q \alpha_r=\frac{t_{q-1}}{s_q}$.
As a consequence, we obtain that:
\begin{equation}
\label{eq:sresPqmoinsun}
sResP_{q-1}(P,Q)=-\eps_{p-q}b_q^{p-q+1}R
\end{equation}

Let us look at the following matrices. 

{\scriptsize
\[
\begin{pmatrix}
b_q& b_{q-1}& \ldots& \ldots & \ldots&\ldots&\ldots&\ldots&X^{p-r-1}Q\\
0    &  b_q    & \ldots & \ldots&\ldots&\ldots&\ldots &\ldots&X^{p-r-2}Q\\
\ldots &\ldots& \ldots & \ldots &\ldots&\ldots&\ldots&\ldots& \ldots\\
0& 0& \ldots& b_q& \ldots&\ldots&\ldots&\ldots&Q\\
0& 0& \ldots& 0& 0&0&\ldots&0&-R\\
0& 0& \ldots& 0 & &&\ldots&\alpha_r&-XR\\
\ldots &\ldots& \ldots & \ldots &\ldots &\ldots&\ldots&\ldots&\ldots\\
0    &  0    & \ldots & \ldots &0&\alpha_r&\ldots&\ldots&-X^{q-r-2}R\\
0& 0& \ldots& \ldots &\alpha_r&\alpha_{r-1}&\ldots&\ldots&-X^{q-r-1}R
\end{pmatrix}
\ 
\begin{pmatrix}
b_q& b_{q-1}& \ldots& \ldots & \ldots&\ldots&\ldots&\ldots&X^{p-r-1}Q\\
0    &  b_q    & \ldots & \ldots&\ldots&\ldots&\ldots &\ldots&X^{p-r-2}Q\\
\ldots &\ldots& \ldots & \ldots &\ldots&\ldots&\ldots&\ldots& \ldots\\
0& 0& \ldots& b_q& \ldots&\ldots&\ldots&\ldots&Q\\
0& 0& \ldots& \ldots &\alpha_r&\alpha_{r-1}&\ldots&\ldots&-X^{q-r-1}R\\
0    &  0    & \ldots & \ldots &0&\alpha_r&\ldots&\ldots&-X^{q-r-2}R\\
\ldots &\ldots& \ldots & \ldots &\ldots &\ldots&\ldots&\ldots&\ldots\\
0& 0& \ldots& 0 & &&\ldots&\alpha_r&-XR\\
0& 0& \ldots& 0& 0&0&\ldots&0&-R
\end{pmatrix}
\]
}

The left matrix is $D_{r}$ and the right matrix
has been obtained by reverting its last $q-r$ columns. So the determinant of
the latter matrix is proportional to the determinant of the former with factor
$\eps_{q-r}$. On the other hand, the determinant of the right matrix
is equal to the determinant of $D_{j-1}$ multiplied by $(b_q \alpha_r)^{q-r-1}$.
Combining the different equalities, we obtain that:
$sRes_r(P,Q)=\eps_{q-r}(\frac{t_{q-1}}{s_q})^{q-r-1}sRes_{q-1}(P,Q)$
and consequently $s_r=t_r= \eps_{q-r}\frac{t_{q-1}^{q-r}}{s_q^{q-r-1}}$.

This concludes the proof.
\qed
\end{proof}

\paragraph*{Computing sign realizations at roots of a polynomial}~\\

Now we consider the special case of $\An = \D$, $ \D$ being a
sign-effective subring of $\R$.  The main ingredient for analyzing
real roots of a univariate polynomial is the Cauchy index. We denote
by $Zer(P)=\{z \in \R \mid P(z)=0\}$, $mult(P,z)=\max\{n \mid (X-z)^n
|P\}$ and $Pole(Q/P)=\{z \in \R \mid mult(Q,z)<mult(P,z)\}$.  For $z$
in $Pole(Q/P)$, remark that $Q/P(w)$ goes to $+\infty$ or $-\infty$ as
$w$ tends to $z$ on the right (respectively on the left), therefore
the sign of $Q/P$ keeps constant sufficiently close on the right
(respectively on the left) of $z$.

\begin{definition}
  Let $P,Q \in \D[X]$. Then the Cauchy index of $Q/P$ is defined by:\\
  \centerline{$Ind(Q/P)= \frac12\sum_{z \in Pole(Q/P)}\sign((Q/P)(z^+))-
    \sign((Q/P)(z^-))$} where $\sign((Q/P)(z^+))$ and
  $\sign((Q/P)(z^-))$) denote respectively the sign of the rational
  function $Q/P$ at the right and at the left of $z$.
\end{definition}

\noindent
For $z \in Pole(Q/P)$, the value $\sign((Q/P)(z^+))- \sign((Q/P)(z^-))$
in $\{-2,0,2\}$ depends on the parity of the difference $\mu_P-\mu_Q$ of
respective multiplicities of $z$ as root of $P$ and $Q$, when $\mu_P
\geq \mu_Q$ (and $\mu_Q=0$ if $z$ is not a root of $Q$).

\begin{example}\label{ex:cauchyInd}
Recall polynomials $P=\alpha X^2-1$ and $Q=X+\beta$ of example~\ref{ex:sRes}.
Let us compute the Cauchy index of $Q/P$ for several values of $\alpha$ and $\beta$.
\begin{itemize}
\item Let $P_1,Q_1$ be the above polynomials with $\alpha=\sqrt{5}$ and $\beta=\frac{\sqrt{5}-7}2$.
These values were obtained by setting $X_1$ to $\frac{1+\sqrt{5}}2$.
The poles of $Q_1/P_1$ are $z_1= -\frac1{\sqrt[4]{5}}$ and $z_2= \frac1{\sqrt[4]{5}}$.
One can see that $X+\beta$ remains negative between those poles.
Hence
\begin{eqnarray*}
Ind(Q_1/P_1) &=& \frac12 (\sign(Q_1/P_1)(z_1^+) -\sign(Q_1/P_1)(z_1^-) \\&&\qquad+ \sign(Q_1/P_1)(z_2^+) -\sign(Q_1/P_1)(z_2^-)) \\
         &=& \frac12 (1- (-1) + (-1) - 1) = 0.
\end{eqnarray*}
\item Let $P_2,Q_2$ be the above polynomials with $\alpha=2\sqrt{5}-1$ and $\beta=0$, which can be obtained by setting $X_1$ to $\sqrt{5}$.
The poles of $Q_2/P_2$ are $z_1 = -\frac1{\sqrt{2\sqrt5-1}}$ and $z_2 = \frac1{\sqrt{2\sqrt5-1}}$.
Now since $Q_2$ has a root between $z_1$ and $z_2$, hence
\begin{eqnarray*}
Ind(Q_2/P_2) &=& \frac12 (\sign(Q_2/P_2)(z_1^+) -\sign(Q_2/P_2)(z_1^-) \\&&\qquad+ \sign(Q_2/P_2)(z_2^+) -\sign(Q_2/P_2)(z_2^-)) \\
         &=& \frac12 (1- (-1) + 1 - (-1)) = 2.
\end{eqnarray*}
\end{itemize}

\end{example}

The Cauchy index can be computed in several ways.  First we observe
that we can assume $q=deg(Q)<deg(P)=p$. Otherwise, let $a_p$ be the
leading coefficient of $P$ and compute the Euclidean division of
$a_p^{2\lceil \frac{q-p+1}{2}\rceil}Q$ by $P$: $a_p^{2\lceil
  \frac{q-p+1}{2}\rceil}Q=PC+R$ with $deg(R) < deg(P)$.  Then
$Ind(Q/P)=Ind(R/P)$. The multiplication by an even power of $a_p$
preserves the signs. Furthermore $R$ is obtained by multiplications,
additions and zero-tests so that it can be performed 
in a general domain $\D$ as indicated in Algorithm~\ref{algo-intrem}.

\begin{algorithm2e}
\DontPrintSemicolon

\SetKwFunction{IntRem}{IntRem}
\SetKwFunction{Degree}{Degree}

{\IntRem}$(\D,Q,q,P,p)$: a polynomial with its degree\;
\KwIn{$P\neq 0,Q$, polynomials in $\D[X]$ with respective degrees $p,q$}
\KwOut{a polynomial positively proportional to $Rem(Q,P)$}
\KwData{$i,j$, some indices}
\BlankLine
\lIf{$q< p$}{\Return $Q,q$}
\For{$i {\bf ~from~} q-p {\bf ~downto~} 0$} 
{
   \lFor{$j {\bf ~from~} 0 {\bf ~to~} p-1$} 
   {$Q[i+j]\leftarrow P[p]Q[i+j]-P[j]Q[i+p]$}
   \lFor{$j {\bf ~from~} 0 {\bf ~to~} i-1$} 
   {$Q[j]\leftarrow P[p]Q[j]$}
}
\lFor{$i {\bf ~from~} p {\bf ~to~} q$} 
   {$Q[i]\leftarrow 0$}
\If{$q-p \mod 2=0$} 
{
   \lFor{$j {\bf ~from~} 0 {\bf ~to~} p-1$} 
   {$Q[j]\leftarrow P[p]Q[j]$}
}
\Return $Q,\Degree(\D,Q)$\;
\caption{Computing a polynomial positively proportional to $Rem(Q,P)$}
\label{algo-intrem}
\end{algorithm2e}

Here we use again the subresultants. 
Let $s=(s_p,\ldots,s_0)$ be a list of reals such that $s_p\neq 0$. Define
$s'$ as the shortest list such that 
$s=(s_p,0,\ldots,0)\cdot s'$. Then we inductively define:
$$
 PmV(s)=
 \left \{
   \begin{array}{l l}
      0  & \mbox{ if } s'=\emptyset\\
      PmV(s')+  \varepsilon_{p-q}sign(s_ps_q) & 
      \mbox{ if } s'=(s_q,\ldots,s_0) \mbox{ and } p-q \mbox{ is odd}\\
      PmV(s') & \mbox{ otherwise}
   \end{array}
   \right .
$$
Here acronym $PmV$ means \emph{ (generalized) permanence minus variations} and as
can be observed from the definition is related to the sign variations
of the sequence $s$. An immediate property of the $PmV$ is
the following one. Let $x_p,\ldots,x_0$ be such that 
$sign(x_p)=\cdots=sign(x_0)\neq 0$, then
$PmV(x_ps_p,\ldots,x_0s_0)=PmV(s_p,\ldots,s_0)$.

Our approach consists in computing the $PmV$ applied on subresultants.

\noindent {\bf Notations.} If $p=deg(P)> q=deg(Q)\geq 0$,  we denote by  $sRes$ the tuple $(sRes_p,\ldots,sRes_0)$.

\begin{example}
For the polynomials of example~\ref{ex:cauchyInd}, we have $sRes(P,Q)=$ \\
$(\alpha,1,-\alpha\beta^2+1)$.
\begin{itemize}
\item In the first case, $sRes(P_1,Q_1) = (\sqrt{5},1,\frac{37-27\sqrt{5}}2)$.
Then
\begin{eqnarray*}
\hspace{-1.5em}PmV(sRes(P_1,Q_1)) &=& PmV\left(1,\frac{37-27\sqrt{5}}2\right) + \sign(\sqrt{5}) \\
                   &=& PmV\left(\frac{37-27\sqrt{5}}2\right) + \sign\left(\frac{37-27\sqrt{5}}2\right) + \sign(\sqrt{5})\\
                   &=& 0  + \sign\left(\frac{37-27\sqrt{5}}2\right) + \sign(\sqrt{5}) = 0 + (-1) + 1 = 0.
\end{eqnarray*}
\item In the second case, $sRes(P_2,Q_2) = (2\sqrt{5}-1,1,0)$.
Then
\begin{eqnarray*}
PmV(sRes(P_2,Q_2)) &=& PmV(1,1) + \sign(2\sqrt{5}-1) \\
                   &=& PmV(1) + \sign(1) + \sign(2\sqrt{5}-1) \\
                   &=& 0 + 1 + 1 = 2.
\end{eqnarray*}
\end{itemize}
\end{example}

\begin{theorem}
Let $P,Q \in \D[X]$ with $p=deg(P)>q= deg(Q)$. Then:\\
\centerline{$PmV(sRes(P,Q))=Ind(Q/P)$}
\end{theorem}
\begin{proof}
  Let $P=\sum_{i\leq p}a_iX^i$, $Q=\sum_{i\leq q}b_iX^i$ and let $R$
  be the remainder of the euclidean division of $P$ by $Q$: $P=QC+R$. We
  consider two cases, according to whether $R=0$ or not.

  \smallskip \noindent If $R=0$ then $Q/P=1/C$ with $a_p/b_q$ the
  leading coefficient of $C$ denoted by $c_{p-q}$, hence
  $sign(c_{p-q})=sign(a_pb_q)$. Observe first that the sign
  of $1/C$ is unchanged between two consecutive poles.
  So the Cauchy index of $1/C$ will be half the sign of $C$
  at $+\infty$ minus the sign of $C$ at $-\infty$.
  If $p-q$ is even then $C(x)$ will go to
  the same sign when $x$ goes either to $+\infty$ or $-\infty$
  entailing that $Ind(Q/P)=0$. Otherwise it will go to opposite signs
  with the sign at $+\infty$ being $sign(a_pb_q)$, thus entailing that
  $Ind(Q/P)=sign(a_pb_q)$.

  \smallskip \noindent On the other hand, $sRes_p(P,Q)=a_p$,
  $sRes_j(P,Q)=0$ for $q<j<p$ and\\
  $sRes_q(P,Q)= 
  \varepsilon_{p-q}b_q^{p-q}$ from
  Remark~\ref{rq:sresq}. By Proposition~\ref{proposition:syha-gcd},
  $sRes_j(P,Q)=0$ for $j<q$. When $p-q$ is even, $PmV(sRes(P,Q))=0$
  and when $p-q$ is odd,
  $PmV(sRes(P,Q))=\varepsilon_{p-q}sign(a_p\varepsilon_{p-q}b_q^{p-q})
  =sign(a_pb_q)$.

\smallskip \noindent
When $R\neq0$, we claim that (1) $Ind(Q/P)=Ind(-R/Q)+sign(a_pb_q)$
when $p-q$ is odd and $Ind(Q/P)=Ind(-R/Q)$ otherwise 
and (2) 
$PmV(sRes(P,Q))= $ \\
$PmV(sRes(Q,-R))+sign(a_pb_q)$
when $p-q$ is odd and $PmV(sRes(P,Q))= $ \\
$PmV(sRes(Q,-R))$ otherwise.
This will imply the  theorem by induction on the degree
of $P$.

\smallskip \noindent
Let $G$ be the gcd of $P$ and $Q$ and write $P=P_1G$, $Q=Q_1G$ and $R=R_1G$.
Obviously $Ind(Q/P)=Ind(Q_1/P_1)$ and $Ind(P/Q)=Ind(P_1/Q_1)$. In addition
the signs of $PQ(x)$ and $P_1Q_1(x)$ coincide on every point which is not
a root of $PQ$. Since the roots of $P_1$ and $Q_1$ are distinct: 
$$\frac{1}{2}(sign(PQ(+\infty))-sign(PQ(-\infty)))=
\frac{1}{2}(sign(P_1Q_1(+\infty))-sign(P_1Q_1(-\infty)))$$
$$=\frac{1}{2}\sum_{z\in Zer(P_1Q_1)} sign((P_1Q_1)(z^+))-sign((P_1Q_1)(z^-))$$
$$=\frac{1}{2}\sum_{z\in Zer(P_1)} sign((Q_1/P_1)(z^+))-sign((Q_1/P_1)(z^-))$$
$$+\frac{1}{2}\sum_{z\in Zer(Q_1)} sign((P_1/Q_1)(z^+))-sign((P_1/Q_1)(z^-))$$
$$= Ind(Q_1/P_1)+Ind(P_1/Q_1)= Ind(Q/P)+Ind(P/Q)= Ind(Q/P)+Ind(R/Q).$$
Since $\frac{1}{2}(sign(PQ(\infty))-sign(PQ(-\infty)))$ is null when
$p-q$ is even and equal to $sign(a_pb_q)$ otherwise we obtain the
first claim.

\smallskip \noindent 
We recall Equation~\ref{eq:sResPj} where $r$ is the degree of $R$:
$$\forall j<q-1\ sResP_j(P,Q)=\varepsilon_{p-q}b_q^{p-r}sResP_j(Q,-R)$$
and Equation~\ref{eq:sresPqmoinsun}:
$$sResP_{q-1}(P,Q)=-\eps_{p-q}b_q^{p-q+1}R.$$
\smallskip \noindent
{\bf Case 1: $q-1>r$.}\\
\begin{footnotesize}
\noindent
{\scriptsize$Pmv(sRes(P,Q))=$ }\\
{\scriptsize $PmV(a_p,0,\ldots,0,\eps_{p-q}b_q^{p-q},
0,\ldots,0,b_q^{p-r}\eps_{p-q}sRes_r(Q,-R),\ldots,
b_q^{p-r}\eps_{p-q}sRes_0(Q,-R))$}\\
{\bf Case 1.1: $q>r-1$ and $p-q$ is even.}\\
{\scriptsize $Pmv(sRes(P,Q))=$}\\
{\scriptsize $PmV(\eps_{p-q}b_q^{p-q},
0,\ldots,0,b_q^{p-r}\eps_{p-q}sRes_r(Q,-R),\ldots,
b_q^{p-r}\eps_{p-q}sRes_0(Q,-R))$}\\
{\scriptsize$=PmV(b_q^{p-q},
0,\ldots,0,b_q^{p-r}sRes_r(Q,-R),\ldots,
b_q^{p-r}sRes_0(Q,-R))$}\\
{\scriptsize$=PmV(1,
0,\ldots,0,b_q^{q-r}sRes_r(Q,-R),\ldots,
b_q^{q-r}sRes_0(Q,-R))$}\\
{\scriptsize$=PmV(b_q^{q-r},
0,\ldots,0,sRes_r(Q,-R),\ldots,
sRes_0(Q,-R))$}\\
{\bf Case 1.1.1: $q>r-1$ and $p-q$ is even and $q-r$ is even.}\\
{\scriptsize$=PmV(sRes_r(Q,-R),\ldots,
sRes_0(Q,-R))=PmV(sRes(Q,-R))$}\\
{\bf Case 1.1.2: $q>r-1$ and $p-q$ is even and $q-r$ is odd.}\\
{\scriptsize$=PmV(b_q,
0,\ldots,0,sRes_r(Q,-R),\ldots,
sRes_0(Q,-R))=PmV(sRes(Q,-R))$}\\
{\bf Case 1.2: $q>r-1$ and $p-q$ is odd.}\\
{\scriptsize$Pmv(sRes(P,Q))$\\
$=PmV(\eps_{p-q}b_q^{p-q},
0,\ldots,0,b_q^{p-r}\eps_{p-q}sRes_r(Q,-R),\ldots,
b_q^{p-r}\eps_{p-q}sRes_0(Q,-R))$\\\hspace*{2em}$+\eps_{p-q}sign(a_p\eps_{p-q}b_q^{p-q})$}\\
{\scriptsize$=PmV(b_q^{p-q},
0,\ldots,0,b_q^{p-r}sRes_r(Q,-R),\ldots,
b_q^{p-r}sRes_0(Q,-R))+sign(a_pb_q)$}\\
{\scriptsize$=PmV(b_q^{q-r},
0,\ldots,0,sRes_r(Q,-R),\ldots,
sRes_0(Q,-R))+sign(a_pb_q)$}\\
where we conclude as in subcases 1.1.1 and 1.1.2.

\smallskip \noindent
{\bf Case 2: $q-1=r$.}\\
In this case using Equation~\ref{eq:sresPqmoinsun},
$sRes_{q-1}(P,Q)=-\eps_{p-q}b-q^{p-q+1}c_r=\eps_{p-q}b_q^{p-r}(-c-r)$\\
where $c_r$ is the leading coefficient of $R$\\
So $Pmv(sRes(P,Q))$\\
$=PmV(a_p,0,\ldots,0,\eps_{p-q}b_q^{p-q},
b_q^{p-r}\eps_{p-q}sRes_{q-1}(Q,-R),\ldots,
b_q^{p-r}\eps_{p-q}sRes_0(Q,-R))$\\
And we conclude as in case 1.
\end{footnotesize}

\qed
\end{proof}

\begin{algorithm2e}
\DontPrintSemicolon

\SetKwFunction{Sign}{Sign}
\SetKwFunction{PmVPol}{PmVPol}
\SetKwFunction{SubResultants}{SubResultants}

{\PmVPol}$(\An,P,p,Q,q)$: an integer\;
\KwIn{$P,Q$, polynomials $\An[X]$ of degree $p$ and $q$ with $q<p$}
\KwOut{$PmV(sRes_p(P,Q)
,\ldots,sRes_0(P,Q))$}
\KwData{$j$ an index, $s_p,\ldots,s_0$ a sequence of signs} 
\KwData{$sReS(P,Q)$ a sequence of items of $\An$}
\BlankLine
\lIf{$q=-\infty$}{\Return 0} \tcp{consistently with Cauchy index definition}
$sRes(P,Q) \leftarrow \SubResultants(\An,P,p,Q,q)$\; 
\tcp{The subresultants computation depends on $\An$}
\tcp{since Algorithm~\ref{algo-subresultant} has an additional assumption.}
\lFor{$j {\bf ~from~} 0 {\bf ~to~ }p$}{$s_j\leftarrow \Sign(\An,sRes_j(P,Q))$}
\Return{$PmV(s_p,\ldots,s_0)$} \tcp{by applying the definition}
\caption{Computing the generalized permanences minus variations}
\label{algo-pmvpol-inD}
\end{algorithm2e}

Algorithm~\ref{algo-pmvpol-inD} describes how to compute the PmV and
so the Cauchy index of two polynomials. Now let us introduce the
Tarski query.
\begin{definition}[Tarski query] Let $P,Q \in \D[X]$. Then:
$$ TaQ(Q,P)=\sum_{z \in Zer(P)} sign(Q(z)).$$
\end{definition}

The Tarski query is closely related to the Cauchy index as established
by the next proposition.
\begin{proposition} Let $P,Q \in \D[X]$. Then:
$$ TaQ(Q,P)=Ind(P'Q/P).$$
\end{proposition}
\begin{proof}
  Let $z$ be a root of $P$ with multiplicity $\mu$. Then
  $P'Q/P=Q(\frac{\mu}{X-z}+R)$
  with $R$ a rational function with no
  pole at $z$. 
  If $Q(z)=0$ then  $P'Q/P$ has no pole in $z$.
  Otherwise  $sign((P'Q/P)(z^+))=sign(Q(z))$ and
  $sign((P'Q/P)(z^-))=-sign(Q(z))$. The assertion of the proposition
  follows.  \qed
\end{proof}

\begin{example}\label{ex:cptTaqQ1P1}
\begin{itemize}
\item For $P_1=\sqrt{5}X^2+1$ and $Q_1= X + \frac{\sqrt{5}-7}2$, we
  have $P_1'=2\sqrt{5}X$.  The sign of $P_1'Q_1$ around the poles of
  $P'_1Q_1/P_1$ is constant: positive around $z_1$ and negative around
  $z_2$.  Hence $Ind(P_1'Q_1/P_1) = \frac12 (-1 - 1 + (-1) -1) = -2$.
  On the other hand, since the sign of $Q_1$ is negative at both $z_1$
  and $z_2$, $TaQ(Q_1,P_1)=-1+ (-1) = -2$.
\item For $P_2=(2\sqrt{5}-1)X^2-1$ and $Q_2=X$, we have
  $P_2'=(4\sqrt{5}-2)X$.  The sign of $P_2'Q_2$ is always
  non-negative, hence it is so at the poles of $P_2'Q_2/P_2$, where it
  is non-zero.  Hence $Ind(P_2'Q_2/P_2) = \frac12 (-1 - 1 + 1 - (-1))
  = 0$ while $Q_2$ has the same sign as the roots of $P_2$, so
  $TaQ(Q_2,P_2)= -1 +1 = 0$.
\end{itemize}
\end{example}

In fact the Tarski question is an auxiliary value. The values we are
really interested in are the following counters:
\begin{itemize}
	\item ${\bf nb}_P(Q)[-1]=|\{z \in Zer(P)\mid Q(z)<0\}|$;
	\item ${\bf nb}_P(Q)[0]=|\{z \in Zer(P)\mid Q(z)=0\}|$.
	\item ${\bf nb}_P(Q)[1]=|\{z \in Zer(P)\mid Q(z)>0\}|$;
\end{itemize}

The following lemma whose proof is obvious is the key for
computing such counters.
\begin{lemma} The Tarski queries and root counters are related by:
\begin{itemize}
	\item $TaQ(1,P)={\bf nb}_P(Q)[-1]+{\bf nb}_P(Q)[0]+{\bf nb}_P(Q)[1]$;
	\item $TaQ(Q,P)=-{\bf nb}_P(Q)[-1] +{\bf nb}_P(Q)[1]$;
	\item $TaQ(Q^2,P)={\bf nb}_P(Q)[-1]+{\bf nb}_P(Q)[1]$.
\end{itemize}
\end{lemma}

\begin{example}
  We previously computed $TaQ(Q_1,P_1)= -2$ (see
  Example~\ref{ex:cptTaqQ1P1}).  The value $TaQ(1,P_1)$, actually
  computed through $Ind(P_1'/P_1)$ yields the number of roots of
  $P_1$, which is $2$.  Finally, computing $TaQ(Q_1^2,P_1)$ can also
  be done through the Cauchy index, and yields the number of roots of
  $P_1$ that are not roots of $Q_1$, in this case also $2$.

  As a result, solving the system induced by the above lemma, there
  are two roots of $P_1$ where $Q_1$ is strictly negative, and no root
  of $P_1$ where $Q_1$ is positive or null. The polynomial $Q_1$ has
  degree $1$, this shows that both roots of $P_1$ are strictly smaller
  than the (only) root of $Q_1$.
\end{example}

Thus defining the invertible matrix ${\bf M}_1$ and vector 
${\bf TaQ}_P(Q)$ by:
$$
{\bf M}_1=
\begin{pmatrix}
1&1 &1\\
-1&0&1\\
1&0&1
\end{pmatrix}
{\bf TaQ}_P(Q)=
\begin{pmatrix}
{\bf TaQ}_P(Q)[0]\\
{\bf TaQ}_P(Q)[1]\\
{\bf TaQ}_P(Q)[2]
\end{pmatrix}
=
\begin{pmatrix}
TaQ(Q^0,P)\\
TaQ(Q^1,P)\\
TaQ(Q^2,P)
\end{pmatrix}
$$

we obtain:
\begin{proposition}
\label{prop:simpleTaQ}
$${\bf TaQ}_P(Q)= {\bf M}_1 \cdot {\bf nb}_P(Q)$$
\end{proposition}

As we are interested in determining the simultaneous signs of polynomials
evaluated on the roots of another polynomial
we generalize mappings ${\bf nb}_P$ and ${\bf TaQ}_P$
to a sequence of polynomials.

\begin{definition}[Generalized counters and Tarski queries] Let $P \in \D[X]$
and $\mathcal Q=(Q_1,\ldots,Q_m)$ be a finite sequence of $\D[X]$. Then:

\noindent
${\bf nb}_P(\mathcal Q)$ is an integer vector whose support
is $\{-1,0,1\}^{\{1,\ldots,m\}}$ such that: 
$${\bf nb}_P(\mathcal Q)[i_1,\ldots,i_m]=
\left| \{z \in Zer(P)\mid \forall j\leq m\ sign(Q_j(z))=i_j\}\right|$$

\noindent
${\bf TaQ}_P(\mathcal Q)$ is an integer vector whose support
is $\{0,1,2\}^{\{1,\ldots,m\}}$ such that: 
$${\bf TaQ}_P(\mathcal Q)[i_1,\ldots,i_m]=TaQ(Q_1^{i_1}\cdots Q_m^{i_m})$$
\end{definition}

The tensor product of two matrices ${\bf A}$ of dimension
$m_a\times n_a$ and ${\bf B}$ of dimension $m_b\times n_b$ is the
matrix ${\bf A} \otimes {\bf B}$ of dimension $m_am_b \times n_an_b$
defined by: ${\bf A} \otimes {\bf B}[(i_a,i_b),(j_a,j_b)]={\bf A}[i_a,j_a]{\bf
  B}[i_b,j_b]$.  We inductively define for $t>1$, ${\bf M}_t={\bf M}_1
\otimes {\bf M}_{t-1}$.

\begin{proposition}
 Let $P \in \D[X]$
and $\mathcal Q=(Q_1,\ldots,Q_m)$ a finite sequence of $\D[X]$. Then:
$${\bf TaQ}_P(\mathcal Q)= {\bf M}_m \cdot {\bf nb}_P(\mathcal Q).$$
\end{proposition}
\begin{proof}

Observe that both ${\bf TaQ}_P(\mathcal Q)$ and ${\bf nb}_P(\mathcal Q)$
only depend on $Zer(P)$. Thus w.l.o.g we assume that $P=\prod_i X-z_i$
with all $z_i$ distinct. In this case, 
\[{\bf TaQ}_P(\mathcal Q)=\sum_i{\bf TaQ}_{X-z_i}(\mathcal Q)
\textrm{ and }
{\bf nb}_P(\mathcal Q)=\sum_i{\bf nb}_{X-z_i}(\mathcal Q).\]

\noindent
So we are left with the case $P = X-z$. 
For all $(i_1,\ldots,i_m)$,
\begin{eqnarray*}
{\bf TaQ}_P(\mathcal Q)[i_1,\ldots,i_m] &=& TaQ_P(Q_1^{i_1}\cdots Q_m^{i_m}) \\
&=& sign(Q_1^{i_1}(z)\cdots Q_m^{i_m}(z))\\
&=& \prod_j sign(Q_j^{i_j}(z))
\ =\ \prod_{j} TaQ_P(Q_j^{i_j})
\end{eqnarray*}
Therefore by definition of tensor product,\\
${\bf TaQ}_P(\mathcal Q) =
{\bf TaQ}_P(Q_1)\otimes \cdots\otimes {\bf TaQ}_P(Q_m)$.

\noindent
On the other hand, for all $(i_1,\ldots,i_m)$,
${\bf nb}_P(\mathcal Q)[i_1,\ldots,i_m]={\bf 1}_{\bigwedge_j sign(Q_j(z))=i_j}$
$=\prod_j {\bf 1}_{sign(Q_j(z))=i_j}$ $=\prod_j {\bf nb}_P(\mathcal Q_j)[i_j]$.
Therefore,\\
${\bf nb}_P(\mathcal Q) =
{\bf nb}_P(Q_1)\otimes \cdots\otimes {\bf nb}_P(Q_m)$.

\noindent
So ${\bf TaQ}_P(\mathcal Q) =
{\bf TaQ}_P(Q_1)\otimes \cdots\otimes {\bf TaQ}_P(Q_m)$\\
$= {\bf M_1}\cdot{\bf nb}_P(Q_1)\otimes \cdots\otimes {\bf M_1}\cdot{\bf nb}_P(Q_m)$
\emph{using Proposition~\ref{prop:simpleTaQ} }\\
$= ({\bf M_1}\otimes \cdots\otimes {\bf M_1})\cdot({\bf nb}_P(Q_1)\otimes \cdots\otimes {\bf nb}_P(Q_m))$
\emph{using a property of tensor product}\\
$= {\bf M_m} \cdot {\bf nb}_P(\mathcal Q)$.

\qed
\end{proof}

Using elementary properties of the tensorial product, one gets
the following corollary.

\begin{corollary}
  Let $P \in \D[X]$ and $\mathcal Q=(Q_1,\ldots,Q_m)$ a finite
  sequence of $\D[X]$. Then:
$$ {\bf nb}_P(\mathcal Q)= ({\bf M}_m)^{-1} \cdot{\bf TaQ}_P(\mathcal Q)
= \left(({\bf M}_1)^{-1}\otimes \cdots \otimes ({\bf M}_1)^{-1}\right) 
\cdot{\bf TaQ}_P(\mathcal Q). $$
\end{corollary}

While the previous corollary provides a way to compute
the number of zeroes of $P$ per sign realization at
family $\mathcal Q$, the procedure is highly inefficient
w.r.t $m$.
Indeed  $M_m$ has size $3^m \times 3^m$ while the values and the size
of the support of vector ${\bf nb}_P(\mathcal Q)$ remain bounded
by the number of zeroes of $P$. So in the next paragraphs,
we refine the procedure by iteratively computing
${\bf nb}_P(Q_i,\ldots,Q_m)$ by decreasing values of $i$
and using the intermediate result
to reduce the size of the matrix to be inverted at the next computation step.

\begin{definition}
Let $m$ be an integer, $\Sigma \subseteq \{-1,0,1\}^m$ 
and $A \subseteq \{0,1,2\}^m$. Then $A$ is \emph{adapted} to
$\Sigma$ if the (sub)matrix $M_m[A,\Sigma]$ is invertible.
\end{definition}

Since $M_m$ is invertible any $\Sigma$ admits some $A$.
However we need a way to efficiently compute such an $A$.

\begin{definition}\label{def:aSigma}
Let $\Sigma \subseteq \{-1,0,1\}^m$. Then $A(\Sigma)$
is inductively defined by:
\begin{itemize}
 \item If $m=1$ then:
 \begin{enumerate}
  \item When $|\Sigma|=1$, $A(\Sigma)=\{0\}$ 
  \item When $|\Sigma|=2$, $A(\Sigma)=\{0,1\}$ 
  \item When $|\Sigma|=3$, $A(\Sigma)=\{0,1,2\}$ 
 \end{enumerate}
 \item Let $\Sigma \subseteq \{-1,0,1\}^{m+1}$.\\ 
 For $k \in \{1,2,3\}$,
 define $\Sigma_k=\{\sigma \in \{-1,0,1\}^{m} \mid |\{(i,\sigma) \in \Sigma\}|\geq k\}$.\\
 Then $A(\Sigma)=\{0\}\times A(\Sigma_1)\cup\{1\}\times A(\Sigma_2)\cup\{2\}\times A(\Sigma_3)$.
\end{itemize}
\end{definition}
Observe that $\Sigma_3\subseteq \Sigma_2 \subseteq \Sigma_1$
and that $|\Sigma_3|+|\Sigma_2|+|\Sigma_1|=|\Sigma|$.

\begin{proposition}\label{prop:aSigmaAdapted}
Let $\Sigma \subseteq \{-1,0,1\}^m$. Then $A(\Sigma)$
is adapted to $\Sigma$.
\end{proposition}
\begin{proof}
The base case $m=1$ is established by 
a straightforward examination of $M_1$.
Assume that the result holds for $m$
and consider $\Sigma \subseteq \{-1,0,1\}^{m+1}$.
For $\sigma \in \Sigma_1$, 
we denote by $C_\sigma$ the column of matrix $M_m[ \{0,1,2\}^m,\Sigma_1]$
indexed by $\sigma$.
Then columns of matrix $M_{m+1}[\{0,1,2\}^{m+1},\Sigma]$ are:  

- $C_{(-1,\sigma)}=\left( \begin{array}{c} 1 \\-1\\1\end{array}\right)\otimes C_\sigma =
\left( \begin{array}{c} C_\sigma \\ -C_\sigma \\C_\sigma \end{array}\right)$ if $(-1, \sigma) \in \Sigma$,\\

- $C_{(0,\sigma)}=\left( \begin{array}{c} 1 \\0\\0 \end{array}\right)\otimes C_\sigma =
\left( \begin{array}{c} C_\sigma \\ 0 \\0\end{array}\right)$ if $(0, \sigma) \in \Sigma$,\\

- $C_{(1,\sigma)}=\left( \begin{array}{c} 1 \\1\\1\end{array}\right)\otimes C_\sigma =
\left( \begin{array}{c} C_\sigma \\ C_\sigma \\C_\sigma \end{array}\right)$ 
if $(1, \sigma) \in \Sigma$.

\noindent
For $\sigma \in \Sigma_1$, 
we pick a minimal $k_{\sigma,1}$ such that
$(k_{\sigma,1},\sigma) \in \Sigma$.
For $\sigma \in \Sigma_2$, 
we pick a minimal $k_{\sigma,2}\neq k_{\sigma,1}$ such that
$(k_{\sigma,2},\sigma) \in \Sigma$.
For $\sigma \in \Sigma_3$, 
we pick the unique $1=k_{\sigma,3} \notin
\{k_{\sigma,1},k_{\sigma,2}\}$ such that
$(k_{\sigma,3},\sigma) \in \Sigma$.
Let us reorder the columns of matrix $M_{m+1}[\{0,1,2\}^{m+1},\Sigma]$
as follows. The first $|\Sigma_1|$ columns are those indexed by all
$(k_{\sigma,1},\sigma) \in \Sigma$.
The next $|\Sigma_2|$ columns are those indexed by all
$(k_{\sigma,2},\sigma) \in \Sigma$.
The last $|\Sigma_3|$ columns are those indexed by all
$(k_{\sigma,3},\sigma) \in \Sigma$.

\noindent
We then perform on this matrix some columns operations that
let the linear independence status of rows unchanged:
\begin{itemize}
 \item when $k_{\sigma,1}=-1$ and $k_{\sigma,2}=0$
 then $C_{0,\sigma}\leftarrow C_{0,\sigma} -C_{-1,\sigma}$
 so that\\ $C_{0,\sigma}=\left( \begin{array}{c} 0 \\
 C_\sigma \\-C_\sigma \end{array}\right)$.
 \item when $k_{\sigma,1}=-1$ and $k_{\sigma,2}=1$
 then $C_{1,\sigma}\leftarrow \frac{1}{2}(C_{1,\sigma} -C_{-1,\sigma})$
 so that\\ $C_{1,\sigma}=\left( \begin{array}{c} 0 \\
 C_\sigma \\0 \end{array}\right)$.
 \item when $k_{\sigma,1}=0$ and $k_{\sigma,2}=1$
 then $C_{1,\sigma}\leftarrow C_{1,\sigma} -C_{-1,\sigma}$
 so that\\ $C_{1,\sigma}=\left( \begin{array}{c} 0 \\
 C_\sigma \\C_\sigma \end{array}\right)$.
 \item when $k_{\sigma,3}$ is defined (and so equal to 1)
 then $C_{1,\sigma}\leftarrow \frac{1}{2}(C_{1,\sigma} 
 -2C_{0,\sigma}+C_{-1,\sigma})$
 so that\\ $C_{1,\sigma}=\left( \begin{array}{c} 0 \\
 0 \\C_\sigma \end{array}\right)$.
\end{itemize}
The resulting matrix has a triangular form~:
$$\left( 
\begin{array}{ccc} M_m[ \{0,1,2\}^m, \Sigma_1] & 0 & 0 \\
* & M_m[ \{0,1,2\}^m, \Sigma_2] & 0\\
* & * & M_m[ \{0,1,2\}^m,\Sigma_3]
\end{array}\right) 
$$
Due to this triangular form, the first 
$|\Sigma_1|+|\Sigma_2|+|\Sigma_3|$ independent rows of 
$M_{m+1}[\{0,1,2\}^{m+1},\Sigma]$  are the first
$|\Sigma_1|$ rows of the first diagonal block 
followed by  the first
$|\Sigma_2|$ rows of the second diagonal block
and the first
$|\Sigma_3|$ rows of the third diagonal block.
\qed
\end{proof}

Computing inductively $A(\Sigma)$ seems to require three
``recursive calls''. However observing
that $\Sigma_3\subseteq\Sigma_2\subseteq\Sigma_1$ and using
the next proposition we will obtain an efficient computation.

\begin{proposition}
\label{prop:lin-ind-extract}
 Let $\Sigma'\subseteq \Sigma \subseteq \{-1,0,1\}^m$. 
 Then $A(\Sigma')$ is obtained by extracting
 the first $|\Sigma'|$ linearly independent rows
 of matrix $M_m[A(\Sigma),\Sigma']$.
\end{proposition}

 
\begin{proof}
We proceed by induction on $m$. The base case $m=1$
is an immediate consequence of the definition of $A(\Sigma)$.

\smallskip \noindent
Assume that result holds for $m$
and consider $\Sigma'\subseteq \Sigma \subseteq \{-1,0,1\}^{m+1}$.
Define as in Definition~\ref{def:aSigma}, 
$\Sigma'_1, \Sigma'_2$ and $\Sigma'_3$.
Observe that for all $i$, $\Sigma'_i\subseteq \Sigma_i$. 
Consider matrix $M_{m+1}[\{-1,0,1\}^{m+1},\Sigma']$.
After performing the same linear transformations
on the columns as those of the previous proof, 
we obtain the following matrix:
$$\left( 
\begin{array}{ccc} M_m[ \{0,1,2\}^m, \Sigma'_1] & 0 & 0 \\
* & M_m[ \{0,1,2\}^m, \Sigma'_2] & 0\\
* & * & M_m[ \{0,1,2\}^m,\Sigma'_3]
\end{array}\right) 
$$
Thus the first maximal set of independent rows of this
matrix will be obtained by the first maximal
sets of independent rows in the three diagonal blocks.
Applying the induction hypothesis, this
corresponds to the first maximal set of independent rows of
the following matrix:
$$\left( 
\begin{array}{ccc} M_m[ A(\Sigma_1), \Sigma'_1] & 0 & 0 \\
* & M_m[ A(\Sigma_2), \Sigma'_2] & 0\\
* & * & M_m[ A(\Sigma_3),\Sigma'_3]
\end{array}\right) 
$$
which (by the inverse linear transformations)
is equivalent to looking for the first $|\Sigma'|$ linearly independent rows
of matrix $M_m[A(\Sigma),\Sigma']$.

\qed
\end{proof}

Algorithm~\ref{algo-sign-realization} implements the whole method
developped above.

\begin{algorithm2e}
\DontPrintSemicolon

\SetKwFunction{Sign}{Sign}
\SetKwFunction{Degree}{Degree}
\SetKwFunction{SignRealization}{SignRealization}
\SetKwFunction{IntRem}{IntRem}
\SetKwFunction{PmVPol}{PmVPol}

{\SignRealization}$(\D,P,p,\mathcal Q)$: a non null vector with its support\;
\KwIn{$P$, a non null polynomial in $\D[X]$ with degree $p$}
\KwIn{$\mathcal Q=\{(Q_1,q_1),\ldots,(Q_m,q_m)\}$, a family of
non null polynomials in $\D[X]$}
\BlankLine
\KwOut{the vector counting the sign realizations for $\mathcal Q$  by the roots of $P$}
\BlankLine
\KwData{$e_1,\ldots,e_m$ degrees in $\{0,1,2\}$}
\KwData{${\bf TaQ}$ a vector indexed by vectors of degrees} 
\KwData{${\bf nb}$ a vector indexed by vectors of signs} 
\KwData{$R$, a polynomial in $\D[X]$, $r$ a degree}
\KwData{${\bf M}$, an integer matrix}
\KwData{$extA, A,A_1,A_2,A_3$, sets of vectors of degrees}
\KwData{$ext\Sigma, \Sigma,\Sigma_1,\Sigma_2,\Sigma_3$, sets of vectors of signs}
\BlankLine
 \For {$e_m {\bf~in~} \{0,1,2\}$}
   {
     $R \leftarrow P'Q_m^{e_m}$;
     $(R,r) \leftarrow \IntRem(\D,R,e_mq_m+p-1,P,p)$\;
     \tcp{see Algorithm~\ref{algo-intrem} for \IntRem} 
     ${\bf TaQ}[e_m]\leftarrow \PmVPol(\D,P,p,R,r)$\;   
   }
   ${\bf nb}\leftarrow {\bf M}_1^{-1}\cdot{\bf TaQ}$\;
   \lIf( // $P$ has no roots){${\bf nb}={\bf 0}$}{\Return $\emptyset,-$}
   $\Sigma \leftarrow supp({\bf nb})$\;
   \lIf {$|\Sigma|=1$}
   {$A \leftarrow \{0\}$}
   \lElseIf {$|\Sigma|=2$}
   {$A \leftarrow \{0,1\}$}
   \lElse
   {$A \leftarrow \{0,1,2\}$}
   ${\bf nb}\leftarrow {\bf nb}_{|\Sigma}$;
   ${\bf M}\leftarrow {\bf M}_{1|A\times \Sigma}$\;
\For {$i {\bf~from~} m-1 {\bf~downto~} 1$}
{
   $ext\Sigma \leftarrow \{-1,0,1\} \times \Sigma$; 
   $extA \leftarrow \{0,1,2\} \times A$;
   ${\bf extM}\leftarrow {\bf M}_1 \otimes {\bf M}$\; 
   \For {$(e_i,\ldots,e_m) {\bf~in~} extA$}
   {
   
     $R \leftarrow P'\prod_{i\leq j\leq m} Q_j^{e_j}$;
     $(R,r) \leftarrow \IntRem(\D,R,\sum_{i\leq j\leq m}q_je_j,P,p)$\;
     ${\bf TaQ}[(e_i,\ldots,e_m)]\leftarrow \PmVPol(\D,P,p,R,r)$\;   
   }
   ${\bf nb}\leftarrow {\bf extM}^{-1}\cdot{\bf TaQ}$\;
   $\Sigma_1\leftarrow \Sigma$\;
   $\Sigma_2\leftarrow \{\sigma \in \Sigma \mid 
   |\{(i,\sigma) \in supp({\bf nb})\}|\geq 2\}$;
   $\Sigma_3\leftarrow \{\sigma \in \Sigma \mid 
   |\{(i,\sigma) \in supp({\bf nb})\}|\geq 3\}$\;
   $A_1\leftarrow A$\;
   $A_2\leftarrow$ the indexes of the first $|\Sigma_2|$ linearly independent rows of 
   ${\bf M}_{|A\times \Sigma_2}$\;
   $A_3\leftarrow$ the indexes of the first $|\Sigma_3|$ linearly independent rows of 
   ${\bf M}_{|A_2\times \Sigma_3}$\;
   $\Sigma\leftarrow supp({\bf nb})$\;
   $A \leftarrow \{0\}\times A_1 \cup \{1\}\times A_2
   \cup \{2\}\times A_3$\;
   ${\bf nb}\leftarrow {\bf nb}_{|\Sigma}$;
   ${\bf M}\leftarrow {\bf extM}_{|A\times \Sigma}$\;   
}
\Return $\Sigma, {\bf nb}$
\caption{Computing sign realizations of family $\mathcal Q$ at roots of $P$}
\label{algo-sign-realization}
\end{algorithm2e}

\paragraph*{Defining and computing encodings for roots}

\begin{definition}[Thom-encoding] Let $P \in \D[X]$ with 
$deg(P)=p>0$ and $x\in \R$. The $P$-encoding of $x$
is the vector: 
$$\sigma_P(x)=(sign(P(x)),sign(P'(x)),\ldots, sign(P^{(p)}(x))).$$
\end{definition}

A $P$-code is a vector of signs indexed by $\{0,\ldots,deg(P)\}$.

\begin{proposition} 
\label{prop:Thom}
Let $P \in \D[X]$ and $\sigma$ be a $P$-code.
Then:
\begin{itemize}
 \item $\sigma_P^{-1}(\sigma)$ is either empty, a point or an open interval.
 \item Let $x\neq x'$ be two roots of $P$. Then $\sigma_P(x)\neq \sigma_P(x')$.
 \item Let $x,x'$ with $\sigma_P(x)\neq \sigma_P(x')$. Then $x<x'$ if
   and only if, denoting $k$ the largest index with
   $\sigma_P(x)[k]\neq \sigma_P(x')[k]$:
 	\begin{enumerate}
        \item either $\sigma_P(x)[k+1]=1$ and $\sigma_P(x)[k]<
          \sigma_P(x')[k]$;
        \item or $\sigma_P(x)[k+1]=-1$ and $\sigma_P(x)[k]>
          \sigma_P(x')[k]$.
	\end{enumerate}
\end{itemize}
\end{proposition}
\begin{proof}
We proceed by induction on the degree of $P$.
The case $deg(P)=1$ is obvious. Assume that it is valid
for all $P$ such that $deg(P)\leq i$. Consider $P$
with $deg(P)=i+1$. Apply the inductive hypothesis on
$\sigma$ restricted to its $i$ last components, denoted $\sigma'$, 
and on $P'$. When $\sigma_{P'}^{-1}(\sigma')$ is empty or a point
then the result is immediate. When $\sigma_{P'}^{-1}(\sigma')$
is an interval, then $\sigma[1]\neq 0$. Thus $P(x)$ is a strictly
monotonous function on the interval which meets
0 at most once. This implies the result.

\smallskip \noindent
The second assertion is a direct consequence of the first assertion.

\smallskip \noindent Considering the third assertion,
$\sigma_{P^{(k+1)}}(x)= \sigma_{P^{(k+1)}}(x')$.
Since $x\neq x'$, the second assertion implies 
that $\sigma_{P^{(k+1)}}(x)\neq 0$.\\
Since $P^{(k+1)}$ is constant in $[\min(x,x'),\max(x,x')]$, this
implies the third assertion.

\qed
\end{proof}
\begin{algorithm2e}
\DontPrintSemicolon

\SetKwFunction{RootCoding}{RootCoding}
\SetKwFunction{PmVPol}{PmVPol}
\SetKwFunction{SignRealization}{SignRealization}

{\RootCoding}$(\D,P,p,Q,q)$: a list\;
\KwIn{$P,Q$, non null polynomials in $\D[X]$ with respective degrees $p,q$}
\KwOut{a list of the $Q$-encoding of roots of $P$}
\KwData{ $(s_0,\ldots,s_q)$ a vector of signs}
\KwData{${\bf nb}$ a vector indexed by vectors of signs} 
\KwData{$\Sigma$ a set of vectors of signs} 
\BlankLine
$(\Sigma,{\bf nb})\leftarrow \SignRealization(\D,P,p,\{(Q^{(0)},q),\ldots,(Q^{(q)},0)\})$\;
Order the $Q$-encodings $(s_0,\ldots,s_q)$ 
of the support $\Sigma$ of {\bf nb}\; 
using Proposition~\ref{prop:Thom} and duplicating them w.r.t. ${\bf nb}[(s_0,\ldots,s_q)]$\;
\Return this list of encodings
\caption{Computing the $Q$-encoding of roots of $P$}
\label{algo-rootcoding-inD}
\end{algorithm2e}

\begin{example}
  Let us consider the $P_1$-encoding of reals for $P_1=\sqrt{5}X^2-1$.
  First remark that the second derivative is always positive, hence
  the third component of the $P_1$-encoding of any real number is
  always $+1$.  This encoding divides the real line into seven
  intervals:
\begin{itemize}[label=\textbullet]
\item $]-\infty,-\frac1{\sqrt[4]{5}}[$ is encoded into $(+1,-1,+1)$,
  since for $x$ in this interval, $P(x)$ is positive but decreasing.
\item The first root $[-\frac1{\sqrt[4]{5}},-\frac1{\sqrt[4]{5}}]$ is
  encoded into $(0,-1,+1)$.
\item $]-\frac1{\sqrt[4]{5}},0[$ corresponds to $(-1,-1,+1)$.
\item The point $[0,0]$ is encoded by $(-1,0,+1)$.
\item $]0,\frac1{\sqrt[4]{5}}[$ corresponds to $(-1,+1,+1)$.
\item The second root $[\frac1{\sqrt[4]{5}},\frac1{\sqrt[4]{5}}]$ is
  encoded into $(0,+1,+1)$.
\item $]\frac1{\sqrt[4]{5}},+\infty[$ is encoded into $(+1,+1,+1)$.
\end{itemize}
\end{example}

As a consequence of our previous developments, we are now in position
to perform two main computations 
in $\D[X]$: (1) determining the
number of roots of a polynomial $P$ and computing their $P$-encoding,
and (2) computing the $Q$-encoding of roots of a polynomial $P$.  Both
results are obtained by Algorithm~\ref{algo-rootcoding-inD}.  For the
first goal it is sufficient to call
${\tt PmVPol}(P,P')$ and if the result is non null
to call ${\tt RootCoding}(P,P)$.

\subsection{Triangular systems}
\label{subsec:triangular}

While we only stated the effective properties of (a representation of)
$\D$ in the previous parts, we now consider specific representations
of real subrings of the form $\D = \Q[\alpha_1, \ldots,
\alpha_{\ell}]$ where the $\alpha_i$'s are real algebraic numbers.
Such representations are called triangular systems and we will show
(in Proposition~\ref{prop:sign-eff}) that they are sign-effective.  In
the sequel, the leading coefficient of $P =\sum_{i\leq p}a_iX^i$ in
$\D[X]$ with $deg(P)=p$ is denoted $lcof(P)= a_p$.  Note that the
leading coefficient of a polynomial $P$ in $\Q[X_1, \ldots,
X_{i-1}][X_i]$ is itself a polynomial in $\Q[X_1, \ldots,X_{i-1}]$.

\begin{definition}[Triangular system]
Let $((n_i,P_i,p_i))_{i=1}^{\ell}$ such that 
for all $i$, $n_i$ is a positive integer and
$P_i \in \Q[X_1,\ldots,X_{i-1}][X_i]$ with $deg(P_i)=p_i>0$.
Let $(\alpha_1,\ldots,\alpha_{\ell})$ be a sequence of reals.
Then $((n_i,P_i,p_i))_{i=1}^{\ell}$ is a 
triangular system of level $\ell$
for $(\alpha_1,\ldots,\alpha_{\ell})$ if:
\begin{itemize}
	\item $\alpha_1$ is the $n_1^{th}$ root of $P_1$ whose degree
	is $p_1$;
	\item For $1\leq i< \ell$, 
	$P_{i+1}(\alpha_1,\ldots,\alpha_i)$ has degree $p_i$
	and $\alpha_{i+1}$ is the $n_{i+1}^{th}$ root of 
        polynomial $P_{i+1}(\alpha_1,\ldots,\alpha_{i})
        \in \Q[\alpha_1,\ldots,\alpha_{i}][X_{i+1}]$.
\end{itemize}
\end{definition}

\begin{algorithm2e}
\DontPrintSemicolon

\SetKwFunction{Sign}{Sign}
\SetKwFunction{PmVPol}{PmVPol}

{\PmVPol}$(\ell,\mathcal T,P,p,Q,q)$: an integer\; 
\KwIn{$\ell$, the current level} 
\KwIn{$\mathcal T=\{(n_i,P_i,p_i)\}_{i=1}^{\ell}$ a
  triangular system for $(\alpha_1,\ldots,\alpha_{\ell})$}
\KwIn{$P,Q$, polynomials $\Q[X_1,\ldots,X_{\ell}][X_{\ell+1}]$ of
  degree $p$ and $q$ with $q<p$ such that
  $P(\alpha_1,\ldots,\alpha_{\ell})\neq 0$ and
  $Q(\alpha_1,\ldots,\alpha_{\ell})\neq 0$ when $q\geq 0$}
\KwOut{$PmV(sRes_p(P(\alpha_1,\ldots,\alpha_{\ell}),Q(\alpha_1,
\ldots,\alpha_{\ell})),$
  \hspace*{2cm}$\ldots,sRes_0(P(\alpha_1,\ldots,\alpha_{\ell}),
Q(\alpha_1,\ldots,\alpha_{\ell})))$}
\KwData{$j$ an index, $s_p,\ldots,s_0$ a sequence of signs} \BlankLine
\lIf( // consistently with Cauchy index
definition){$q=-\infty$}{\Return 0} $sRes(P,Q) \leftarrow
\SubResultants(\Q[X_1,\ldots,X_{\ell}],P,p,Q,q)$ \tcp{using
  Algorithm~\ref{algo-subresultant}} \lFor{$j {\bf ~from~} 0 {\bf ~to~
  }p$}{$s_j\leftarrow \Sign(\ell,\mathcal T,sRes_j(P,Q))$}
\Return{$PmV(s_p,\ldots,s_0)$} \tcp{by applying the definition}\;
\caption{Computing $PmV$ in triangular systems}
\label{algo-pmvpol}
\end{algorithm2e}

By convention, a triangular system of level 0 is the empty sequence.
Observe that \emph{a priori} we do not know how to decide
whether $((n_i,P_i,p_i))_{i=1}^{\ell}$  is a triangular system
for some sequence of reals. Given a triangular system
$((n_i,P_i,p_i))_{i=1}^{\ell}$, a representation of an item
of $\Q[\alpha_1,\ldots, \alpha_{\ell}]$ is nothing
but some polynomial $P \in \Q[X_1,\ldots,X_{l}]$
denoting $P(\alpha_1,\ldots,\alpha_{\ell})$.

\begin{example}
  The system $((2,X_1^2-X_1-1,2),(1,(2X_1-1)X_2^2-1,2))$ is a
  triangular system for the reals
  $(\frac{1+\sqrt5}2,-\frac{1}{\sqrt[4]5})$.  Indeed, polynomial
  $X_1^2-X_1-1$ has two roots $\frac{1-\sqrt5}2 < \frac{1+\sqrt5}2$.
  In addition, when $X_1 = \frac{1+\sqrt5}2$, polynomial
  $(2X_1-1)X_2^2-1$ becomes $P_1=\sqrt5 X_2^2-1$, with two roots
  $-\frac{1}{\sqrt[4]5} < \frac{1}{\sqrt[4]5}$.
\end{example}

\begin{proposition}\label{prop:sign-eff}
  Let $\ell\geq 0$ and $((n_i,P_i,p_i))_{i=1}^{\ell}$ such that for
  all $i$, $n_i$ is a positive integer and $P_i \in
  \Q[X_1,\ldots,X_{i-1}][X_i]$ with $deg(P_i)=p_i>0$.  Then we can
  decide whether $((n_i,P_i,p_i))_{i=1}^{\ell}$ is a triangular system
  for some $\{\alpha_i\}_{i=1}^{\ell}$.  Furthermore with this
  representation, the rings $\Q[\alpha_1,\ldots,\alpha_{\ell}]$ and
  $\Z[\alpha_1,\ldots,\alpha_{\ell}]$ are sign-effective.
\end{proposition}
\begin{proof}
The proof is done by induction on $\ell$. The base case $\ell=0$ corresponds
to the case where the ring is $\Q$ or $\Z$ and so there is nothing to prove.

\noindent
For the inductive case, 
in order to check whether $((n_i,P_i,p_i))_{i=1}^{\ell+1}$ 
is a triangular system, we first check that
$((n_i,P_i,p_i))_{i=1}^{\ell}$ is a triangular system.
In the positive case $\Q[\alpha_1,\ldots,\alpha_{\ell}]$
is sign-effective so that we can check whether
$P_{\ell+1}(\alpha_1,\ldots,\alpha_{\ell})$ has degree $p_{\ell+1}$
and compute the number of roots of 
$P_{l+1}(\alpha_1,\ldots,\alpha_{\ell})$
by using ${\tt PmvPol}(l,\mathcal T,P_{l+1},p_{l+1},P'_{\ell+1},p_{\ell+1}-1)$ 
in $\Q[\alpha_1,\ldots,\alpha_{\ell}]$.
We have rewritten the corresponding algorithm 
(see Algorithm~\ref{algo-pmvpol})
in order to exploit the representation 
provided by Algorithm~\ref{algo-subresultant}.

\noindent
Assume that $((n_i,P_i,p_i))_{i=1}^{\ell+1}$ 
is a triangular system. Again using induction hypothesis
$\Q[\alpha_1,\ldots,\alpha_{\ell}]$ is sign-effective. 
So in addition to sign determination in 
$\Q[\alpha_1,\ldots,\alpha_{\ell}]$,
we are also able to compute {\tt Degree} and {\tt RootCoding}  
in this ring. Thus Algorithm~\ref{algo-sign} (applied at level $l+1$)
determines the sign of
$P(\alpha_1,\ldots,\alpha_{\ell+1})$ by computing the degree
of $P$ in $\Q[\alpha_1,\ldots,\alpha_{\ell}]$
and then determining the $P$-encodings
of roots of $P_{\ell+1}$ in $\Q[\alpha_1,\ldots,\alpha_{\ell}]$ and returning
the sign of $P$ corresponding to the $n_{\ell+1}^{th}$ root.

\qed
\end{proof}

The sign determination is then obtained by a set of mutually recursive
functions.  In order to clarify their behavior we have represented
their calls in Figure~\ref{fig:functioncalls}.

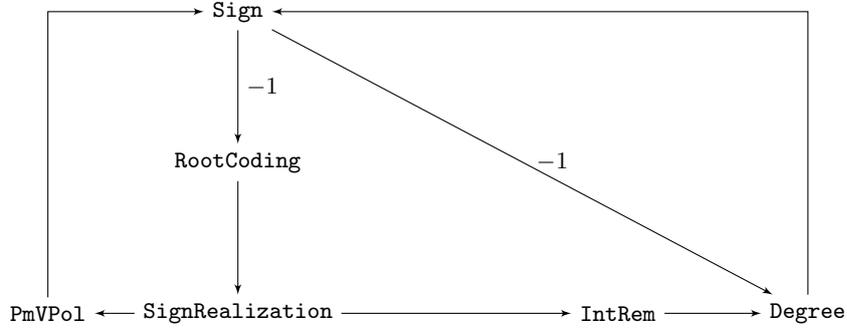
\begin{figure}[htbp]
\begin{center}
\begin{tikzpicture}[xscale=1,yscale=1]

\path (2.5,0) node[] (r) {{\tt RootCoding}};
\path (10,-2) node[] (d) {{\tt Degree}};
\path (7.5,-2) node[] (i) {{\tt IntRem}};
\path (2.5,2) node[] (s) {{\tt Sign}};
\path (2.5,-2) node[] (sr) {{\tt SignRealization}};
\path (0,-2) node[] (p) {{\tt PmVPol}};

 \draw[arrows=-latex'] (i) -- (d) node[pos=.5,above] {};
 \draw[arrows=-latex'] (r) -- (sr) node[pos=.5,above] {};
 \draw[arrows=-latex'] (sr) -- (i) node[pos=.5,above] {};
 \draw[arrows=-latex'] (sr) -- (p) node[pos=.5,below] {};
 \draw[arrows=-latex'] (d) -- (10,2) -- (s) node[pos=.5,above] {};
 \draw[arrows=-latex'] (p) -- (0,2) -- (s) node[pos=.5,below] {};

  \draw[arrows=-latex'] (s) -- (r) node[pos=.5,right] {$-1$};
  \draw[arrows=-latex'] (s) -- (d) node[pos=.5,right] {$\:-1$};

\end{tikzpicture}
\caption{Links between function calls with level $\ell$ changing.}
\label{fig:functioncalls}
\end{center}
\end{figure}

\begin{algorithm2e}
\DontPrintSemicolon

\SetKwFunction{Sign}{Sign}
\SetKwFunction{Degree}{Degree}
\SetKwFunction{RootCoding}{RootCoding}

{\Sign}$(\ell,\mathcal T,P)$: a sign\;
\KwIn{$P$, a polynomial in $\Q[X_1,\ldots,X_{\ell}]=
\Q[X_1,\ldots,X_{\ell-1}][X_{\ell}]$}
\KwIn{$\ell$, the current level}
\KwIn{$\mathcal T=\{(n_i,P_i,p_i)\}_{i=1}^{\ell}$  a 
triangular system
for $(\alpha_1,\ldots,\alpha_{\ell})$}
\KwOut{the sign of $P(\alpha_1,\ldots,\alpha_{\ell})$}
\KwData{$\Sigma$ a list of sign vectors}
\BlankLine
\lIf( // $P$ is a rational){${\ell}=0$}{\Return{$\Sign(\Q,P)$}}
$p \leftarrow \Degree({\ell}-1,\mathcal T_{\downarrow {\ell}-1},P)$\;
\tcp{$\mathcal T_{\downarrow {\ell}-1}$ is the restriction of $\mathcal T$ 
at level ${\ell}-1$}
\lIf {$p=-\infty$}{\Return 0}
$\Sigma=\RootCoding({\ell}-1,\mathcal T_{\downarrow {\ell}-1},P_{\ell},p_{\ell},P,p)$\;
Let $\vect{v}$ be the $n_{\ell}^{th}$ item of $\Sigma$\;
\Return{$\vect{v}[0]$}\;
\caption{Determining the sign in a triangular system.}
\label{algo-sign}
\end{algorithm2e}

\subsection{Building a cylindrical algebraic decomposition}
\label{subsec:decomposition}


We have the following result~\cite{Collins75}:


\begin{theorem}
  For every finite family of sets of polynomials $\P=\{\P_i\}_{i\leq
    n}$ such that $\P_i \subseteq \Q[X_1, \ldots, X_i]$, one can build
  a cylindrical algebraic decomposition of $\R^n$ adapted to $\P$ in
  2EXPTIME.
\end{theorem}

We devote the rest of the subsection to the proof of this theorem. The
algorithm that builds the cylindrical algebraic decomposition of
$\R^n$ proceeds in two steps: the \emph{elimination step} and the
\emph{lifting step}. The elimination step ensures the existence of a
cylindrical algebraic decomposition while enlarging the set of polynomials of
polynomials $\P_i$.  Once $\P$ has been completed, the lifting step
provides an effective way to compute the cylindrical algebraic decomposition.
Accordingly, one considers the coefficients of polynomials in $\R$
during the elimination step and restrict them to belong to $\Q$ during
the lifting step.

\subsubsection{Elimination step.}

The following lemma establishes that the roots of a polynomial
are ``continuous'' w.r.t. the coefficients of the polynomial
when the degree of the polynomial remains constant.

\begin{lemma}
\label{lemma:rootcont}
Let $P \in \C[X_1,\ldots,X_{k-1}][X_k]$, $S\subseteq \C^{k-1}$ 
such that $deg(P(x))$ is constant over $x\in S$. Let $a\in S$
such that $z_1, \ldots, z_m$ are the roots of $P(a)$ with multiplicities
$\mu_1, \ldots,\mu_m$, respectively. 
Let $0<r<\min_{i\neq j}(|z_i-z_j|/2)$.
Then there exists an open neighborhood $U$ of $a$ such that for
$x \in U$, $P(x)$ has exactly $\mu_i$ roots counted with multiplicities
in the disc $D(z_i,r)$
for all $i\leq m$.
\end{lemma}
\begin{proof}
Since the degree of $P$ is constant we can divide the coefficients
by the leading coefficient, obtaining a monic polynomial with same
roots and multiplicities and coefficients being rational functions.

\noindent
Assume that $P=X_k^\mu$. Consider $Q=X^\mu-\sum_{i< \mu} b_i X^i$ with 
$\delta=\max_{i<\mu} |b_i|<\frac{\min(1,r^\mu)}{\mu}$.
Since $\delta<\frac{1}{\mu}$, any root of $Q$ has a module less than one.
Let $z$ be such a root. Then $z^\mu=\sum_{i< \mu} b_i z^i$. 
So $|z^\mu|\leq \mu\delta<r^\mu$ which implies $|z|<|r|$. 

\noindent
Let us consider the mapping from pairs $(Q,R)$ of monic polynomials of
degree respectively $q$ and $r$ to their product $\varphi(Q,R)=QR$ of
degree $q+r$ (viewed as mapping of their coefficients).  This mapping
is differentiable. It is routine to check that the Jacobian matrix of
this mapping is equal or opposite to the subresultant $Sres_0(Q,R)$
and so it locally admits a differentiable inverse if $Q$ and $R$ are
coprime.  Therefore, factoring $P=QR$ such that $Q$ and $R$ are
coprime, there exists some neighborhoods $\mathcal{V}_Q$,
$\mathcal{V}_R$ respectively of $Q$ and $R$, such that $\mathcal{V} =
\varphi(\mathcal{V}_Q \times \mathcal{V}_R)$ is a neighborhood of $P$.

\noindent
By iteration, the polynomial $P_0=(X_k-z_1)^{\mu_1}\cdots(X_k-z_m)^{\mu_m}$
admits an open neighborhood $\mathcal{V}$ of its coefficients 
such that every monic polynomial $P_1 \in \mathcal{V}$ admits
a decomposition $P_1=Q_1\ldots Q_m$ with every $Q_i$ of degree $\mu_i$
and whose roots belong to the disc $D(z_i,r)$. Since the discs have no
intersection, every disc contains exactly $\mu_i$ roots counted
with multiplicities.

\noindent
Since the coefficients of $P$ are rational functions of $X_1,\ldots X_{k-1}$ 
and so continuous, there is a neighborhood $U$ of $a$ 
that fulfills the conclusion of the lemma.
\qed
\end{proof}

The next proposition establishes that the real roots 
of a set of polynomials 
are ``continuous'' w.r.t. the coefficients of the polynomials
when the degrees of some appropriate polynomials (including the original ones) 
remain constant.

\begin{proposition}
\label{proposition:rootcontset}
Let $P_1,\ldots,P_s \in \R[X_1,\ldots,X_{k-1}][X_k]$, $S\subseteq \R^{k-1}$
connected. Assume that over
$x \in S$, 
for all $1\leq i,j \leq s$,
$P_i(x)$ is not identically 0, $deg(P_i(x))$,
$deg(gcd(P_i(x),P_j(x))$, $deg(gcd(P_i(x),P_i'(x))$ are both constant.

\noindent
Then there exist $\ell$ (with $\ell$ possibly null) continuous functions
$f_1<\cdots<f_{\ell}$ from $S$ to $\R$ such that for every
$x\in S$, the set of real roots of $\prod_{j\leq s} P_j(x)$ is exactly
$\{f_1(x),\ldots,f_{\ell}(x)\}$.

\noindent
Moreover for all $i \leq \ell, j \leq s$, 
the multiplicity of the (possible) root  $f_i(x)$
of $P_j(x)$ is constant over $x \in S$.
\end{proposition}
\begin{proof}
Let $a \in S$  and $z_1(a),\ldots,z_m(a)$ be the roots in $\C$
of $\prod_{j\leq s} P_i(a)$ with  $\mu_i^j$  being the multiplicity
of $z_i(a)$ for $P_j(a)$. The degree of $R_{jk}(a)=gcd(P_j(a),P_k(a))$
is $\sum_{i\leq m} \min(\mu_i^j,\mu_i^k)$ and $\min(\mu_i^j,\mu_i^k)$
is the (possibly null) multiplicity of $z_i(a)$ for $R_{jk}(a)$. 

\noindent
Pick $r>0$ such that the discs $D(z_i(a),r)$ are disjoint. 
Observe that since  $deg(gcd(P_j(x),P_j'(x))$ is constant
over $x \in S$ the number
of distinct roots of $P_j(x)$ is constant over $x\in S$.
Let $i,j$ such that $\mu_i^j>0$, applying
Lemma~\ref{lemma:rootcont} and the previous
observation, there is a neighborhood $U$ of $a$
such that for all $x\in U$, $D(z_i,r)$ contains exactly
a root, denoted $z_i^j(x)$,
of $P_j(x)$ with multiplicity $\mu_i^j$. Assume there
exists $k\neq j$ with $\mu_i^k>0$,
since $deg(R_{jk}(x))$ is constant over $x \in S$, $z_i^j(x)=z_i^k(x)$
for all $x\in U$.
Otherwise for such an $x$ where the equality does not hold
$deg(R_{jk}(x))<deg(R_{jk}(a))$. So we can omit the superscript $j$
in $z_i^j(x)$ (defined when $\mu_i^j>0$).

\noindent
If $z_i(a)$ is real then $z_i(x)$ is real otherwise its conjugate
would be another root in $D(z_i(a),r)$. If $z_i(a)$ is complex,
its conjugate being also a root, $D(z_i(a),r)$ and $D(\overline{z_i(a)},r)$
are disjoint and so $z_i(x)$ is not real. Hence the number of 
real roots of $(x)$ is constant over $x \in U$. As the number of real roots
is locally constant and $S$ is connected then the number of
real roots of $\prod_{j\leq s} P_j(x)$ is constant over $x \in S$, say $\ell$.

\noindent
Let $f_i(x)$, for $i\leq l$
be the function that associates with $x$ the $i^{th}$
real root of $\prod_{j\leq s} P_j(x)$ in increasing order. Since $r$
could be chosen arbitrarily small, $f_i$ is
continuous. As the multiplicity of $f_i(x)$
w.r.t. any $P_j(x)$ and $Q(x)$ is locally constant,
it is constant over $x \in S$.
\qed
\end{proof}

%
%
%

The next definition is a basic construction that will
be the atomic step of the elimination stage.

\begin{definition}
Let $P = \sum_{i\leq p} a_i X_k^i \in \R[X_1,\ldots,X_{k-1}][X_k]$.
Then $lcof(P)=a_p$ and 
$Tru(P)=\{\sum_{i\leq r} a_i X_k^i \mid \forall i>r\ a_{i} \notin \R^* \wedge a_r\neq 0 \}$.

\noindent
Let $\P$ be a finite subset of $\R[X_1,\ldots,X_{k-1}][X_k]$.
Then $Elim_{X_k}(\P)$ is the set of polynomials of $\R[X_1,\ldots,X_{k-1}]$
defined as follows. For all $P,Q \in \P, R \in Tru(P)$, $T \in Tru(Q)$
with $deg (T)\leq deg(R)$:
\begin{itemize}
 \item If $lcof(R)$ does not belong to $\R$ then $lcof(R) \in Elim_{X_k}(\P)$; 
 \item If $deg(R)\geq 2$ then for all $sRes_j(R,R')$
 that are defined and do not belong to $\R$, $sRes_j(R,R') \in Elim_{X_k}(\P)$; 
 \item for all $sRes_j(R,T)$
 that are defined and do not belong to $\R$, $sRes_j(R,T) \in Elim_{X_k}(\P)$.
\end{itemize}
\end{definition}


The next lemma establishes the interest of the $Elim_{X_k}$
construction.

\begin{lemma}
\label{lemma:rootcontsetsuff}
Let $\P$ be a finite set of $\R[X_1,\ldots,X_{k-1}][X_k]$, 
$S\subseteq \R^{k-1}$
a connected set. Assume that 
$S$ is $Elim_{X_k}(\P)$-invariant.

\noindent
Then there exist $\ell$ (with $\ell$ posibly null)
continuous  functions
$f_1<\cdots<f_{\ell}$ from $S$ to $\R$ such that for every
$x\in S$, the set of real roots of $\prod_{P \in \P^*}P(x)$ is exactly
$\{f_1(x),\ldots,f_{\ell}(x)\}$
where $\P^*$ is the subset of $\P$ consisting of polynomials
not identically null over $S$.

\noindent
Moreover for all $i \leq l$ and for all $P\in \P^*$, 
the multiplicity of the root  $f_i(x)$
of $P(x)$ is constant over $x \in S$.
\end{lemma}
\begin{proof}
Let $P \in \P$. Since the leading coefficients of $Tru(P)$  belong to
$Elim_{X_k}(\P)$, the degree of $P(x)$ is constant over $x \in S$.

Let $R \in Tru(P)$ be the \emph{appropriate} polynomial for $P$
(i.e. whose degree is the degree of $P(x)$ for $x \in S$).
Then, by $deg(gcd(R,R'))$ 
is determined
by the signs of polynomials of the sequence $Sres(R,R')$
due to Proposition~\ref{proposition:syha-gcd}. Since all these polynomials
belong to $Elim_{X_k}(\P)$, the number of distinct complex roots
of $deg(gcd(P(x),P'(x))$ is constant over $x \in S$.

Let $T \in Tru(Q)$ be the appropriate polynomial of $Q$ for $Q\in \P$.
Then, by Proposition~\ref{proposition:syha-gcd}, $deg(gcd(R,T))$ is determined
by the signs of polynomials of the sequence $Sres(R,T)$.
Since all these polynomials
belong to $Elim_{X_k}(\P)$, the degree
of $gcd(P(x),Q(x))$ is constant over $x \in S$.

The conclusion follows using Proposition~\ref{proposition:rootcontset}.

\qed 
\end{proof}

We are now in position define the elimination step
and to prove its correctness.

\begin{theorem}
  Let $\mathcal Q=\{\mathcal Q_i\}_{i\leq n}$ be a family of finite
  set of polynomials such that $\mathcal Q_i \subseteq
  \R[X_1,\ldots,X_i]$.  Define $\P_n=\mathcal Q_n$ and inductively
  $\P_{i-1}=\mathcal Q_{i-1} \cup Elim_{X_i}(\mathcal Q_i)$ for
  $i>1$. Then there exists a cylindrical algebraic decomposition
  adapted to $\P$ (and thus to $\mathcal Q$).
\end{theorem}
\begin{proof}
Let us prove the existence of a cylindrical algebraic decomposition
of $\R^i$ adapted to $\P_i$ by induction.

\noindent 
The children of $\R^0$ form the partition defined 
by \[(-\infty,r_1),r_1,(r_1,r_2),\ldots,(r_{m-1},r_m),r_m,(r_m,\infty)\]
where $\{r_1,\ldots,r_m\}$ is the set of roots of all $P \in \P_1$
(or $\R$ if there is no root).
By construction, the cells of $\S_1$ are $\P_1$-invariant
and open intervals or points.

\noindent
Assume that we have built our tree up to level $i<n$. 
Pick any cell $C$ of level $i$. 
$C$ is $Elim_{X_{i+1}}(\P_{i+1})$-invariant since 
$Elim_{X_{i+1}}(\P_{i+1}) \subseteq \P_i$. Applying 
Lemma~\ref{lemma:rootcontsetsuff} yields the children of $C$.

\qed 
\end{proof}

\paragraph*{Complexity of elimination step.}
Let $s=|\mathcal{Q}|$, $d$ be the maximal total degree of polynomials
of $\mathcal{Q}$, and $v$ the maximal constant appearing in a coefficient of $\mathcal{Q}$.
A straightforward recurrence shows that
\begin{itemize}
\item the maximal number of bits of a coefficient of any $\P_i$ is $O(d^n\cdot3^{\log(n)\cdot\frac{n(n-1)}2}\cdot \log(v))$,
\item the maximal total degree of polynomials of all $\P_i$ is in $O(d^{3^n})$, and
\item the total number of polynomials is in $O((sd)^{3^n})$.
\end{itemize}

\begin{example}\label{eliminationExample}
  Let us build the family $\P_1,\P_2$ of polynomials associated with
  the automaton of \figurename~\ref{fig:exPolIta}.  We set $I_1=X_1$,
  $I_2=X_2$, $A=X_1^2-X_1-1$, $B=(2X_1-1)X_2^2-1$ and
  $C=X_2+(X_1^2-5)$.  We start with $\P_2=\{I_2, B, C\}$, $\P_1=\{I_1,
  A\}$ and add to $\P_1$ polynomials computed by $Elim_{X_2}(\P_2)$.


We first add $lcof(B) = 2X_1-1$ to $\P_1$.
Note that we do not add $lcof(C)$ since it is in $\Q$.

Let us now compute all subresultants of (potentially truncated)
polynomials of $\P_2$:
\begin{itemize}[label=\textbullet]
\item $sRes_0(I_2,C) = \left|\begin{matrix}1&0\\1&X_1^2-5\end{matrix}\right| = X_1^2-5$ is added to $\P_1$.
\item We then add to $\P_1$ the polynomial\begin{eqnarray*}
sRes_0(B,C) &=& \left|\begin{matrix}
2X_1-1 & 0 & -1 \\
0 & 1 & X_1^2-5 \\
1 & X_1^2-5 & 0
\end{matrix}\right| \\
&=& -(2X_1-1)(X_1^2-5)^2+1 \\
&=&-2X_1^5+X_1^4+20X_1^3-10X_1^2-50X_1+26
\end{eqnarray*}
\item Remark that $sRes_0(B,I_2) = 1 \in \Q$, hence it is not added to $\P_1$.
It is also the case for $sRes_1(B,I_2)$ and $sRes_1(B,C)$.
\end{itemize}

We then need to compute the subresultants of each polynomial of degree $\geq 2$ with its derivative.
In our case, that means computing $sRes_0(B,B')$ and $sRes_1(B,B')$.
We have $B'= 2(2X_1-1)X_2$.
We obtain $sRes_1(B,B') = 2(2X_1-1)$ that should be added to $\P_1$.
However, since $sRes_1(B,B') = 2 lcof(B)$, their sign will coincide.
For simplicity we will not keep it in $\P_1$, although the automatic procedure does; nonetheless, this would not affect the elimination at lower levels.
Finally, we have \[sRes_0(B,B') = \left|\begin{matrix}
2X_1-1 & 0 & -1 \\
0 & 2(2X_1-1) & 0 \\
2(2X_1-1) & 0 & 0
\end{matrix}\right| = 4(2X_1-1)^2\]
which is added to $\P_1$.
This concludes the elimination phase.

The final sets $\P_1$ and $\P_2$ are given in \tablename~\ref{tab:namePolynomials} (page~\pageref{tab:namePolynomials}).
\end{example}

\subsubsection{Lifting step.}

\begin{algorithm2e}
\DontPrintSemicolon

\KwIn{$\P=\{\P_\ell\}_{\ell\leq k}$ a family of subsets of polynomials obtained by decomposition}

\KwOut{$\mathcal A$ a tree whose nodes at
level $\ell$ are sample points of the decomposition  equipped with their sign
evaluation for $\P_\ell$}

\SetKwFunction{Lcof}{Lcof}
\SetKwFunction{OrderedMerge}{OrderedMerge}
\SetKwFunction{EnlargeWith}{EnlargeWith}
\SetKwFunction{Singleton}{Singleton}
\SetKwFunction{Trunk}{Trunk}
\SetKwFunction{Degree}{Degree}
\SetKwFunction{Normalize}{Normalize}
\SetKwFunction{RootCoding}{RootCoding}
\SetKwFunction{EuclideanDivision}{EuclideanDivision}
\SetKwFunction{PmV}{PmV}
\SetKwFunction{Sign}{Sign}
\SetKwFunction{PmV}{PmV}
\SetKwFunction{Completing}{Completing}
\SetKwFunction{Lifting}{Lifting}
\SetKwFunction{First}{First}
\SetKwFunction{Last}{Last}
\SetKwFunction{Length}{Length}
\SetKwFunction{LinePartition}{LinePartition}

{\Lifting}$(\ell,\T)$: an integer\;
\KwIn{$\ell$, the current level; 
$\T=\{(n_i,P_i,p_i)\}_{i=1}^\ell$  a 
triangular system
for $(\alpha_1,\ldots,\alpha_\ell)$
corresponding to a node of $\mathcal A$.}
\KwData{$L$ a list of   triangular systems equipped with sign vectors, $E$ a triangular system 
with a sign vector}

$L\leftarrow \LinePartition(\ell,\T)$\;

 \uIf{$L=\emptyset$} 
 {
   $\T' \leftarrow \mathcal T \cup \{(1,X_{\ell+1},1)\}$;
   $\T'\cdot Eval \leftarrow \{(P,\Sign(\ell,\T,\Lcof(P))\mid P \in \P_{\ell+1} \}$\; 
   $\A \leftarrow \A \cup (\T \rightarrow \T')$;
   \lIf{$\ell+1<k$}{$\Lifting(\ell+1,\T')$} 
 }
\Else
{
  $L \leftarrow \Completing(\ell,\T,L)$\;
 \For{$E \in L$}
 {
   Pick some $(r,v,P)\in E$ such that $r$ is defined\;
   $\T' \leftarrow \T \cup \{(r,P,\Degree(\ell,\T,P))\}$;
   $\T'\cdot Eval \leftarrow \{(Q,v[0])\mid Q \in \P_{\ell+1} \wedge \exists (m,v,Q)\in E \}$\; 
   $\A \leftarrow \A \cup (\T \rightarrow \T')$;
   \lIf{$\ell+1<k$}{$\Lifting(\ell+1,\T')$} 
 }
}
\caption{Lifting the cylindrical algebraic decomposition at a point of level $\ell$}
\label{algo-lifting}
\end{algorithm2e}


We build the cylindrical algebraic decomposition as follows: every
cell $C$ of level $\ell$ is represented by a \emph{sample point},
represented by a triangular system.  In addition, the representation
of $C$ includes the evaluation of the sign of all $P \in \P_\ell$.
Observe that evaluation of a $P \in \P_j$ with $j<\ell$ is found in
its ancestor cell of level $j$.  The construction is performed by
Algorithm~\ref{algo-lifting}.  An atomic step of the lifting phase
corresponds to build, given a sample point $\R^{\ell}$, the ordered
list of all sample points of $\R^{\ell+1}$ representing the cells of
the cylinder above $S$. It corresponds to a call to {\tt Lifting}
(without the recursive calls).  The whole construction is done by the
call ${\tt Lifting}(0,\emptyset)$.  {\tt Lifting} first calls {\tt
  LinePartition} in order to get an ordered list of the roots of all
$P \in \P_{\ell+1}$. Every real $\alpha$ of this list is represented
by a set of triplets $(r,v,P)$ where $P$ is a polynomial whose
coefficients are algebraic numbers over $\mathcal T$ (and thus
represented by polynomials in $Q[X_1,\ldots,X_\ell]$), $v$ is the
$P$-encoding of $\alpha$.  $r$ may be undefined but when defined it
means that $\alpha$ is the $r^{th}$ root of $P$.  For at least one
triplet of the set $r$ is defined allowing to extend the triangular
system $\mathcal T$ by $\alpha$.  Since one wants to represent the
interval between these roots by sample points, the list is completed
by a call to {\tt Completing}.  After this call either the list is
empty (corresponding to the case of a single child $C\times \R$) and
this child is represented by $\alpha_{\ell+1}=0$, first root of
$X_{\ell+1}$.  The representation of this cell is now enlarged by the
evaluation of all $P \in \P_{\ell+1}$ at this sample point. Otherwise
for every item of the list one picks some arbitrary $(r,v,P)$ with $r$
defined and proceeds as previously to produce all the children of $C$.

Algorithm~\ref{algo-linepartitions} 
produces the list of roots of all $P(\alpha_1,\ldots,\alpha_{\ell})$
for $P \in \P_{\ell+1}$. For any such $P$, it first normalizes
it by determining its higher non null coefficient. Thus $R \in Tru(P)$.
$SL[P]$ will contain the singletons $\{(r,v,P)\}$ for every root
of $P(\alpha_1,\ldots,\alpha_{\ell})$. Then the algorithm enlarges these
singletons with triplets $\{(r',v',Q)\}$ for all $Q$
that preceed $P$ in $\P_{\ell+1}$. All these triplets are obtained
using the lists provided by appropriate calls to {\tt RootCoding}.
Conversely the sets of the list $SL[Q]$ are enlarged with the
triplets related to $P$. Once all roots have been produced
in $SL$, it remains to order them and (possibly) merge them. 
This can be easily  done
with the help of their Thom-encoding and it is performed 
by a call to {\tt OrderedMerge}.

\begin{algorithm2e}
\DontPrintSemicolon

\KwIn{$\P=\{\P_\ell\}_{\ell\leq k}$ a family of subsets of polynomials}

\SetKwFunction{Lcof}{Lcof}
\SetKwFunction{OrderedMerge}{OrderedMerge}
\SetKwFunction{EnlargeWith}{EnlargeWith}
\SetKwFunction{Singleton}{Singleton}
\SetKwFunction{Trunk}{Trunk}
\SetKwFunction{Degree}{Degree}
\SetKwFunction{Normalize}{Normalize}
\SetKwFunction{RootCoding}{RootCoding}
\SetKwFunction{EuclideanDivision}{EuclideanDivision}
\SetKwFunction{PmV}{PmV}
\SetKwFunction{Sign}{Sign}
\SetKwFunction{PmV}{PmV}
\SetKwFunction{Lifting}{Lifting}
\SetKwFunction{First}{First}
\SetKwFunction{Last}{Last}
\SetKwFunction{LinePartition}{LinePartition}

{\LinePartition}$(\ell,\mathcal T)$: a list\;
\KwIn{$\ell$, the current level}
\KwIn{$\mathcal T=\{(n_i,P_i,p_i)\}_{i=1}^\ell$  a 
triangular system
for $(\alpha_1,\ldots,\alpha_\ell)$
corresponding to a node of $\mathcal A$ whose children have to be computed.}
\KwOut{$L$ a list of sample points  of the decomposition equipped with their sign
evaluation for $\P_{\ell+1}$ related to $\mathcal T$.}
\For{$P \in \P_{\ell+1}$}
{
   $(R,r) \leftarrow \Normalize(\ell,\mathcal T,P)$\;
   \lIf {$r \leq 0$}  {$SL[P] \leftarrow \emptyset$}
   \Else
   {
     $SLL \leftarrow \RootCoding(\ell,\mathcal T,R,r,R,r)$\;
     \tcp{\Singleton transforms a list of items into a list of singletons which contain these items.
     Furthermore it adds the number of the root of $R$ for subsequent use.}
     $SL[P] \leftarrow \Singleton(SLL)$\;
     \For{$Q \in \P_{\ell+1} {\bf ~such~that~} Q\prec P$}
     {$(S,s)\leftarrow \Normalize(\ell,\mathcal T,Q)$\;
      $SLL \leftarrow \RootCoding(\ell,\mathcal T,R,r,S,s)$;
      $\EnlargeWith(SL[P],SLL,Q)$\;
     }
   } 
     \For{$Q \in \P_{\ell+1} {\bf ~such~that~} Q\prec P$}
     {
     \If {$SL[Q]\neq \emptyset$}{
        $(S,s) \leftarrow \Normalize(\ell,\mathcal T,Q)$\;
        $SLL \leftarrow \RootCoding(\ell,\mathcal T,S,s,R,r)$;
        $\EnlargeWith(SL[Q],SLL,P)$\;}
     }
}   
$L\leftarrow \OrderedMerge(SL)$\;
\Return $L$\;
\caption{Partitioning the real line at a point of level $\ell$.}
\label{algo-linepartitions}
\end{algorithm2e}

\begin{algorithm2e}
\DontPrintSemicolon

\KwIn{$\P=\{\P_\ell\}_{l\leq k}$ a family of subsets of polynomials obtained by decomposition}
\SetKwFunction{Lcof}{Lcof}
\SetKwFunction{OrderedMerge}{OrderedMerge}
\SetKwFunction{EnlargeWith}{EnlargeWith}
\SetKwFunction{Singleton}{Singleton}
\SetKwFunction{Trunk}{Trunk}
\SetKwFunction{Degree}{Degree}
\SetKwFunction{Normalize}{Normalize}
\SetKwFunction{RootCoding}{RootCoding}
\SetKwFunction{EuclideanDivision}{EuclideanDivision}
\SetKwFunction{PmV}{PmV}
\SetKwFunction{Sign}{Sign}
\SetKwFunction{PmV}{PmV}
\SetKwFunction{Completing}{Completing}
\SetKwFunction{First}{First}
\SetKwFunction{Last}{Last}
\SetKwFunction{Length}{Length}

{\Completing}$(\ell,\mathcal T,L)$: a list\;
\KwIn{$\ell$, the current level}
\KwIn{$\T=\{(n_i,P_i,p_i)\}_{i=1}^\ell$  a 
triangular system
for $(\alpha_1,\ldots,\alpha_\ell)$
corresponding to a node of $\mathcal A$ whose children have to be computed.}

\KwIn{$L$ a list of sample points  of the decomposition represented by a
triangular system equipped with their sign
evaluation for $\P_{\ell+1}$ related to $\T$.}

\KwOut{the input list $L$ enriched with of sample points 
for the intervals before, between and beyond the original sample points.}

 \For{$E \in L$}
 {
   Pick some $(r,v,P)\in E$ such that $r$ is defined\;
   \uIf {$E=\First(L)$}
   {
    $(R,r) \leftarrow \Normalize(\ell,\T,P(X_{\ell+1}+1))$;
    $SLL \leftarrow \RootCoding(\ell,\mathcal T,R,r,R,r)$\;
    $shortL \leftarrow \Singleton(SLL)$\;
    \For{$Q \in \P_{\ell+1}$}
     {$(S,s) \leftarrow \Normalize(\ell,\mathcal T,Q)$\;
      $SLL \leftarrow \RootCoding(\ell,\mathcal T,R,r,S,s)$;
      $\EnlargeWith(shortL,SLL,Q)$\;
     }
   Insert $\First(shortL)$ before $E$ in $L$\;
   }
   \Else{
    $(R,r) \leftarrow \Normalize(\ell,\mathcal T,(P\cdot oldP)'$;
    $SLL \leftarrow \RootCoding(\ell,\mathcal T,R,r,R,r)$\;
    $shortL \leftarrow \Singleton(SLL)$\;
    \For{$Q \in \P_{\ell+1}$}
     {$(S,s)\leftarrow \Normalize(\ell,\mathcal T,Q)$\;
      $SLL \leftarrow \RootCoding(\ell,\mathcal T,R,r,S,s)$;
      $\EnlargeWith(shortL,SLL,Q)$\;
     }
   Find $F$ in $shortL$ such that $\exists (x,vP,P), (y,voldP,oldP) \in F$\;
   with $vP<v$ and $voldP>oldv$;
   Insert $F$ before $E$ in $L$\;
   }   
   $oldv \leftarrow v$; $oldP \leftarrow P$\;
   }
    Let $E$ be $\Last(L)$\;    
    Pick some $(r,v,P)\in E$ such that $r$ is defined\;
    $(R,r) \leftarrow \Normalize(\ell,\mathcal T,P(X_{\ell+1}-1))$\;
    $SLL \leftarrow \RootCoding(\ell,\mathcal T,R,r,R,r)$;
    $shortL \leftarrow \Singleton(SLL)$\;
    \For{$Q \in \P_{\ell+1}$}
     {$(S,s)\leftarrow \Normalize(\ell,\mathcal T,Q)$\;
      $SLL \leftarrow \RootCoding(\ell,\mathcal T,R,r,S,s)$;
      $\EnlargeWith(shortL,SLL,Q)$\;
     }
   Insert $\Last(shortL)$ after $E$ in $L$; \Return $L$\;

\caption{Completing the line partition with samples of intervals.}
\label{algo-between-roots}
\end{algorithm2e}

Algorithm~\ref{algo-between-roots} completes the list of roots by sample points
representing the intervals between the roots.
This is done as follows.
Given a root $\alpha$ of $P$ and a root $\beta$ of $Q$, 
such that $\alpha$ and $\beta$ are consecutive items
of the list, there exists a root $f$ of $(PQ)'$ 
such that $f \in ]\alpha,\beta[$. Thus the sample
point will be an arbitrary root of $(PQ)'$ strictly between $\alpha$ 
and $\beta$.
If $\alpha$ is the smallest (resp. largest) root in the list
for of some $P$ then the first (resp. last) root of $P[X_{\ell+1}+1]$
(resp. $P[X_{\ell+1}-1]$) is $\alpha-1\in ]-\infty,\alpha[$
(resp. $\alpha+1 \in ]\alpha,+\infty[$). In this algorithm
$E$ represents the current item, say $\beta$ of the list of roots, $P$
some polynomial whose $\beta$ is a root and $v$ is its $P$-encoding.
Let $\alpha$ be the previous item of the list (when it exists). $oldP$
is some polynomial whose $\alpha$ is a root and $oldv$ is its $oldP$-encoding.
Thus in order to find a root of $(P \cdot oldP)'$ between $\alpha$ and $\beta$,
one computes the $P$ and $oldP$ encoding of the roots of $(P \cdot oldP)'$.



\begin{example}
We first (by Algorithm~\ref{algo-linepartitions}) compute the line partition of $\R$ at level $1$ for $\P_1 = \{I_1,A,D,E,F,G\}$ (see \tablename~\ref{tab:namePolynomials}) obtained previously.
This is done by comparing the $P$-encodings of roots of $Q$ for all pairs $(P,Q) \in \P_1^2$.
The result is (partially) depicted in \figurename~\ref{fig:partitionExample}.
Each bullet represents the (relative) position of a root, given by a triangular system (where the degree of the polynomial is not represented for clarity).
In the table, the line labeled by $P$ gives the $P$-encodings of the roots.
\end{example}

\begin{table}
\centering
\[\begin{array}{rcl}
I_1 &=& X_1 \\
I_2 &=& X_2 \\
A &=& X_1^2 - X_1 -1 \\
B &=& (2X_1-1)X_2^2-1 \\
C &=& X_2+X_1^2-5 \\
D &=& 2X_1-1 \ (=lcof(B)) \\
E &=& X_1^2-5 \ (=sRes_0(I_2,C)) \\
F &=& -2X_1^5+X_1^4+20X_1^3-10X_1^2-50X_1+26 \ (= sRes_0(B,C)) \\
G &=& 4(2X_1-1)^2 \ (=sRes_0(B,B'))\\
Int &=&-14X_1^6+18X_1^5+105X_1^4-124X_1^3-180X_1^2+172X_1+24 \ (=(FA)') \\
\end{array}\]
\begin{mathpar}
\P_1 = \{I_1,A,D,E,F,G\} \and \P_2 = \{I_2,B,C\}
\end{mathpar}
\caption{Polynomials used in the cylindrical decomposition.}
\label{tab:namePolynomials}
\end{table}

\begin{figure}
\centering
	\begin{tikzpicture}[auto,node distance=0.125cm and 0.9cm]

\tikzstyle{pointlabel}=[anchor=south]
\tikzstyle{pointvalue}=[anchor=north]
\tikzstyle{point}=[inner ysep=0pt]
\tikzstyle{poly}=[anchor=west,inner xsep=0pt]
\tikzstyle{coding}=[font=\scriptsize]

\node[point] (zero) at (0,0) {\textbullet};
\node[pointlabel] (zeroLab) at (zero.north) {$(1,I_1)$};
\node[pointvalue] (zeroVal) at (zero.south) {$0$};

\node[point,left= of zero] (a1)  {\textbullet};
\node[pointlabel] (a1Lab) at (a1.north) {$(1,A)$};
\node[pointvalue] (a1Val) at (a1.south) {$\frac{1-\sqrt5}2$};

\node[point,left= of a1] (e1)  {\textbullet};
\node[pointlabel] (e1Lab) at (e1.north) {$(1,E)$};
\node[pointvalue] (e1Val) at (e1.south) {$-\sqrt5$};

\node[point,right= of zero] (d1)  {\textbullet};
\node[pointlabel] (d1Lab) at (d1.north) {\begin{tabular}{c}$(1,D)$\\$(1,G)$\end{tabular}};
\node[pointvalue] (d1Val) at (d1.south) {$\frac12$};

\node[point,right= of d1] (f1)  {\textbullet};
\node[pointlabel] (f1Lab) at (f1.north) {$(1,F)$};

\node[point,right= of f1] (a2)  {\textbullet};
\node[pointlabel] (a2Lab) at (a2.north) {$(2,A)$};
\node[pointvalue] (a2Val) at (a2.south) {$\frac{1+\sqrt5}2$};

\node[point,right= of a2] (f2)  {\textbullet};
\node[pointlabel] (f2Lab) at (f2.north) {$(2,F)$};

\node[point,right= of f2] (e2)  {\textbullet};
\node[pointlabel] (e2Lab) at (e2.north) {$(2,E)$};
\node[pointvalue] (e2Val) at (e2.south) {$\sqrt5$};

\node[point,right= of e2] (f3)  {\textbullet};
\node[pointlabel] (f3Lab) at (f3.north) {$(3,F)$};

\draw ($(e1) + (-0.5,0.5pt)$) -- ($(f3) + (0.5,0.5pt)$);

\node[point,fill=white,inner sep=0pt] (int4) at ($0.5*(f1)+0.5*(a2)$) {$\circ$};
\node[pointlabel,inner ysep=0pt] (int4lab) at ($(int4.north) + (0,0.5)$) {$(4,Int)$};
\path[draw,dotted] (int4) edge (int4lab);

\node[poly] (i1) at ($(e1) + (-1.25,-1.25)$) {$I_1$};
\node[poly, below=of i1] (a) {$A$};
\node[poly, below=of a] (d) {$D$};
\node[poly, below=of d] (e) {$E$};
\node[poly, below=of e] (f) {$F$};
\node[poly, below=of f] (g) {$G$};

\node[coding] at (i1 -| zero) {$(0,1)$};
\node[coding] at (a -| zero) {$(-1,-1,1)$};
\node[coding] at (d -| zero) {$(-1,1)$};
\node[coding] at (e -| zero) {$(-1,-1,1)$};
\node[coding] at (f -| zero) {$(1,1,-1,1,1,-1)$};
\node[coding] at (g -| zero) {$(1,-1,1)$};

\node[coding] at (i1 -| e1) {$(-1,1)$};
\node[coding] at (i1 -| a1) {$(-1,1)$};
\node[coding] at (i1 -| d1) {$(1,1)$};
\node[coding] at (i1 -| f1) {$(1,1)$};
\node[coding] at (i1 -| a2) {$(1,1)$};
\node[coding] at (i1 -| f2) {$(1,1)$};
\node[coding] at (i1 -| e2) {$(1,1)$};
\node[coding] at (i1 -| f3) {$(1,1)$};

\node[coding] at (a -| e1) {$(1,-1,1)$};
\node[coding] at (a -| a1) {$(0,-1,1)$};
\node[coding] at (a -| d1) {$(-1,0,1)$};
\node[coding] at (a -| f1) {$(-1,1,1)$};
\node[coding] at (a -| a2) {$(0,1,1)$};
\node[coding] at (a -| f2) {$(1,1,1)$};
\node[coding] at (a -| e2) {$(1,1,1)$};
\node[coding] at (a -| f3) {$(1,1,1)$};


\node at (e -| e1) {$\vdots$};
\node at (e -| f2) {$\ddots$};
\end{tikzpicture}
\caption[Partition of $\R$ according to $\P_1$ and Thom encodings.]{Partition of $\R$ according to $\P_1$ and Thom encodings. The scale is not accurate.}
\label{fig:partitionExample}
\end{figure}
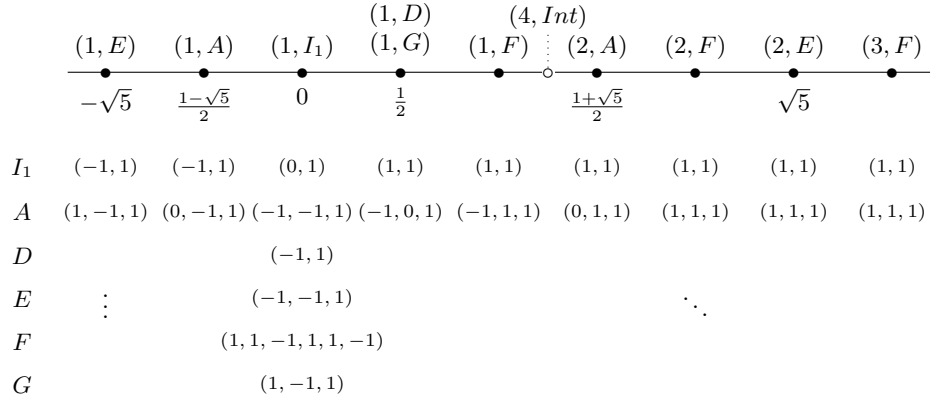

\begin{example}
We can now complete the line built above by computing sample points corresponding to intervals between consecutive roots (Algorithm~\ref{algo-between-roots}).
For instance to compute a sample point at the left of $(1,E)=-\sqrt5$, one can choose $-1-\sqrt5$ which is the first root of $H=(X+1)^2-5$ (\emph{i.e.} $E$ where $X$ is replaced by $X+1$).
In order to compute a value between $(1,F)$ and $(2,A)$, we consider the polynomial $Int=(FA)'=-14X_1^6+18X_1^5+105X_1^4-124X_1^3-180X_1^2+172X_1+24$.
Computing the $F$-encodings of roots of $Int$ gives the number $k$ of roots of $Int$ smaller than or equal to $(1,F)$.
Taking the $k+1$th root of $Int$ yields a root greater than $(1,F)$.
The value $(k+1,Int)$ is smaller than $(2,A)$ (since one such root exists).
Here, one can show that the appropriate root is the $4$th.
Hence the sample point $(4,Int)$ written $\alpha_1$ is added to the line in order to represent interval $](1,F),(2,A)[$, as depicted by the empty bullet on \figurename~\ref{fig:partitionExample}.
In addition, for all polynomials $P$ of $\P_1$, the $P$-encoding of $(4,Int)$ is computed: the first component yields the sign of $P$ in the interval.
Namely:
\begin{mathpar}
I_1(\alpha_1) > 0 \and
A(\alpha_1) < 0 \and
D(\alpha_1) >0 \\
E(\alpha_1) <0 \and
F(\alpha_1) < 0 \and
G(\alpha_1) >0
\end{mathpar}
Remark that this interval corresponds to the one where transition $a$ of \figurename~\ref{fig:exPolIta} is fired in the trajectory of \figurename~\ref{fig:exPolItaTrajectory}.

Sample points (and their encodings) for all intervals should be computed and added to the line.
This is omitted for readability.\MaS{Ils sont cach\'es en commentaires dans la \tablename~\ref{tab:namePolynomials}.}
\end{example}

\begin{example}
We illustrate the lifting (Algorithm~\ref{algo-lifting}) to $\R^2$ for the interval represented by the sample point $(4,Int)$ built above.
In this case, one must partition the real line with roots of polynomials of $\P_2 = \{I_2,B,C\}$ when $X_1=\alpha_1$.
Note that $I_1$ and $A$ are constants.

In the computation of the $\P_2$-encodings, the $\P_1$-encodings of $\alpha_1$ are used, in particular the encodings of polynomials constructed in the elimination phase.
For example, since $D(\alpha_1)>0$, the leading coefficient of $B$ is positive, hence $B$ has two roots.
And since $E(\alpha_1)<0$, the root of $C(\alpha_1)$ is positive (greater than the root of $I_2$).
Finding that all the roots of $B(\alpha_1)$ are smaller than $\gamma$ the root of $C(\alpha_1)$ involves not only the sign of $F(\alpha_1)$ (which only shows that $\gamma$ is not between the roots of $B(\alpha_1)$) but additional components of the encoding, namely in this case the sign of the second derivative of $F$.
This is partially represented in \figurename~\ref{fig:linePartitioningR2} (again, the degrees of the polynomials are omitted).
Note that this lifting corresponds to the trajectory depicted in \figurename~\ref{fig:exPolItaTrajectory}, page~\pageref{fig:exPolItaTrajectory}.
\end{example}

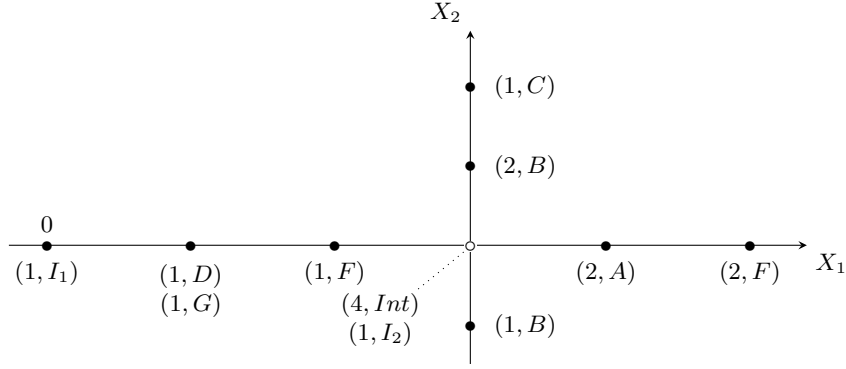
\begin{figure}
\centering

\pgfdeclarelayer{foreground}
\pgfsetlayers{main,foreground}

\begin{tikzpicture}[auto,node distance=0.9cm and 1.5cm]

\tikzstyle{pointlabel}=[anchor=north]
\tikzstyle{point2label}=[anchor=west]
\tikzstyle{pointvalue}=[anchor=south]
\tikzstyle{point}=[inner ysep=0pt]
\tikzstyle{poly}=[anchor=west,inner xsep=0pt]
\tikzstyle{coding}=[font=\scriptsize]

\begin{pgfonlayer}{foreground}
\node[point,fill=white,inner sep=0pt] (int4) at (0,0) {$\circ$};
\end{pgfonlayer}

\node[anchor=north east,inner ysep=0pt] (int4lab) at ($(int4.south) + (-0.5,-0.5)$) {\begin{tabular}{c}$(4,Int)$\\$(1,I_2)$\end{tabular}};
\path[draw,dotted] (int4) edge (int4lab);

\node[point,left= of int4] (f1)  {\textbullet};
\node[pointlabel] (f1Lab) at (f1.south) {$(1,F)$};

\node[point,left= of f1] (d1)  {\textbullet};
\node[pointlabel] (d1Lab) at (d1.south) {\begin{tabular}{c}$(1,D)$\\$(1,G)$\end{tabular}};

\node[point,left= of d1] (zero) {\textbullet};
\node[pointlabel] (zeroLab) at (zero.south) {$(1,I_1)$};
\node[pointvalue] (zeroVal) at (zero.north) {$0$};



\node[point,right= of int4] (a2)  {\textbullet};
\node[pointlabel] (a2Lab) at (a2.south) {$(2,A)$};

\node[point,right= of a2] (f2)  {\textbullet};
\node[pointlabel] (f2Lab) at (f2.south) {$(2,F)$};



\node[point, below= of int4] (b1) {\textbullet};
\node[point2label] (b1Lab) at (b1.east) {$(1,B)$};

\node[point, above= of int4] (b2) {\textbullet};
\node[point2label] (b2Lab) at (b2.east) {$(2,B)$};

\node[point, above= of b2] (c1) {\textbullet};
\node[point2label] (c1Lab) at (c1.east) {$(1,C)$};

\draw[->] ($(zero) + (-0.5,0.5pt)$) -- ($(f2) + (0.75,0.5pt)$) node[anchor=north west] {$X_1$};
\draw[->] ($(b1) + (0.pt,-0.5)$) -- ($(c1) + (0.pt,0.75)$) node[anchor=south east] {$X_2$};

\end{tikzpicture}
\caption{Line partitioning for $X_2$ above $\alpha_1=(4,Int)$.}
\label{fig:linePartitioningR2}
\end{figure}

\section{Verification algorithms for \polita}
\label{sec:reach}
We now use the cylindrical decomposition to build a finite abstraction
of the transition system associated with a \polita. The model checking
problem (hence also the reachability problem) can be solved with this
abstraction. An on-the-fly construction is then given to produce a
more efficient practical algorithm.
Formally, we prove the following:
\begin{theorem}\label{thm:modelchecking}
  The model checking problem of \TCTLint over \polita is decidable in
  time $(|\A|\cdot |\psi|\cdot d)^{2^{O(n)}}$ where $n$ is the number
  of clocks in $\A$ and $d$ the maximal degree of polynomials
  appearing in $\A$ and $\psi$.
\end{theorem}

\subsection{Abstraction construction}
Let $\A=\langle\Sigma,Q, q_0, F, X, \lambda, \Delta\rangle$ be a
\polita with $X=\{x_1,\dots,x_n\}$.  We define $\Poly(\A)$ the set of
all polynomials appearing in guards and updates of $\A$ (including all
clocks) as follows:
\begin{eqnarray*}
  \Poly(\A) &=& X \cup \bigcup_{(q,g,a,u,q') \in \Delta} 
\left(\left\{\bigcup_i \{P_i\} \mmid \fee= \bigwedge_i P_i \rel_i 0\right\}
  \right. \\ & & \hspace{3.5cm}\left.
    \cup \left\{\bigcup_{i=1}^n \{x_i - P_i\} \mmid 
u= \bigwedge_{i=1}^n x_i := P_i\right\}\right)
\end{eqnarray*}
Given a \TCTLint formula $\psi$, we define $\Poly(\psi)$ the set of all polynomials appearing in $\psi$, \emph{i.e.} in subformulas of the form $P \rel 0$.
Note that in the case of the reachability problem, $\Poly(\psi) = \emptyset$.

Let $\Dc_{\A,\psi}$ be the cylindrical algebraic decomposition adapted to
$\Poly(\A) \cup \Poly(\psi)$ and $X$.  Since $\Dc_{\A,\psi}$ is adapted to $X$, the cells can
be arranged in levels $\Dc_{\A,\psi}^1, \dots, \Dc_{\A,\psi}^n$, such that for
$1\leq i \leq n$, $\bigcup_{k=1}^i \Dc_{\A,\psi}^k$ is a CAD of
$\R^{\{x_1,\dots,x_i\}}$.  As a result, the projection of a cell of
level $i$ over the axis $x_i =0$ yields a cell of level $i-1$.

We define $\Reg_{\A,\psi}$ the finite transition system with states in $Q
\times \Dc_{\A,\psi}$, specifically, they can also be arranged by layer, with
respect to the level of the state: $\bigcup_{i=1}^n \lambda^{-1}(i)
\times \Dc_{\A,\psi}^i$.  Indeed, given a configuration $(q,v)$ with
$\lambda(q)=k$, the semantics of \polita require that for $k < i \leq
n$, $v(x_i)=0$, hence $v$ belongs to a cell of $\Dc_{\A,\psi}^k$.  We now
define the transitions of $\Reg_{\A,\psi}$ as follows.

\subsubsection{Time successors.}
Let $\timeSuc \notin\Sigma$ be a letter representing time elapsing.  Let
$(q,C)$ be a state of $\Reg_{\A,\psi}$, with $\lambda(q)=k$, and let
$\underline{C} \in \Dc_{\A,\psi}^{k-1}$ be the projection of $C$ onto
$\R^{k-1}$ and $-\infty=f_0 < \cdots <f_{r+1}=+\infty$ be
the functions dividing $\underline{C}$ as in
Definition~\ref{def:cad}.
The $\timeSuc$ transitions are defined as
follows:
\begin{itemize}
\item if $C = \left\{\left(x,f_i(x)\right) \mmid x \in
    \underline{C}\right\}$ for some $i \in \oneto{r}$, then there is a
  transition $(q,C) \tr{\timeSuc} (q,C')$ where $C' = \left\{(x,y) \mmid x
    \in \underline{C}, f_{i}(x) < y < f_{i+1}(x)\right\}$;
\item if $C = \left\{(x,y) \mmid x \in \underline{C}, f_{i-1}(x)
    < y < f_{i}(x)\right\}$ for some $i \in \oneto{r}$, then
  there is a transition $(q,C) \tr{\timeSuc} (q,C')$ where $C' =
  \left\{\left(x,f_{i}(x)\right) \mmid x \in
    \underline{C}\right\}$;
\item otherwise, $C = \left\{(x,y) \mmid x \in \underline{C},
    f_{r}(x) < y < f_{r+1}(x)\right\}$, and there is a
  self-loop labeled by $\timeSuc$: $(q,C) \tr{\timeSuc} (q,C)$.
\end{itemize}
In all the above cases, $C'$ is called the \emph{time successor} of
$C$ (in the last case, $C$ is its own time successor).

\begin{proposition}[Correctness w.r.t. time elapsing]\label{prop:timeSteps}
Let $v$ be a valuation of a cell $C$ of level $k$.
\begin{itemize}
\item There exists $d>0$ such that the elapsing of $d$ time units for
  $x_k$ yields a valuation $v+_k d \in C'$, the time successor of $C$.
\item For any $0<d'<d$, the elapsing of $d'$ time units for $x_k$
  yields a valuation $v+_k d$ that is either in $C$ or in
  $C'$.
\end{itemize}
\end{proposition}

\begin{proof}
We again distinguish the possible cases for $C$:
\begin{itemize}
\item If $C = \left\{\left(x,f_i(x)\right) \mmid x \in
    \underline{C}\right\}$ for some $i \in \oneto{r}$, then the time
  successor $C' = \left\{(x,y) \mmid x \in \underline{C},
    f_{i}(x) < y < f_{i+1}(x)\right\}$.  Then
  $v=(x,f_i(x))$.  By elapsing
  $\frac{f_{i+1}(x)-f_i(x)}2$ time units in level $k$, one
  clearly obtains a valuation of $C'$.  Moreover, for every inferior
  delay $d'$, $v+_k d'$ is also in $C'$.
\item If $C = \left\{(x,y) \mmid x \in \underline{C}, f_{i-1}(x)
    < y < f_{i}(x)\right\}$ for some $i \in \oneto{r}$, then $C'
  = \left\{\left(x,f_{i}(x)\right) \mmid x \in
    \underline{C}\right\}$.  Then $v=(x,y)$ with $f_{i-1}(x) < y
  < f_{i}(x)$.  By elapsing $f_{i}(x)-y$ time units in level
  $k$, one clearly obtains a valuation of $C'$.  Moreover, for every
  inferior delay $d'$, $v+_k d'$ remains in $C$.
\item Otherwise, $C = \left\{(x,y) \mmid x \in \underline{C},
    f_{r}(x) < y < f_{r+1}(x)=+\infty\right\}$, and any time
  elapsing for $x_k$ keeps the valuation in $C$.\qed
\end{itemize}
\end{proof}

\subsubsection{Discrete successors.}
Since $\Dc_{\A,\psi}$ is adapted in particular to $\Poly(\A)$ which contains all guards, we
have the following result:
\begin{lemma}\label{lem:guardsCAD}
  Let $C \in \Dc_{\A,\psi}$ be a cell of the aforementioned CAD.  Let $v\in
  C$ be a valuation.  Then for any $v' \in C$ and for every guard
  $\fee$ appearing in $\A$, $v'\models \fee$ if, and only if,
  $v\models \fee$.
\end{lemma}
Hence we can write $C \models \fee$ whenever $v\models \fee$ and $v \in C$.

Moreover, for every update $x_i := P_i$ there is a polynomial
$x_i-P_i$ in $\Poly(\A)$, which has value $0$ if and only if $x_i =
P_i$; as a result:
\begin{lemma}\label{lem:updatesCAD}
  Let $C \in \Dc_{\A,\psi}^k$ be a cell of level $k$, $\underline{C}$ be the
  projection of $C$ onto $\R^{k-1}$ and $-\infty=f_0 < \cdots
  <f_{r+1}=+\infty$ be the semi-algebraic functions dividing
  $\underline{C}$ as in Definition~\ref{def:cad}.  Let $u$ be an
  update of the form $x_k := P$ for some polynomial $P \in
  \Q[x_1,\dots,x_{k-1}]$.  Then there exists an index $i \in\oneto{r}$
  such that, over $\underline{C}$, $f_i = P$.
\end{lemma}
As a corollary, there exists a unique cell $C' \in \Dc_{\A,\psi}^k$ such that for any
valuation $v\in C$, $v[u] \in C'$, namely $C' = \{(x,f_i(x)) \mid
x\in\underline{C}\}$, which can be written $C[u]$.
\medskip

Discrete transitions of $\A$ are translated as follows into $\Reg_{\A,\psi}$:
if $(q,\fee,a,u,q') \in \Delta$ and $C \models \fee$, there is a
transition $(q,C) \tr{a} (q',C[u])$.

\begin{proposition}[Correctness w.r.t. discrete steps]\label{prop:discreteSteps}
\begin{itemize}
\item If $(q,v) \tr{a} (q',v') \in \T_\A$, then $(q,C) \tr{a} (q',C')
  \in \Reg_\A$ with $v \in C$ and $v' \in C'$.
\item If $(q,C) \tr{a} (q',C') \in \Reg_\A$ then for all $v \in C$
  there exists $v' \in C'$ such that $(q,v) \tr{a} (q',v') \in \T_\A$.
\end{itemize}
\end{proposition}

\begin{proof}~
\begin{itemize}
\item First, $(q,v) \tr{a} (q',v') \in \T_\A$ implies that there is a
  transition $(q,\fee,a,u,q')$ such that $v\models \fee$ and
  $v'=v[u]$.  By Lemma~\ref{lem:guardsCAD}, we have that $C\models
  \fee$.  In addition, we have by Lemma~\ref{lem:updatesCAD} that
  $v'=v[u] \in C[u]$.  By the definition of $\Reg_{\A,\psi}$, there is a
  transition $(q,C) \tr{a} (q',C[u]) \in\Reg_{\A,\psi}$.
\item Transition $(q,C) \tr{a} (q',C') \in \Reg_{\A,\psi}$ only exists
  because of a transition $(q,\fee,a,u,q') \Delta$, and we have
  $C'=C[u]$.  Let $v \in C$.  Since $C \models \fee$, by
  Lemma~\ref{lem:guardsCAD} we have that $v\models \fee$.  Hence there
  is a transition $(q,v) \tr{a} (q',v[u]) \in \T_\A$.  By
  Lemma~\ref{lem:updatesCAD}, $v[u] \in C[u]$, which concludes the
  proof.  \qed
\end{itemize}
\end{proof}

\begin{example}
Part of this abstraction for deciding reachability in \polita $\A_0$ (\figurename~\ref{fig:exPolIta}, page~\pageref{fig:exPolIta}) is depicted on \figurename~\ref{fig:exPolItaRegions}.
In this figure, points are given by the triangular system representing them.
Computations of sample points for intervals between roots where omitted, and only appear in the graph as roots of derivatives.
Note that  having no $a$ edge from state $q_0, 1, (5,Int)$ is not an omission, but a consequence of the guard $x_1^2 \leq x_1 +1$ no longer being satisfied.
In this graph, $C_+$ is the polynomial obtained when replacing $X_2$ by $X_2-1$ in $C$.
Faded states and transitions are unreachable but are nonetheless constructed from the decomposition.
\end{example}

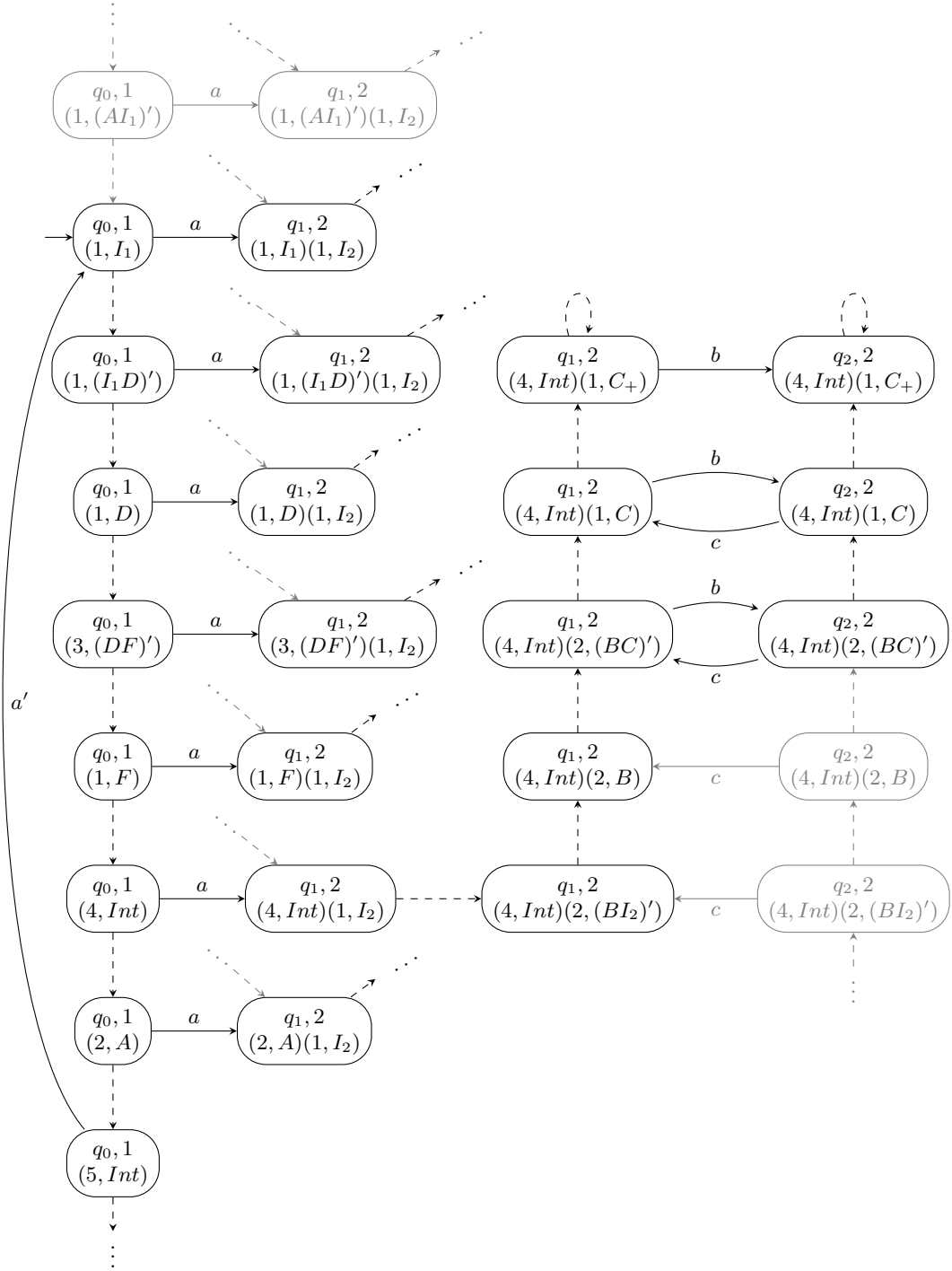
\begin{figure}
\centering
\begin{tikzpicture}[auto,node distance=0.95cm and 1.25cm]

\useasboundingbox (-0.5,3) rectangle (11.5,-15);

\tikzstyle{absstate}=[state,shape=rectangle,rounded corners=12pt]
\tikzstyle{unreach}=[opacity=0.5]
\tikzstyle{dots}=[node distance=0.5cm and 0.75cm,inner ysep=0pt,inner xsep=1pt]
\tikzstyle{ddots}=[node distance=0.25cm and 0.05cm,inner ysep=0pt,inner xsep=1pt]
\tikzstyle{time}=[->,dashed]
\tikzstyle{trans}=[->]

\node[absstate,initial] (q0zero) {$\begin{array}{c} q_0, 1 \\ (1,I_1) \end{array}$};
\node[absstate,right=of q0zero] (q1zero) {$\begin{array}{c} q_1, 2 \\ (1,I_1) (1,I_2)\end{array}$};

\node[absstate,unreach,above=of q0zero] (q0AIpr1) {$\begin{array}{c} q_0, 1 \\ (1,(AI_1)')\end{array}$};
\node[dots,unreach,above=of q0AIpr1] (forDotsAbove) {\raisebox{5pt}{$\vdots$}};
\node[absstate,unreach,right=of q0AIpr1] (q1AIpr1) {$\begin{array}{c} q_1, 2 \\ (1,(AI_1)') (1,I_2)\end{array}$};

\node[absstate,below=of q0zero] (q0IDpr1) {$\begin{array}{c} q_0, 1 \\ (1,(I_1D)')\end{array}$};
\node[absstate,right=of q0IDpr1] (q1IDpr1) {$\begin{array}{c} q_1, 2 \\ (1,(I_1D)') (1,I_2)\end{array}$};

\node[absstate,below=of q0IDpr1] (q0D1) {$\begin{array}{c} q_0, 1 \\ (1,D)\end{array}$};
\node[absstate,right=of q0D1] (q1D1) {$\begin{array}{c} q_1, 2 \\ (1,D) (1,I_2)\end{array}$};

\node[absstate,below=of q0D1] (q0DFpr3) {$\begin{array}{c} q_0, 1 \\ (3,(DF)')\end{array}$};
\node[absstate,right=of q0DFpr3] (q1DFpr3) {$\begin{array}{c} q_1, 2 \\ (3,(DF)') (1,I_2)\end{array}$};

\node[absstate,below=of q0DFpr3] (q0F1) {$\begin{array}{c} q_0, 1 \\ (1,F)\end{array}$};
\node[absstate,right=of q0F1] (q1F1) {$\begin{array}{c} q_1, 2 \\ (1,F) (1,I_2)\end{array}$};

\node[absstate,below=of q0F1] (q0FApr4) {$\begin{array}{c} q_0, 1 \\ (4,Int)\end{array}$};
\node[absstate,right=of q0FApr4] (q1FApr4zero) {$\begin{array}{c} q_1, 2 \\ (4,Int) (1,I_2)\end{array}$};
\node[ddots,unreach,above left=of q1FApr4zero] (q1FApr4zeroDotsAbove) {\raisebox{5pt}{$\ddots$}};
\path[time,unreach] (q1FApr4zeroDotsAbove) edge (q1FApr4zero);
\path[trans] (q0FApr4) edge node{$a$} (q1FApr4zero);
\node[absstate,right=of q1FApr4zero] (q1FApr4BIpr2) {$\begin{array}{c} q_1, 2 \\ (4,Int) (2,(BI_2)')\end{array}$};
\path[time] (q1FApr4zero) edge (q1FApr4BIpr2);
\node[absstate,above=of q1FApr4BIpr2] (q1FApr4B2) {$\begin{array}{c} q_1, 2 \\ (4,Int) (2,B)\end{array}$};
\path[time] (q1FApr4BIpr2) edge (q1FApr4B2);
\node[absstate,above=of q1FApr4B2] (q1FApr4BCpr2) {$\begin{array}{c} q_1, 2 \\ (4,Int) (2,(BC)')\end{array}$};
\path[time] (q1FApr4B2) edge (q1FApr4BCpr2);
\node[absstate,above=of q1FApr4BCpr2] (q1FApr4C1) {$\begin{array}{c} q_1, 2 \\ (4,Int) (1,C)\end{array}$}; 
\path[time] (q1FApr4BCpr2) edge (q1FApr4C1);
\node[absstate,above=of q1FApr4C1] (q1FApr4Cp1) {$\begin{array}{c} q_1, 2 \\ (4,Int) (1,C_+)\end{array}$}; 
\path[time] (q1FApr4C1) edge (q1FApr4Cp1);
\path[time] (q1FApr4Cp1) edge[loop above] (q1FApr4Cp1);

\node[absstate,right=of q1FApr4BCpr2] (q2FApr4BCpr2) {$\begin{array}{c} q_2, 2 \\ (4,Int) (2,(BC)')\end{array}$};
\node[absstate,above=of q2FApr4BCpr2] (q2FApr4C1) {$\begin{array}{c} q_2, 2 \\ (4,Int) (1,C)\end{array}$};
\path[time] (q2FApr4BCpr2) edge (q2FApr4C1);
\node[absstate,above=of q2FApr4C1] (q2FApr4Cp1) {$\begin{array}{c} q_2, 2 \\ (4,Int) (1,C_+)\end{array}$};
\path[time] (q2FApr4C1) edge (q2FApr4Cp1);
\path[time] (q2FApr4Cp1) edge[loop above] (q2FApr4Cp1);
\node[absstate,unreach,below=of q2FApr4BCpr2] (q2FApr4B2) {$\begin{array}{c} q_2, 2 \\ (4,Int) (2,B)\end{array}$};
\node[absstate,unreach,below=of q2FApr4B2] (q2FApr4BIpr2) {$\begin{array}{c} q_2, 2 \\ (4,Int) (2,(BI_2)')\end{array}$};
\node[dots,unreach,below=of q2FApr4BIpr2] (q2FApr4BIpr2DotsBelow) {\raisebox{5pt}{$\vdots$}};
\path[time,unreach] (q2FApr4B2) edge (q2FApr4BCpr2);
\path[time,unreach] (q2FApr4BIpr2) edge (q2FApr4B2);
\path[time,unreach] (q2FApr4BIpr2DotsBelow) edge (q2FApr4BIpr2);
\path[trans,unreach] (q2FApr4B2) edge node {$c$} (q1FApr4B2);
\path[trans,unreach] (q2FApr4BIpr2) edge node {$c$} (q1FApr4BIpr2);

\path[trans] (q1FApr4BCpr2) edge[bend left=15] node {$b$} (q2FApr4BCpr2);
\path[trans] (q2FApr4BCpr2) edge[bend left=15] node {$c$} (q1FApr4BCpr2);
\path[trans] (q1FApr4C1) edge[bend left=15] node {$b$} (q2FApr4C1);
\path[trans] (q2FApr4C1) edge[bend left=15] node {$c$} (q1FApr4C1);
\path[trans] (q1FApr4Cp1) edge  node {$b$} (q2FApr4Cp1);

\node[absstate,below=of q0FApr4] (q0A2) {$\begin{array}{c} q_0, 1 \\ (2,A)\end{array}$};
\node[absstate,right=of q0A2] (q1A2) {$\begin{array}{c} q_1, 2 \\ (2,A) (1,I_2)\end{array}$};

\node[absstate,below=of q0A2] (q0FApr5) {$\begin{array}{c} q_0, 1 \\ (5,Int)\end{array}$};
\node[dots,below=of q0FApr5] (forDotsBelow) {\raisebox{5pt}{$\vdots$}};
%
%
%
%
%
%

\path[time,unreach] (forDotsAbove) edge (q0AIpr1);
\path[time,unreach] (q0AIpr1) edge (q0zero);
\path[time] (q0zero) edge (q0IDpr1);
\path[time] (q0IDpr1) edge (q0D1);
\path[time] (q0D1) edge (q0DFpr3);
\path[time] (q0DFpr3) edge (q0F1);
\path[time] (q0F1) edge (q0FApr4);
\path[time] (q0FApr4) edge (q0A2);
\path[time] (q0A2) edge (q0FApr5);
\path[trans] (q0FApr5) edge[bend left=40,looseness=0.5,swap] node {$a'$} (q0zero);
\path[time] (q0FApr5) edge (forDotsBelow);

\foreach \sample in {zero,IDpr1,D1,DFpr3,F1,A2} {
\node[ddots,unreach,above left=of q1\sample] (q1\sample DotsAbove) {\raisebox{5pt}{$\ddots$}};
\node[ddots,above right=of q1\sample] (q1\sample DotsRight) {\rotatebox{70}{$\ddots$}};
\path[time] (q1\sample) edge (q1\sample DotsRight);
\path[time,unreach] (q1\sample DotsAbove) edge (q1\sample);
\path[trans] (q0\sample) edge node{$a$} (q1\sample);
}

\foreach \sample in {AIpr1} {
\node[ddots,unreach,above left=of q1\sample] (q1\sample DotsAbove) {\raisebox{5pt}{$\ddots$}};
\node[ddots,unreach,above right=of q1\sample] (q1\sample DotsRight) {\rotatebox{70}{$\ddots$}};
\path[time,unreach] (q1\sample) edge (q1\sample DotsRight);
\path[time,unreach] (q1\sample DotsAbove) edge (q1\sample);
\path[trans,unreach] (q0\sample) edge node{$a$} (q1\sample);
}

\end{tikzpicture}
\caption[Partial depiction of $\Reg_{\A_0}$.]{Partial depiction of $\Reg_{\A_0}$.\\ Dashed edges correspond to time successors $\timeSuc$; faded states are unreachable.}
\label{fig:exPolItaRegions}
\end{figure}

\subsubsection{Labeling with atomic propositions.}

Finally, we translate a comparison $P \rel 0$ in $\psi$ into a fresh atomic proposition $p_{P \rel 0}$ and label $\Reg_{\A,\psi}$ as follows.
Note that since $\Dc_{\A,\psi}$ is in particular adapted to $\Poly(\psi)$, every cell $C$ of $\Dc_{\A,\psi}$ is sign-invariant for $P$, hence the truth value of $P \rel 0$ is constant in $C$.
As a result, it makes sense to write $C\models P \rel 0$ whenever $v \models P\rel0$ for some $v\in C$, and proposition $p_{P \rel 0}$ is true in every state $(q,C)$ where $C \models P \rel 0$.
We write $\overline{\psi}$ the formula where each $P \rel 0$ has been replaced by $p_{P \rel 0}$.

\begin{proposition}\label{prop:formulaConservation}
$\A \models \psi$ if, and only if, $\Reg_{\A,\psi} \models \overline{\psi}$.
\end{proposition}

Note that $\overline{\psi}$ is a \CTL formula, which can be checked with the usual polynomial time labeling procedure.
Since the number of cells in a cylindrical decomposition is doubly exponential in the number of clocks and polynomial in the number and maximal degree of polynomials to which it is adapted~\cite{BPR}, we obtain the complexity stated in Theorem~\ref{thm:modelchecking}.

\subsection{On-the-fly algorithm}

Propositions~\ref{prop:timeSteps} and~\ref{lem:updatesCAD} provide decidability of the model checking problem, by the algorithm that builds the finite graph $\Reg_{\A,\psi}$ verifies that $\overline\psi$ is satisfied in this graph.

However, building the complete graph is not efficient in practice, since it requires to build the set of all cells beforehand.
In the sequel, we show an on-the-fly algorithm that builds only the reachable part of $\Reg_{\A,\psi}$.
This algorithm would not, for example, build the faded states of $\Reg_{\A_0}$ in \figurename~\ref{fig:exPolItaRegions}.

The key to the on-the-fly algorithm is to store only the part of the tree corresponding to the current sample point and its time successors.
This construction is akin to what is done in \figurename~\ref{fig:linePartitioningR2}, where only the line partitioning for $X_2$ above the current sample point is computed by the lifting phase, while line partitioning above, for, say, sample point $(1,F)$ is not computed.
As a result, we do not keep the whole tree but only part of it.

We show that this information is sufficient to compute the successors through time elapsing and transition firing.
Nonetheless, remark that although this pruning yields better performances in practice, the computational complexity in the worst case is not improved: the line partitioning at the first level already requires doubly exponential time, since the elimination phase is required.

\begin{definition}[Pruned tree]
Let $\{\P_k\}_{k\leq n}$ be the polynomials obtained by the elimination phase.
The \emph{pruned tree} for sample point $(\alpha_1,\dots,\alpha_k)$ is the sequence of completed line partitionings for sample points $\{(\alpha_1,\dots,\alpha_i)\}_{1\leq i\leq k}$.
By convention, the pruned tree for the empty sample point ($k=0$) is the line partitioning at level $1$.
\end{definition}

Given a clock valuation $(v_1,\dots,v_k,0,\dots,0)$ at level $k$, it can be represented by a sample point $(\alpha_1,\dots,\alpha_k)$, or, equivalently, by a pruned tree for sample point $(\alpha_1,\dots,\alpha_{k-1})$ and the index $m$ of $\alpha_k$ in the line partitioning for $(\alpha_1,\dots,\alpha_{k-1})$.
In this representation, computing the time successors of $(\alpha_1,\dots,\alpha_k)$ is simply done by incrementing $m$ (if it is not the maximal index in the line partitioning).
Note that in this algorithm we do not loop on the rightmost cell; although it is convenient to assume in $\Reg_{\A}$ that a time successor always exists, it has no effect regarding the reachability problem.

The set of enabled discrete transitions can be generated by computing the signs (see Algorithm~\ref{algo-sign-realization} page~\pageref{algo-sign-realization}) of polynomials appearing in guards.
When a discrete transition $q\xrightarrow{g,a,u}q'$ is chosen, several cases should be distinguished with respect to the level of states $q$ and $q'$.
\begin{itemize}
\item If the level decreases, \emph{i.e.} $\lambda(q')<\lambda(q)$.
Then the pruned tree corresponding to the new configuration is only the topmost-part of height $\lambda(q')$ of the original pruned tree.
Otherwise said, we ``forget'' line partitionings for levels above $\lambda(q')$; however, the partitionings can be kept in memory in order not to have to recompute them later.
The new index is the index of $\alpha_{\lambda(q')}$ in the partitioned line for this level.
\item If the level doesn't change, \emph{i.e.} $\lambda(q')=\lambda(q)=k$.
The only way to change the clock values is through an update $x_k := P$ with $P \in \Q[X_1,\dots,X_{k-1}]$.
Then the polynomial of degree $1$ $R=X_k -P$ was added to $\Poly(\A)$ and its unique root $\alpha_k'$ appears in the line partitioning of level $\ell$.
Note that in the triangular system representing $(\alpha_1,\dots,\alpha_k')$ it may appear as $\dots(1,R)$ or some other equivalent value, hence to determine the index in the partitioned line the algorithm must actually determine the sign of $R$ for all sample points of the line until $0$ is found.
\item If the level increases, \emph{i.e.} $\lambda(q')>\lambda(q)$.
First there can be an update of $x_k$, hence the same computations as above must be performed in order to find the new sample point corresponding to the valuation of clocks up to $\lambda(q)$.
Then the pruned tree of height $\lambda(q')$ has to be computed.
This is done by $\lambda(q')-\lambda(q)$ lifting steps (Algorithm~\ref{algo-lifting} page~\ref{algo-lifting}).
Since all clocks remain null for levels above $\lambda(q)$, the sample points given as input\footnote{Although the actual input of the algorithm are triangular systems, assuming we have the system $\T$ for $(\alpha_1,\dots,\alpha_{\lambda(q)})$, the subsequent triangular systems are $\T \cup (1,X_{\lambda(q)+1})\dots$.} are $(\alpha_1,\dots,\alpha_{\lambda(q)},0,\dots,0)$.
\end{itemize}
\bigskip

Now the on-the-fly algorithm works as follows:
\begin{itemize}
\item Compute sets of polynomials $\{\P_i\}_{i\leq n}$ by the elimination phase.
\item Compute the completed line partitioning at level $1$.
\item Start at a the initial state.
If the level of the initial state is $k>1$, proceed with $k-1$ lifting phases as in the case of level increase.
Add this state in a queue.
\item Until the queue is empty:
\begin{itemize}
\item Compute the list of fireable discrete transitions and whether time successor is allowed.
\item Add all new successors through a  fireable discrete transition or a time step to the queue.
\end{itemize}
\item Apply the model checking algorithm on this graph.
\end{itemize}

\paragraph*{A note on efficient memory usage}
As noted above, a line partitioning only needs to be computed once.
In addition --~and this also holds for the complete construction of $\Reg_{\A,\psi}$~--,  the \emph{triangular} structure of triangular systems enables a sharing of line partitioning at lower levels.
Thus the size of the graph in memory is at most the size of the complete tree of the decomposition added, and not multiplied, by the number of states of the \polita.

\section{Expressiveness and extensions}
\label{sec:extensions}

We finally focus on expressiveness of \polita. After comparing this
class with stopwatch automata, we show how to extend it while keeping
decidable the above verification problems.  For sake of clarity, in
section~\ref{sec:definition} we have presented a basic model of
\polita. Here we show how to add three features consisting in: (1)
including parameters in the expressions of guards and updates, (2)
associating with each level a subset of auxiliary clocks, and (3)
allowing to update clocks of lower levels than the current one. Since
in the context of \ita, the first two extensions have already been
studied in~\cite{BHJL-rp13} and the third one in~\cite{BHS-fmsd2012},
our presentation will not be fully formalized.

\subsection{\polita \emph{vs} Stopwatch automata}

By syntax inclusion, \polita are at least as expressive as \ita.  As a
direct consequence, there exists a timed language accepted by a
\polita that is not accepted by a TA~\cite{BH-Fossacs09}.

There exists a timed language accepted by a timed automaton that is
not accepted by any \polita as presented above (the proof is a direct
adaptation from the one proving said language is not accepted by an
\ita~\cite{BHS-fmsd2012}), although it is accepted by the extension
with auxiliary clocks provided below (Section~\ref{sec:auxClocks}).

\smallskip The class of stopwatch automata (\swa), which also
syntactically contains the class of \ita, is however incomparable to
\polita.

\begin{proposition}\label{prop:politaNOTswa}
  There exists a timed language accepted by a \polita with a single
  clock that cannot be accepted by a stopwatch automaton.
\end{proposition}

The proof of the above proposition relies on a lemma about runs accepted by a \swa.
Recall that in a stopwatch automaton, each clock can be active or inactive in every state.
Also recall that updates are restricted to resets\footnote{It is possible to simulate affectations to rational constants, but it does not change expressiveness of the model.} $x:=0$ and guards are comparisons to a rational constant\footnote{Again, diagonal constraints $x-y\rel c$ for $c\in\Q$ can be simulated.}.
In the remainder of the section, we use $+_{q}$ to denote addition only on stopwatches active in $q$.

\begin{lemma}\label{lem:irratDelaysSwa}
Let $\rho= (q_0,v_0) \xrightarrow{\delta_0} (q_0,v_0+_{q_0} \delta_0) \xrightarrow{g_0,a_0,u_0} (q_1,v_1) \cdots$ be a run in a stopwatch automaton.
Then there exists $\rho'= (q_0,v_0) \xrightarrow{\delta_0'} (q_0,v_0+_{q_0} \delta_0') \xrightarrow{g,a_1,u} (q_1,v_1) \cdots$ taking the same discrete transitions as $\rho$ such that $\forall i, \delta_i \in \Q$.
\end{lemma}

\begin{proof}
We assume that stopwatches are never reset throughout the run.
This can be done since one can assume that a reset stopwatch is actually a fresh one.
Consider the linear system with a variable $\delta_i$ per delay and rational coefficients which corresponds to all guards appearing after $q_k$.
We write
\[\gamma_i^x=\left\{\begin{array}{l}
1 \textrm{ if } x \textrm{ is active in } q_i \\
0 \textrm{ otherwise}
\end{array}\right.\]
For each stopwatch $x$, we add the constraints
\[\bigwedge_{i=0}^{|\rho|} \left(\sum_{\ell=0}^{i} \gamma_i^x \cdot \delta_i\right) \models g_i\]
Note that since guards have rational coefficients, this system has rational coefficients.
In addition since $\rho$ is an accepted run, this system has a solution $(\delta_0,\dots)$.
Also note that for every solution $(\delta_i')_i$, replacing each delay $\delta_i$ with $\delta_i'$ in $\rho$ still yields a valid run $\rho'$, since all guards are still respected.
The set of solutions of a linear system with rational coefficient is a rational polyhedron, so the projection over each variable yields an interval with rational endpoints (or $-\infty$ or $+\infty$).
If for some $i$, $\delta_i$ is irrational, the interval cannot be reduced to a point, so it contains an open set around $\delta_i$, in which there is a rational $\delta_i'$.
Therefore, there exists a solution $(\delta_i')_i \in \Q^{|\rho|}$ and $\rho'$ is a run with rational delays.
\qed
\end{proof}

\begin{proof}[Proposition~\ref{prop:politaNOTswa}]
Consider \polita of \figurename~\ref{fig:politaNotStopwatch}, which accepts the timed language $\Lang$ containing the single word $(a,1)(b,\sqrt2)$.
Assume $\Lang$ is accepted by a stopwatch automaton $\A_\Lang$.
Let $\rho= (q_0,v_0) \xrightarrow{\delta_0} (q_0,v_0+_{q_0} \delta_0) \xrightarrow{g,a_1,u} (q_1,v_1) \cdots$ be a run accepting $(a,1)(b,\sqrt2)$.
Note that some $a_i$s may actually be $\varepsilon$.
Since $b$ occurs at an irrational instant, there is at least an irrational delay before the occurrence of $b$.
By Lemma~\ref{lem:irratDelaysSwa}, $\rho'$ the run where all delays are rational is also accepted.
Therefore the instant of $b$ in $\rho'$ is rational and cannot be $\sqrt2$.
Furthermore any time rescaling for $\Lang$ does not change this result since either $a$ or $b$ is taken at an irrational instant.
\qed
\end{proof}

\begin{figure}
\centering
\begin{tikzpicture}[auto]
\node[state,initial] (q0) at (0,0) {};
\node[state] (q1) at (3,0) {};
\node[state,accepting,accepting where=right] (q2) at (6,0) {};

\path[->] (q0) edge node{$x_1=1$, $a$} (q1);
\path[->] (q1) edge node{$x_1^2=2$, $b$} (q2);
\end{tikzpicture}
\caption{A \polita whose timed language is not accepted by a stopwatch automaton.}
\label{fig:politaNotStopwatch}
\end{figure}
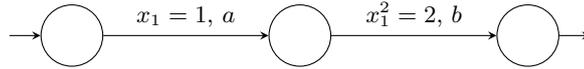

On the other hand, the (untimed) language of a \polita (and the extensions of Section~\ref{sec:extensions}) is regular, as shown by the construction of a finite abstraction of $\T_\A$ in Section~\ref{sec:reach}.
It is not necessarily the case of (untimed) languages of stopwatch automata~\cite{cassez00,alur95}, hence there are some timed languages accepted by a SWA that are not accepted by any \polita.

\subsection{Parameters}

Getting a complete knowledge of a
system is often impossible, especially when integrating quantitative
constraints. Moreover, even if these constraints are known, when the
execution of the system slightly deviates from the expected behavior,
due to implementation choices, previously established properties may
not hold anymore. Additionally, considering a wide range of values for
constants allows for a more flexible and robust design.  
Introducing parameters instead of concrete values is an elegant way of
addressing these three issues. Parametrization however makes
verification more difficult. For instance, in timed automata, allowing
a single clock to be compared to parameters leads to undecidability 
of the reachability problem~\cite{miller00}.

\smallskip
Suppose that we enlarge  \polita allowing expressions to be polynomials
whose set of variables is the union
of a set of clocks $\{x_1,\ldots,x_n\}$ and a set of parameters $\{p_1,\ldots,p_k\}$. 
Then we consider the cylindrical decomposition where the order of variables
is $p_1,\ldots,p_k,x_1,\ldots,x_n$. Now assume that the relevant values of parameters
are specified by a first-order formula {\it val}. Then using the cylindrical
decomposition, we can answer reachability questions like
``for all $p_1 \cdots p_k$ satisfying  {\it val}, is $q$ reachable?''
or safety questions like ``for all $p_1 \cdots p_k$ satisfying  {\it val}, is $q$ unreachable?''.

\subsection{Auxiliary clocks}
\label{sec:auxClocks}

With each level $i$, one may associate a set of \emph{auxiliary} clocks $Y_i$ in addition
to the \emph{main} clock $x_i$. Since there are multiple clocks for some level $i$,
in this \polita, with every state of level $i$, is associated an \emph{active} clock
among $X_i=\{x_i\} \cup Y_i$, specifying which clock evolves with time in this state.
Auxiliary clocks may be used in a restrictive setting
w.r.t. the main clocks to influence the behavior of the \polita. 
Let us detail these restrictions:
\begin{itemize}
 \item In a guard of a transition outgoing from a state at level $i$,
 among auxiliary clocks only those of the level $i$ may occur
 and they are only be compared between them or with the main
 clock (i.e. $z\bowtie z'$ with $z,z' \in X_i$);
 \item In a transition outgoing from state at level $i$,
 an auxiliary clock of level $i$ may be updated by another clock
 of level $i$ (i.e. $y:= z$ with $y\in Y_i$ and $z \in X_i$)
 while the main clock may be updated by an auxiliary clock
 only if the destination state of the transition is also
 at level $i$ (i.e. $x_i:= y$ with $y\in Y_i$).
\end{itemize}
The decision procedure works as follows. The cylindrical
decomposition does not take into account the auxiliary clocks.
However the definition of a class specifies in which interval 
of level $i$ lies any
clock of level $i$ and their relative position
for clocks inside the same interval. 

Adding auxiliary clocks strictly extends expressiveness
of \polita w.r.t. timed languages. It was shown
in~\cite{BHS-fmsd2012} that the language
\begin{eqnarray*}
L=\big\{ (a,t_1)(b,t_2) &\ldots& (a,t_{2p+1})(b,t_{2p+2}) 
\mid p \in \N,\\
&&\forall 0\leq i\leq p,\ t_{2i+1}=i+1 \mbox{ and }  i+1<t_{2i+2}<i+2, \\
&&\forall 1\leq i\leq p, \ t_{2i+2} - t_{2i+1}< t_{2i} - t_{2i-1} \big\}
\end{eqnarray*}
is not a language of an \ita. The proof also holds for
\polita since it is only based on the following hypotheses:
(1) there is a single clock per level, (2) at level $i$, the behavior
is only determined by the current state and the values of clocks at levels less
or equal than $i$, and (3) the clock $x_i$ is null at level $j<i$.

The untimed language of $L$ is $(ab)^+$. In the accepted timed words, there is an
occurrence of $a$ at each time unit and the successive occurrences
of $b$ come each time closer to the next occurrence of $a$ than
previously. Consider the \polita of Figure~\ref{fig:expressiveness} with a
single level and single final state $q_2$. The main clock $x$ is
active in all states and $y$ is an auxiliary clock.
It is routine to check that the timed language of this automaton
is $L$.

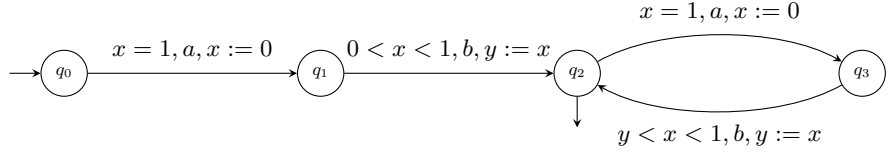
\begin{figure}[ht]
\centering
\begin{tikzpicture}[node distance=4.5cm,auto]
\useasboundingbox (-0.75,-1) rectangle (10.875,1);
\tikzstyle{every state}+=[scale=0.75]
\node[state,initial] (q0) at (0,0) {$q_0$};
\node[state] (q1) [right of=q0] {$q_1$};
\node[state,accepting,accepting where=below] (q2) [right of=q1] {$q_2$};
\node[state] (q3) [node distance=5cm,right of=q2] {$q_3$};

\path[->] (q0) edge node {\timedtrans{}{}{x=1, a, x:=0}} (q1);
\path[->] (q1) edge node {\timedtrans{}{}{0<x<1, b, y:=x}} (q2);
\path[->] (q2) edge [bend left,looseness=0.75] node {\timedtrans{}{}{x=1, a, x:=0}} (q3);
\path[->] (q3) edge [bend left,looseness=0.75] node 
{\timedtrans{y<x<1,b,y:=x}{}{}} (q2);
\end{tikzpicture}
\caption{A \polita with a single level and an auxiliary clock}
\label{fig:expressiveness}
\end{figure}

\subsection{Allowing more updates}

At level $i$, the value of a clock of level $j<i$ is relevant. So it is interesting
to allow updates of such a clock. Again for keeping decidability, such updates
have the following restrictions:
\begin{itemize}
 \item At level $i$, the main clock of level $j<i$ can only be updated
 by a polynomial of the main clocks of level less than $j$: 
 $x_j:=P(x_1,\ldots,x_{j-1})$;
 \item At level $i$, an auxiliary clock of level $j<i$ 
 may be updated by a clock of level $j$: $y:= z$ with $y\in Y_j$ and $z \in X_j$.
\end{itemize}

The decision procedure for this extension consists 
in translating the extended \polita in a \polita
with the same behavior by at level $i$: (1) delaying the update of clocks of
level $j<i$
that should have been done until the current level becomes $j$
and (2) duplicating the states by memorizing the current
value of such a clock as an expression of the values
of the clock when the level $j$ was left. Guards and updates
outgoing from a duplicated state are modified to take into account 
these expressions.

\smallskip
Let us illustrate this transformation on the \polita of Figure~\ref{fig:exitatoitamoins}
that is transformed in the \polita of Figure~\ref{fig:exitamoinsfromita}.
The original clock has only main clocks and the level of the state is
indicated inside the state. In the transformed state the superscript
'+' means that this corresponds to a state of of the original ITA ready
to be simulated while the superscript '-' indicates that the delayed
updates have to be performed. Let us start with the transition
outgoing the state $q_0$, the update of $x_1$ is delayed
but memorized in the state `$q_2^+,x_1:=2$'. The transition
outgoing from this state corresponds to the transition outgoing
from $q_2$ but in the guard the occurrence of $x_1$ has been substituted
by $2$. With this transformation, the update becomes $x_2:=5$ but since we are
at level 3, this update is memorized in state `$q_3^+,x_1:=2,x_2:=5$'. 
The transition from $q_3$ at level 3 to $q_5$ at level 2 is split in two transitions
in the simulating \polita. First we enter state `$q_5^-,x_1:=2,x_2:=5$' at
level 2 where the active clock is an auxiliary clock of level 2, $y_2$.
Then in null time due to the guard we perform the delayed update
of $x_2$, still memorizing the update of $x_1$ and enter the state 
`$q_5^+,x_1:=2$'.

\begin{figure}[ht]
\centering
\begin{tikzpicture}[auto]
\node[state,initial] (q0) at (-0.25,-0.25) {$q_0,2$};
\node[state,initial] (q1) at (-0.25,-1.75) {$q_1,2$};
\node[state] (q2) at (2.5,-1) {$q_2,3$};
\node[state] (q3) at (6.5,-1) {$q_3,3$};
\node[state,accepting] (q4) at (9.25,-0.25) {$q_4,3$};
\node[state,accepting] (q5) at (9.25,-1.75) {$q_5,2$};

\path[->] (q0) edge node [pos=0.25] {$x_1:=2$} (q2);
\path[->] (q1) edge (q2);
\path[->] (q2) edge node {\timedtransnoreset{2x_2+x_1 > 3 \wedge x_3 < 2}{x_2:=2x_1+1}} (q3);
\path[->] (q3) edge node [pos=0.75] {\timedtransnoreset{x_2:=x_1+1}{x_3:= 2x_2}} (q4);
\path[->] (q3) edge node [swap,pos=0.75] {$x_1:=1$} (q5);

\end{tikzpicture}
\caption{A \polita containing extended updates of clocks}
\label{fig:exitatoitamoins}
\end{figure}
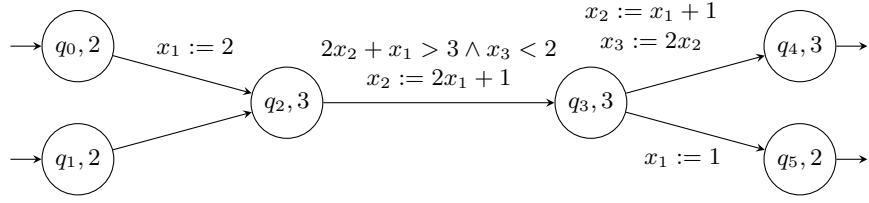

\begin{figure}[h]
\centering
\begin{tikzpicture}[node distance=3.125cm,auto]
\tikzstyle{big state}=[state,shape=rounded rectangle,inner xsep=0pt]

\node[state,initial] (q0) at (0,0) {$q_0^+,2$};
\node[big state,node distance=1.75cm,right of=q0] (q2cond) 
{$\begin{array}{c}q_2^+,3\\x_1:=2\end{array}$};
\node[big state,right of=q2cond] (q3cond1)
{$\begin{array}{c}q_3^+,3\\x_1:=2\\x_2:=5\end{array}$};
\node[big state,accepting,node distance=4.25cm,right of=q3cond1] (q4cond1) 
{$\begin{array}{c}q_4^+,3\\x_1:=2\\x_2:=3\end{array}$};

\node[big state,accepting,below of=q4cond1] (q5p) 
{$\begin{array}{c}q_5^+,2\\x_1:=1\end{array}$};
\node[big state] (q5m1) at ($(q3cond1 |- q5p) + (0,1)$) 
{$\begin{array}{c}q_5^-,2\\x_1:=1\\x_2:=5\end{array}$};
\node[big state] (q5m2) at ($(q3cond1 |- q5p) - (0,1)$) 
{$\begin{array}{c}q_5^-,2\\x_1:=1\\x_2:=2x_1+1\end{array}$};

\node[big state] (q3cond2) at ($(q5m2) + (0,-2)$) {$\begin{array}{c}q_3^+,3\\x_2:=2x_1+1\end{array}$};
\node[state,left of=q3cond2] (q2) {$q_2^+,3$};
\node[state,initial,node distance=1.75cm,left of=q2] (q1) {$q_1^+,2$};
\node[big state,accepting,node distance=4.25cm,right of=q3cond2] (q4cond2) {$\begin{array}{c}q_4^+,3\\x_2:=x_1+1\end{array}$};

\path[->] (q0) edge (q2cond);
\path[->] (q2cond) edge node {\timedtransnoreset{2x_2+2>3}{\wedge\, x_3<2}} (q3cond1);
\path[->] (q3cond1) edge node {$x_3:=10$} (q4cond1);
\path[->] (q3cond1) edge (q5m1);
\path[->] (q5m1) edge node [pos=0.375] {$y_2=0,\varepsilon,x_2:=5$} (q5p);

\path[->] (q1) edge (q2);
\path[->] (q2) edgenode {\timedtransnoreset{2x_2+x_1>3}{\wedge\, x_3<2}} (q3cond2);
\path[->] (q3cond2) edge node {$x_3:=4x_1+2$} (q4cond2);
\path[->] (q3cond2) edge (q5m2);
\path[->] (q5m2) edge node [swap,pos=0.375] {$y_2=0,\varepsilon,x_2:=x_1+1$} (q5p);

\end{tikzpicture}
\caption{A \polita equivalent to the \polita of Figure~\ref{fig:exitatoitamoins}}
\label{fig:exitamoinsfromita}
\end{figure}
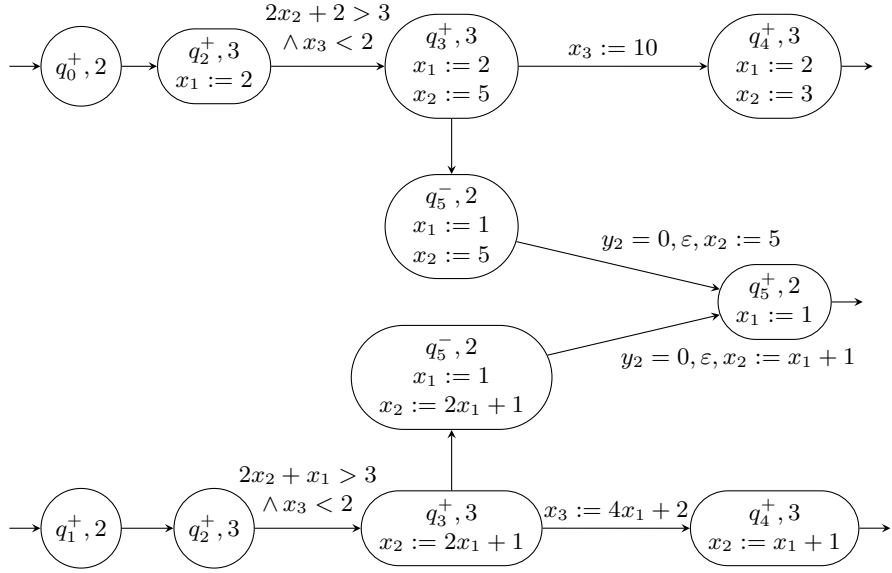
\section{Conclusion}
\label{sec:conc}

We extend Interrupt Timed Automata with polynomial
expressions on clocks, and prove that reachability and model checking
of some timed temporal logic are decidable using the cylindrical
decomposition.  We also show that an on-the-fly construction of a
class automaton is possible during the lifting phase of this
decomposition. We establish that \polita and \swa are incomparable and
provide some additional interesting features to the model.  In order
to experiment the practical complexity of the decision procedures, an
implementation is in progress. Since the current construction still
requires the full complexity of the cylindrical decomposition, we plan
for future work to investigate if recent methods~\cite{DinS03,HongD12}
with a lower complexity could be used to achieve reachability,
possibly for a restricted version of \polita.

\bibliographystyle{splncs03}
\bibliography{PolyITA}

\end{document}